\documentclass{nws}

\usepackage{amsmath, amsfonts, amssymb, graphicx, float, upgreek}
\usepackage{amsbsy,amsthm}
\usepackage{mathrsfs,multirow}
\usepackage{bm,comment}
\usepackage{array,booktabs}
\usepackage{verbatim}
\usepackage{caption,subcaption}
\usepackage{algorithm,algorithmic}
\usepackage{url,color}
\usepackage{latexsym}
\usepackage{dblfloatfix}
\usepackage{tikz}
\usetikzlibrary{matrix}

\newcommand{\secref}[1]{Section~\ref{#1}}
\newcommand{\figref}[1]{Figure~\ref{#1}}
\newcommand{\tabref}[1]{Table~\ref{#1}}
\newtheorem{theorem}{Theorem}
\newtheorem{lemma}{Lemma}

\newtheorem{proposition}{Proposition}
\newtheorem{corollary}{Corollary}
\newcommand{\remove}[1]{}

\def\vv{\boldsymbol{v}}
\def\uu{\boldsymbol{u}}
\def\DD{\boldsymbol{D}}
\def\BB{\boldsymbol{B}}
\def\II{\boldsymbol{I}}
\def\LL{\boldsymbol{\mathcal{L}}}
\def\TT{\boldsymbol{T}}
\def\RR{\boldsymbol{R}}
\def\VV{\boldsymbol{V}}

\def\WW{\boldsymbol{W}}
\def\AA{\boldsymbol{A}}

\def\ppi{\boldsymbol{\pi}}
\def\ttheta{\boldsymbol{\theta}}

\def\KK{\boldsymbol{K}}
\def\vol{{\rm vol}}
\def\ttheta{\boldsymbol{\theta}}
\def\KK{\boldsymbol{K}}
\def\cut{{\rm cut}}
\def\bvec#1{{\mbox{\boldmath $#1$}}}
\def\sbvec#1{{\mbox{\scriptsize \boldmath $#1$}}}
\def\calL{\boldsymbol{\mathcal{L}}}

\def\expec#1#2{\mbox{\bf E}_{#1}\left[ #2 \right]}
\usepackage{color}

\title[Generalized Laplacian]{The Interplay Between Dynamics and Networks: Centrality, Communities, and Cheeger Inequality}

\author[Ghosh, et al.]
{RUMI GHOSH\\Robert Bosch, LLC.\\
 KRISTINA LERMAN\\University of Southern California, Information Sciences Institute\\
 SHANG-HUA TENG\\University of Southern California, Computer Science\\
 and XIAORAN YAN\footnote{Authors are ordered alphabetically.}\\University of Southern California, Information Sciences Institute
}

\begin{document}
\maketitle

\begin{abstract}
We study the interplay between a \emph{dynamic process} and the \emph{structure} of the network on which it unfolds. Specifically, we examine the impact of this interaction on the quality-measure of network clusters and centrality. This enables us to effectively identify important vertices participating in the dynamics and communities in the network.  We introduce a mathematical framework for defining and characterizing an ensemble of dynamic processes on a network that generalizes the traditional {\em Laplacian framework}. 
For each dynamic process in our framework, we define a generalized centrality measures that captures a vertex's participation in the dynamic process on a given network and also define a function that measures the quality of every subset of vertices as a potential cluster (or community) with respect to this process. We show that the subset-quality function generalizes the traditional {\em conductance measure} for graph partitioning. We partially justify our choice of the quality function by showing that the classic Cheeger's inequality, which relates the conductance of the best cluster in a  network with a spectral quantity of its Laplacian matrix, can be extended to the generalized Laplacian. The generalized Laplacian framework brings under the same umbrella a surprising variety of dynamical processes and allows us to systematically compare the different perspectives they create on network structure.

\end{abstract}

\section{Introduction}
As flexible representations of complex systems, networks model entities and relations between them as vertex and links. In a social network for example, vertices are people, and the links between them represent friendships. As another example, world wide web is a collection of web pages with hyperlinks between them. An unprecedented amount of such relational data is now available. While discovery and fortune await, the challenge is to extract useful information from these large and complex data.

\emph{Centrality} and \emph{community} detection have emerged as two fundamental tasks in network analysis. The goal of centrality identification is to find the important vertex that controls the processes taking place on the network. Page Rank~\cite{Page99thepagerank} was such a measure developed by Google to rank web pages.  Other centrality measures, such as degree centrality, Katz score and eigenvector centrality~\cite{katz1953, Bonacich01,ghosh_rethinking_2012}, are used in communication networks for effective routing of information. Methods used to maximize influence~\cite{Kempe03} or limit the spread of a disease, depend on identifying central vertices.

The objective of community detection is to discover subsets of well-interacting vertices in a given network. Discovering such communities allows us to follow the classic reductionist approach, dividing the vertices into categories, each of which can then be analyzed separately. For example, US-based political networks usually exhibit a bipolar structure, representing the democrat/republican divisions~\cite{adamic2005political}. Communities within on-line social networks like Facebook might correspond to real social groups which can be targeted with various advertisements. However, just like with the different notions of centrality, there are an assortment of community detection algorithms, each leading to a different community structure on the same network (see~\cite{Fortunato10,porter} for reviews).

With so many choices for both centrality and community detection, practitioners often face a difficult decision of which measures to use. In this paper, instead of looking for the ``best" centrality or community measures, we propose an umbrella framework that unifies some of the well known measures under a single model, connecting the ideas of centrality, communities and dynamic processes on networks.
In this dynamics-oriented view, a vertex's centrality describes its  participation in the dynamic process taking place on the network~\cite{Borgatti05,lambiotte_flow_2011}. Similarly, communities are groups of vertices that interact more frequently with each other (according to the rules of the dynamic process) than with vertices from other communities~\cite{Lerman12pre}. In fact, this view of modeling is not new: when choosing \emph{conductance} as a measure of community quality, one implicitly assumes that \emph{unbiased random walk} is taking place on the network~\cite{KananVampelaVetta,SpielmanTengLinear,Chung1997Spectral,delvenne_stability_2008}. Under the random walk assumption, \emph{heat kernel page rank}~\cite{Chung07pnas} could be a measure of centrality.
Other dynamic processes, such as the spread of information, ideas, or epidemics, arise from different interactions than the unbiased random walk. An epidemic is a stochastic process that, unlike a random walk, attempts to transition to (i.e., infect) every neighbor of a vertex. Epidemic dynamics may be specified by the \emph{replicator} operator~\cite{Lerman12pre}, whose stationary distribution defines \emph{eigenvector centrality}~\cite{Bonacich01,Ghosh11physrev}. It is natural, then, that centrality of a vertex depends on the specifics of the dynamic process, which together with network topology influences its activity level. As a result, vertices that are visited most frequently by a random walk (specified by the heat kernel page rank) are different from the vertices that are infected most often during an epidemic (specified by eigenvector centrality).


In this paper, we study the interplay between a dynamic process and the underlying network on which it unfolds. We focus on the impact of this interaction on the emergence of central vertices and the formation of communities in the network, and on the design of efficient algorithms for their identification. Our paper makes the following contributions.

{ \bf{Generalized Laplacian:}} We present an umbrella framework for describing dynamic processes on a network that generalizes the traditional Laplacian framework for diffusion and random walks. Recall that a random walk on a network is a stochastic dynamic process that transitions from a vertex to a random neighbor of that vertex. It defines a Markov chain that can be specified by the normalized Laplacian of the network. Our framework  attempts to capture a class of dynamic processes that evolve in time according to the rules that generalize the normalized Laplacian, which allows arbitrary bias and delay during the process (Section~\ref{sec:interplay-dynamics}).

{\bf{Formal analysis of interaction dynamics: }} Our framework defines a class of dynamic processes with relatively simple parameterized transformations of the normalized Laplacian matrix, which enables rigorous analysis of the impact of these parameters on the measures of network centrality and communities. Its inclusion of diffusion and random walks allows us to build upon the insights from previous work, including stationary distribution as centrality measures, conductance as community-quality measures. With the generalized Laplacian framework, we are able to extend these ideas to existing and new processes whose properties may offer useful insights with their corresponding centrality and community measures. 

{\bf{Generalized centrality:}} Based on the connection between centrality measures and the stationary distribution of a random walk~\cite{Page99thepagerank, Chung07pnas}, we generalize the definition of centrality to all dynamic processes that are modeled by the generalized Laplacian. Some well known centrality measures are identified as special cases under this unified definition, which allows us to systematically compare them using linear transformations. In particular, we show that seemingly different formulations of dynamics are in fact the same after a change of basis. Generalized centrality also leads to the discovery of a conservation law and generalized volume measures for all dynamical processes in our framework (Section~\ref{sec:centrality}).

{\bf{Generalized conductance:}} We extend conductance to a general quality-measure for communities that reflects the dynamics on the network. For all dynamical processes under the framework, the corresponding generalized conductance is defined in terms of its parameterizations. This quantity measures the quality of every subset (of vertices) as a potential community with respect to this process on the given network. Recall that, the {\em conductance} balances between minimizing the cross-community interactions and the volume of each community. Generalized conductance is defined in the exact same formulation but with the generalized notions of interaction as well as volume. As with centrality, some existing community measures turn out to be special cases. The generalized Laplacian framework enables systematical comparison among them and new community measures, as they are now unified and connected by linear transformations (Section~\ref{sec:community}).

{\bf{Efficient algorithms:}}
Recall that the classical Cheeger's inequality relates a spectral quantity of the normalized Laplacian to conductance. We show that the same relation holds for the generalized conductance and the extreme eigenvalue of the corresponding generalized Laplacian operator. By proving this generalized versions of Cheeger's inequality and other related theorems, we are able to adapt existing spectral algorithms (~\cite{SpielmanTengLinear,AndersenChungLang,AndersenPeres}) for detecting provably good communities to the generalized Laplacian framework (Section~\ref{sec:algorithms}).

{\bf{Empirical evaluation on real-world networks:}}
We apply our formalism to study the structure of several real-world networks. We contrast the central vertices and communities identified by different dynamic processes and provide an intuitive explanation for these differences (Section~\ref{sec:experiments}).

While the generalized Laplacian framework described in this paper cannot model every dynamic process of interest, it is still flexible enough to include a surprising variety of dynamical processes which are seemingly unrelated. It allows us to systematically study and compare these processes under a unified framework. We hope this study will help lead to better approaches for defining and understanding the general interaction between dynamics and networks.


\section{Background and Related Work}
\label{sec:background}
Before introducing our framework, we  briefly review some closely related existing models. We will later show that these models are captured by our framework. The intuitions from these well-known systems can be used to understand our framework.

We represent a network as a weighted, undirected graph $G = (V,E,\AA)$ with $n$ vertices, where for  $i,j\in V$,  $a_{ij}$ assigns an non-negative weight (affinity) to each edge $\in E$. We follow the tradition that $a_{ij} = 0$ if and only if $(i,j)\not\in E$;
i.e., $\AA$ is the weighted adjacency matrix. By  convention we assume it is symmetric and $a_{ii}=0$ for all $i\in V$.
In the discussion below, the {\em (weighted) degree} of vertex $i \in V$ is defined as the total weight of edges incident on it, that is, $d_i= \sum_j a_{ij}$.
A dynamic process describes a state variable $\theta_i(t)$ associated with each vertex $i$. This variable changes its value based on interactions with the vertex's neighbors according to the rules of the dynamic process.

In this paper, since we view dynamics as operators on the vector composed of vertex state variables, we adopt the linear algebra convention, i.e., using column vertex state vectors and left-multiply them by linear operators. We summarize the terms and notation in the glossary table:
\begin{table}
\caption{Glossary of terms and notations.}
\begin{minipage}{\textwidth}\begin{center}
\setlength{\tabcolsep}{.3em}
\bgroup\def\arraystretch{1.5}
	\begin{tabular}{clcl}
	\hline\noalign{\vspace {-.5cm}}\hline
	Term & Description & Term & Description\\
	\hline
	$\AA$ &Weighted adjacency matrix		&$a_{ij}$ & Entry $i,j$ of $\AA$ \\
	$\WW$ &Interaction matrix	&$w_{ij}$ & Entry $i,j$ of $\WW$\\
	$\ttheta(t)$ &Vertex state vector at time $t$	&$\theta_i(t)$ & Entry $i$ of $\ttheta(0)$\\
	$\DD_{\AA}$ &Diagonal degree matrix of $\AA$	&$d_i$ & Degree of vertex $i$ in $\AA$\\	
	$\DD_{\WW}$ &Diagonal degree matrix of $\WW$	&${d_{\WW}}_i$ & Degree of vertex $i$ in $\WW$	\\
	$\TT$ &Diagonal delay matrix	&$\tau_i$ & Delay factor of vertex $i$\\
	$\LL$ &Generalized Laplacian Operator	&$P_{ij}$ & Random walk probability from $j$ to $i$\\
	$\vv_{\AA}$ & Dominant eigenvector of $\AA$ &${\vv_{\AA}}_i$ & Entry $i$ of $\vv_{\AA}$ 	\\
	$\VV_{\AA}$ & Diagonal matrix with $\vv_{\AA}$ entries &$\vv_i$ & $i$th eigenvector of $\LL$	\\
	$c_i$ &Centrality of vertex $i$	&$S$ & Subset of $V$, defines a community \\
	\hline\noalign{\vspace {-.5cm}}\hline
	\end{tabular}\egroup
\end{center}\end{minipage}
\label{tab:glossary}
\end{table}

\subsection{Unbiased Random Walks}
One of the best known of dynamic processes on networks is the random walk. The simplest is the discrete time \emph{unbiased random walk} (URW), where a walker located at the vertex $i$ follows one of the edges with a probability proportional to the weight of the edge \cite{lambiotte_flow_2011}. In this case, the state vector $\theta$ forms a distribution, whose expected value follows the following update equation
\begin{equation*}
 \theta_i (t+1) = \sum_j P_{ij}\theta_j (t).
\end{equation*}
Here $P$ is a stochastic matrix whose entry $P_{ij}$ is the transition probability for a walker to go from the vertex $j$ to $i$, $P_{ij} = a_{ij}/d_j$.

The update equation of an unbiased random walk leads to the difference equation
\begin{equation*}
 \Delta\theta_i  = \theta_i (t+1) - \theta_i (t) = \sum_j P_{ij}\theta_j(t) - \theta_i(t) =  - \sum_j L^{RW}_{ij}\theta_j(t),
\end{equation*}
where $L^{RW}$ is the normalized \emph{random walk Laplacian matrix} with $L^{RW} =  \II- \AA\DD_{\AA}^{-1}$.

To go from a discrete time synchronous random walk to a continuous time dynamics, we introduce a waiting time function for the asynchronous jumps performed by the walk \cite{lambiotte_flow_2011}. Assuming a simple Poisson process where the waiting times between jumps are exponentially distributed as the PDF $f(t,\tau) = \frac{1}{\tau_i} e^{- \frac{t}{\tau_i}}$, we can rewrite the above difference equations as differential equations,
\begin{equation*}
\frac{d{\theta_i}}{dt}=- \sum_j \frac{L^{RW}_{ij}}{\tau_j}\theta_j \;.
\end{equation*}

The solution to the above differential equations gives the state vector of the random walk at any time $t$:
\begin{equation*}
\ttheta(t) =\ttheta(0)\cdot e^{ -L^{RW} \TT^{-1} t} \;,
\end{equation*}
where $\TT$ is the $n\times n$ diagonal matrix with the mean waiting time $\tau_i$ as entries.
If the dynamic process converges, then regardless of its initial value $\ttheta(0)$, the stationary distribution
$\pi_i$ has the following density:
\begin{equation}
\label{eq:equilibrium_rw}
\pi_i = lim_{t \rightarrow \infty}\theta_i(t) = \frac{d_i\tau_i}{\sum_j d_j\tau_j}\;.
\end{equation}
Intuitively, the stationary distribution is proportional to the product of vertex degree and the mean waiting time.

\subsection{Biased Random Walks}
A natural extension to the simple random walks is to allow biases towards certain destinations, making it a \emph{biased random walk} (BRW). According to \cite{lambiotte_flow_2011}, any biased random walk defined with the transition probability $P_{ij} \propto b_i a_{ij}$ (where $b_i$ is the bias towards vertex $i$) can be reduced to a URW on a re-weighted ``interaction network" with the adjacency matrix
\begin{equation*}
\WW = \BB \AA \BB\;,
\end{equation*}
where $\BB$ is a diagonal matrix with $\BB_{ii}=b_i$.
The above symmetric re-weighting ensures that
$$P_{ij} = \frac{b_i a_{ij} b_j}{\sum_i b_i a_{ij} b_j} \propto b_i a_{ij},\quad\quad
P_{ji} =\frac{b_j a_{ji} b_i}{\sum_j b_j a_{ji} b_i} \propto b_j a_{ji}\;. $$

Previously studied in network communications \cite{ling_effects_2013,fronczak_biased_2009,gomez-gardenes_entropy_2008}, one class of BRWs is where the bias $b_i$ has a power-law dependence on degree: $P_{ij} \propto d_i^{\beta} a_{ij}$. The exponent $\beta$ controls the amount of bias. The URW is recovered with $\beta = 0$; When $\beta >0$, biases toward high degree vertices are introduced, and when $\beta <0$, the random walk is more likely to jump to a lower degree neighbor.

The stationary distribution for this class of BRWs in general is
\begin{equation*}
 \pi_i = \dfrac{\sum_i d_i^{\beta} a_{ij} d_j^{\beta}}{\sum_{ij} d_i^{\beta} a_{ij} d_j^{\beta}} \;.
\end{equation*}

Another BRW is the maximum-entropy random walk \cite{burda_localization_2009,lambiotte_flow_2011}, defined as
\begin{equation*}
 \theta_i (t+1) = \sum_j \dfrac{{v_{\AA}}_ia_{ij}}{\lambda_{max}{v_{\AA}}_j}\theta_j(t) \;,
\end{equation*}
where $v_{\AA}$ is the eigenvector of $\AA$ associated with its largest eigenvalue $\lambda_{max}$: $\AA v_{\AA}=\lambda_{max}v_{\AA}$.
Again, an unbiased random walk on the interaction network $\WW = \VV_{\AA}\AA\VV_{\AA}$ is equivalent to biased random walk on the original network $\AA$ (the entries of diagonal matrix $\VV $ is the components of the eigenvector $\boldsymbol{V}$). In particular, the stationary distributions of each can be written as
\begin{equation*}
 \pi_i = \dfrac{{v_{\AA}}_i^2}{\sum_i {v_{\AA}}_i^2} \;.
\end{equation*}

\subsection{Consensus and Opinion Dynamics}
Another closely related class of discrete time dynamic processes is the so-called the ``consensus process'' \cite{lambiotte_flow_2011, olfati-saber_consensus_2007, krause_compromise_2008}. Consensus process models coordination across a network where each vertex updates its ``belief" based on the average ``beliefs" of its neighbors. Unlike random walks, which conserves total state value throughout the network (since the state vector is always a distribution), the consensus process follows the following update equation
\begin{equation*}
 \theta_i (t+1) = \dfrac{1}{d_i}\sum_j a_{ij}\theta_j(t)\;.
\end{equation*}
This leads to the difference equation
\begin{equation*}
 \Delta\theta_i  = \theta_i ^{t+1} - \theta_i ^{t} = - \sum_j L^{CON}_{ij}\theta_j(t)
\end{equation*}
where $L^{CON}$ is the \emph{consensus Laplacian matrix} with $L^{CON} = \II - \DD_{\AA}^{-1}\AA$. For an undirected graph with a symmetric $\AA$, $L^{CON} = [L^{RW}]^T$.

Consensus can also be turned into asynchronous continuous time dynamics. Again, assuming a Poisson process where the update interval at each vertex is exponentially distributed as $\uptau_i(t) = \frac{1}{\tau_i} e^{- \frac{t}{\tau_i}}$, we can rewrite the above difference equations as differential equations,
\begin{equation*}
\frac{d{\theta_i}}{dt}=- \sum_j \dfrac{L^{CON}_{ij}}{\tau_i}\theta_j \;.
\end{equation*}

The solution to the above differential equations gives the dynamic states:
\begin{equation*}
\ttheta(t) =\ttheta_{0}\cdot e^{ -\TT^{-1} L^{CON} t}
\end{equation*}

\noindent
The consensus process always converge to a uniform ``belief" state with the value,
\begin{equation}
\label{eq:equilibrium_con}
\pi_i = \dfrac{1}{\sum_j d_j \tau_j}\sum_i \theta_i(0) d_i \tau_i \;.
\end{equation}

Just like the URW, unbiased consensus can also be generalized by introducing a weight when averaging over neighbors' values. Similarly, any biased consensus defined with the update weights $\AA^{\prime}_{ij} \propto b_i a_{ij}$ (where $b_i$ is the bias towards vertex $i$) can be reduced to a unbiased consensus on a re-weighted ``interaction network"
\begin{equation*}
 w_{ij} = b_i a_{ij} b_j \;.
\end{equation*}

This opens the door to consensus dynamics such as opinion dynamics \cite{krause_compromise_2008}, and linearized approach to synchronization of different variants of the Kuramoto model \cite{Lerman12pre,Motter2005Sync,Arenas2006Sync}.

\subsection{Communities and Conductance}
In network clustering and community detection, previous work has focused on identifying subsets of vertices $S\subseteq V$ that interacts more frequently to vertices in the same community than to vertices in other subsets \cite{Fortunato10,porter}. A standard approach to clustering involves defining an objective function that measures the \emph{quality} of a cluster.
For a subset $S\subseteq V$, let $\bar{S}= V\setminus S$ to denote the complement of $S$, which consists of vertices that are not in $S$. Let $\cut(S,\bar{S})=\sum_{i \in S, j \in \bar{S}} a_{i,j}$ denote   the total interaction strength of all edges used by  $S$ to connect with the outside world. Let $\vol(S)=\sum_{i\in} d_i = \sum_{i\in S,j\in V} a_{i,j}$  denote the volume of weighted "importance" for all vertices in $S$.

One popular heuristic to measure the quality of a subset $S$ as a  potential good cluster (or a community) \cite{KananVampelaVetta,SpielmanTengLinear,Chung1997Spectral} is to use the ratio of these two quantities:
\begin{equation}
\label{eq:conductance}
   \phi(S)  =  \frac{\cut(S,\bar{S})}{{\min(\vol(S),\vol(\bar{S}))}}
\end{equation}
\noindent
For example, a subset that (approximately) minimizes this quantity --- the {\em conductance} of $S$ ---  is a desirable cluster, as it maximizes the fraction of  affinities within the subset.
If interactions among vertices are proportional to their affinity weights, then a set with small conductance also means that its members interact significantly more with each other than with members not in the subset.
Other well-known quality functions are normalized cut \cite{ShiMalik00} and ratio-cut, given by
$$
\frac{\cut(S,\bar{S})}{\vol(S)} + \frac{\cut(S,\bar{S})}{\vol(\bar S)}
\quad \mbox{and} \quad
\frac{\cut(S,\bar{S})}{{\min(|S|,|\bar S|)}},
$$
respectively. The smallest achievable such ratio is known as the {\em isoperimetric number}.

Algorithmically, once a quality function is selected, one can then perform a partitioning-based algorithm or mathematical programming-based method to find a cluster or clusters that optimizes the conductance. The optimization, however, is usually a combinatorial problem. To address this problem on large networks, various efficient approximate solutions have been developed, such as Spielman-Teng \cite{SpielmanTengLinear}, Andersen-Chung-Lang \cite{AndersenChungLang}, and Andersen-Peres \cite{AndersenPeres}.

While most community detection algorithms does not explicitly model the dynamic process that defines the interactions between vertices, the connection between conductance and unbiased random walks is quite well studied \cite{KananVampelaVetta,SpielmanTengLinear,Chung1997Spectral}. In particular, Chung's work on heat kernel page rank and Cheeger inequality, where a dynamical system is built using the normalized Laplacian, provides a theoretical framework for provably good approximations to the isoperimetric number \cite{Chung07pnas}. Intuitively, the relationship between clustering and dynamics can be captured as: a community is a cluster of vertices that ``trap'' a random walk for a long period of time before it jumps to other communities \cite{LovaszS,ShiMalik00,Rosvall08,SpielmanTengLinear}. Therefore, the presence of a good cluster implies that it will take a random walk a long time to reach its stationary distribution.

\section{Generalized Laplacian Framework}
\label{sec:interplay-dynamics}

Consider a linear dynamic process of the following form:
\begin{equation}
\label{eq:dynamics}
\frac{d{\ttheta}}{dt}=- \LL\ttheta,
\end{equation}
where $\ttheta$ is a column vector of size $n$ containing the values of the dynamic variable for all vertices, and $\LL$ is a positive semi-definite matrix, the \emph{spreading operator}, which defines the dynamic process.

As discussed in the introduction, we focus on dynamic processes that generalize the traditional normalized Laplacian for diffusion and random walks. Recall that the {\em symmetric normalized Laplacian matrix} of a weighted graph $G=(V,E,\AA)$ is defined as
$$\DD_{\AA}^{-1/2}(\DD_{\AA} - \AA)\DD_{\AA}^{-1/2},$$
where $\DD_{\AA}$ is the diagonal matrix defined by $(d_1,...,d_n)$. We study the properties of a dynamic process that can be further parameterized as:
\begin{equation}
\label{eq:spreading-operator}
\LL<\rho,\TT,\WW> =(\TT\DD_{\WW})^{-1/2-\rho}(\DD_{\WW}-\WW)(\DD_{\WW}\TT)^{-1/2+\rho}.
\end{equation}
We name this operator with the parameters $<\rho,\TT,\WW>$ \emph{generalized Laplacian}, and we shall represent it using $\LL$ in the rest of the paper. Here $\TT$  is the $n\times n$ diagonal matrix of \emph{vertex delay factors}. Its $i$th element $\tau_{ii}$ represents the average delay of vertex $i$. We assume that the operator is {\em properly scaled}: specifically, $\tau_{ii} = \tau_i \geq 1$, for all $i\in V$.
Another generalization from the traditional Laplacian is the use of the \emph{interaction matrix} $\WW$ instead of the adjacency matrix $\AA$. In theory, $\WW$ can be any $n\times n$ symmetric positive-definite matrix; however, we restrict our attention to scaling transformations of the adjacency matrix $\AA$. Note that the degree matrix $\DD_{\WW}$ is now also defined in terms of the interaction matrix, that is ${d_{\WW}}_i= \sum_j w_{i,j}$. While the $\rho$ parameter can technically be any real number, in this work we limit ourselves to three special cases: $\rho = 1/2, 0, -1/2$. These cases correspond to three equivalent linear operators with ``consensus", ``symmetric" and ``random walk" interpretations respectively.

We show that by transforming the generalized Laplacian in different ways we can express a number of different dynamic processes. We focus on the three simplest cases: (a) the ``similarity transformation'', which corresponds to the parameter $\rho$ in parameters in Eq.~\eqref{eq:spreading-operator}, (b) the ``scaling transformation'', which correspond to the parameter $\TT$, and (c) the ``reweighing transformation", which corresponds to the parameter $\WW$.

\subsection{Similarity Transformations}
\label{sec:similarity-transformations}
Changing $\rho$ in Equation~\eqref{eq:spreading-operator} leads to different representations of the same linear operator, unifying seemingly unrelated dynamics, such as ``consensus'' and ``random walk''. To see this, we refer to the idea of matrix similarity.

In linear algebra, similarity is an equivalence relation on the space of square matrices.  Two $n\times n$ matrices $X$ and $Y$ are similar if
\begin{equation}
    X = Q Y Q^{-1}\;,
\end{equation}
where the invertible $n\times n$ matrix $Q$ is called the change of basis matrix.
Similar matrices share many key properties, including:
\begin{itemize}
  \item Matrix rank
  \item Determinant
  \item Eigenvalues, and their multiplicities
  \item Eigenvectors are transforms of each other under the change of basis matrix $Q$
\end{itemize}
Matrix similarity provides a direct intuition about a new operator $B$ if we already understand $C$. Both represent the same linear operator up to a change of basis.

Recall that under our  framework, the symmetric version of the generalized Laplacian matrix is
\begin{align*}
\LL^{SYM}=\TT^{-1/2}\DD_{\WW}^{-1/2}(\DD_{\WW}-\WW)\DD_{\WW}^{-1/2}\TT^{-1/2}\;.
\end{align*}
We can rewrite the operator describing random walk dynamics as:
\begin{equation}
\LL^{RW} = (\DD_{\WW}-\WW) (\DD_{\WW}\TT)^{-1} = (\DD_{\WW}\TT)^{1/2} \LL^{SYM} (\DD_{\WW}\TT)^{-1/2}
\end{equation}
Thus, continuous time random walk with delay factors $\TT$ is similar to the symmetric normalized Laplacian.
Similarly, we can rewrite the continuous time consensus dynamics under our framework as
\begin{equation}
\LL^{CON} = (\DD_{\WW}\TT)^{-1} (\DD_{\WW}-\WW)  = (\DD_{\WW}\TT)^{-1/2} \LL^{SYM} (\DD_{\WW}\TT)^{1/2} = {\LL^{RW}}^T
\end{equation}
The fact that``consensus", ``symmetric" and ``random walk" operators are similar means that they model the same dynamics on a network, provided that we observe them in a consistent basis.

The random walk Laplacian matrix provides a physical intuition for our framework. An unbiased random walk on the interaction graph $\WW$ is equivalent to a biased random walk on the original adjacency matrix $\AA$ \cite{lambiotte_flow_2011}.  $\tau_i$ specifies the mean delay time of the random walk on vertex $i$ before a transition, assuming a simple Poisson process. This interpretation naturally extends to the other orthogonal parameters: namely $\WW$ controls the distribution of walk trajectories and $\TT$ controls the delay time of vertex transitions along each trajectory.

While we use symmetric operators for mathematical convenience in definitions and proofs and abuse the notation $\LL = \LL^{SYM}$, it is often more intuitive to think from the random walk or consensus perspective. In the following subsections, we will use the random walk formulation ($\rho=-1/2$) as examples, but all results apply to arbitrary $\rho$ values under a simple change of basis. More discussion about the similarity transformation follows after we introduce a few properties of the generalized Laplacian.

\subsection{Scaling transformations}
\label{sec:scaling-transformations}
Next, we investigate the effect of changing the time delay matrix $\TT$ while holding the other parameters fixed.

\paragraph{Uniform scaling}
One of the simplest transformations is uniform scaling, which is given by the diagonal matrix $\TT$ with identical entries:
\begin{equation}
    X = Y Q = \gamma Y  \;,
\end{equation}
where the scalar matrix $Q$ can be rewritten as $\gamma \II$, where $\gamma$ is a scalar. Uniform scaling preserves almost all matrix properties, including the eigenvalue and eigenvector pairs associated with the operator.

\begin{figure}
    \includegraphics[width=0.32\textwidth]{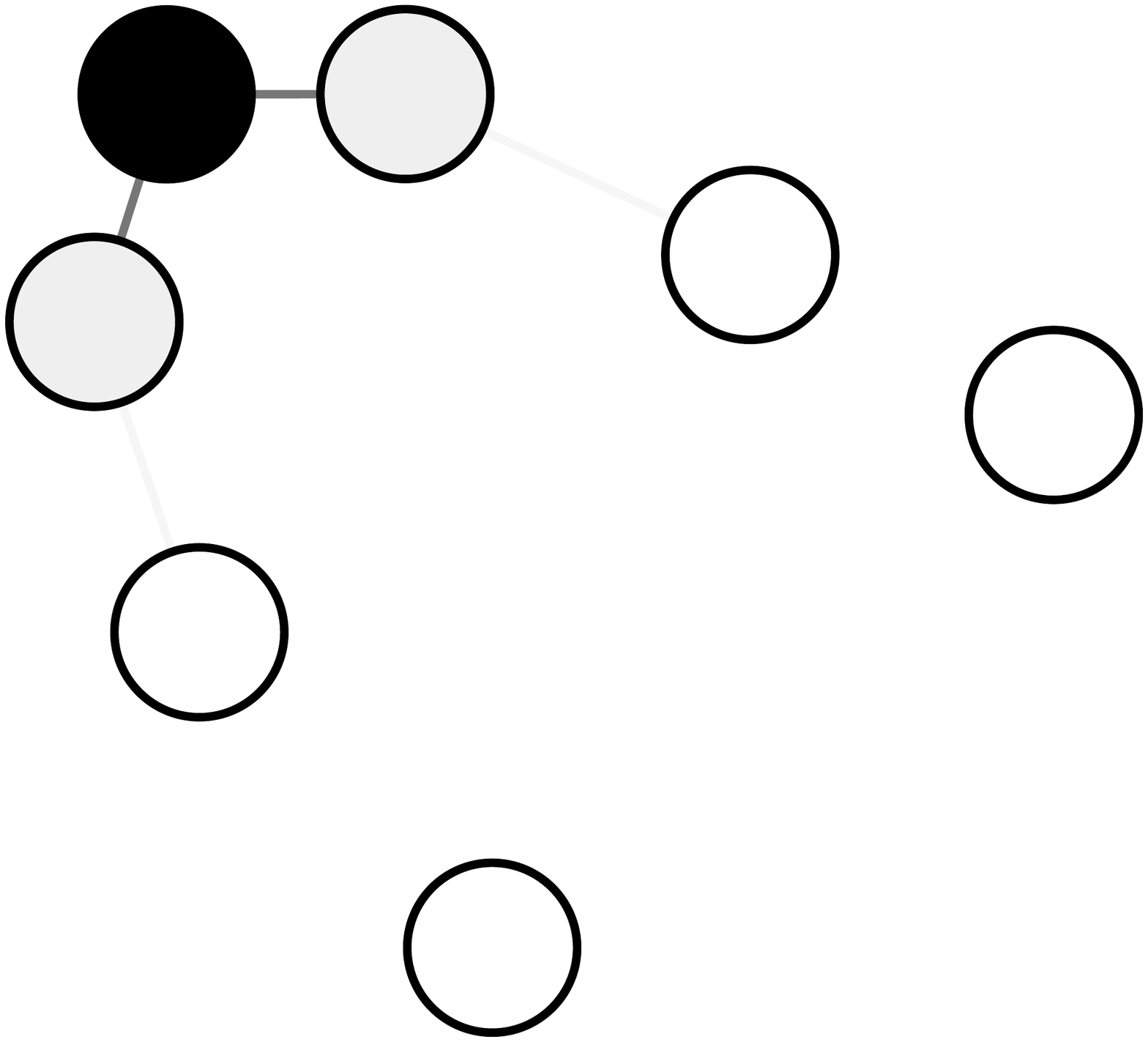}
    \includegraphics[width=0.32\textwidth]{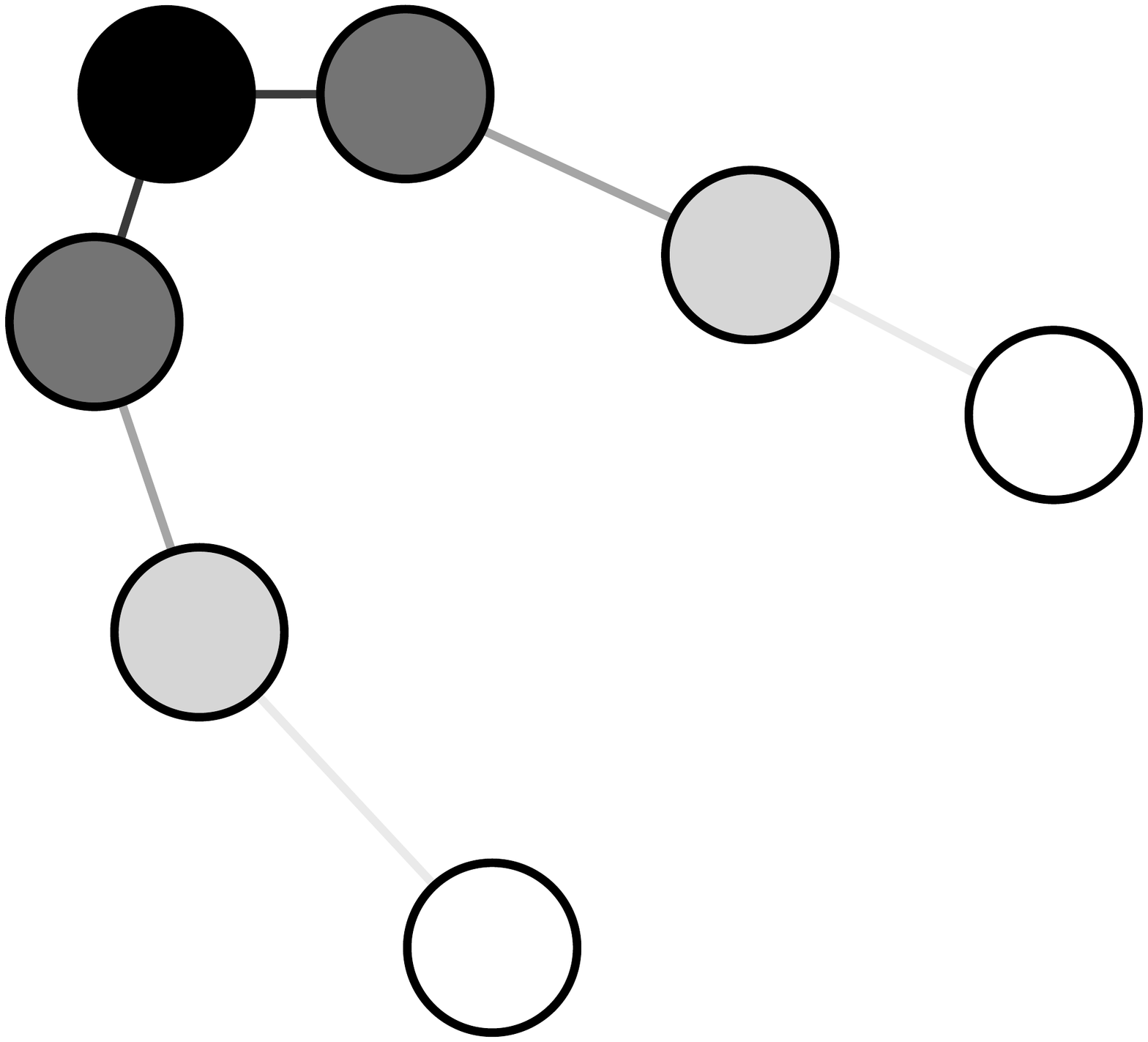}
    \includegraphics[width=0.32\textwidth]{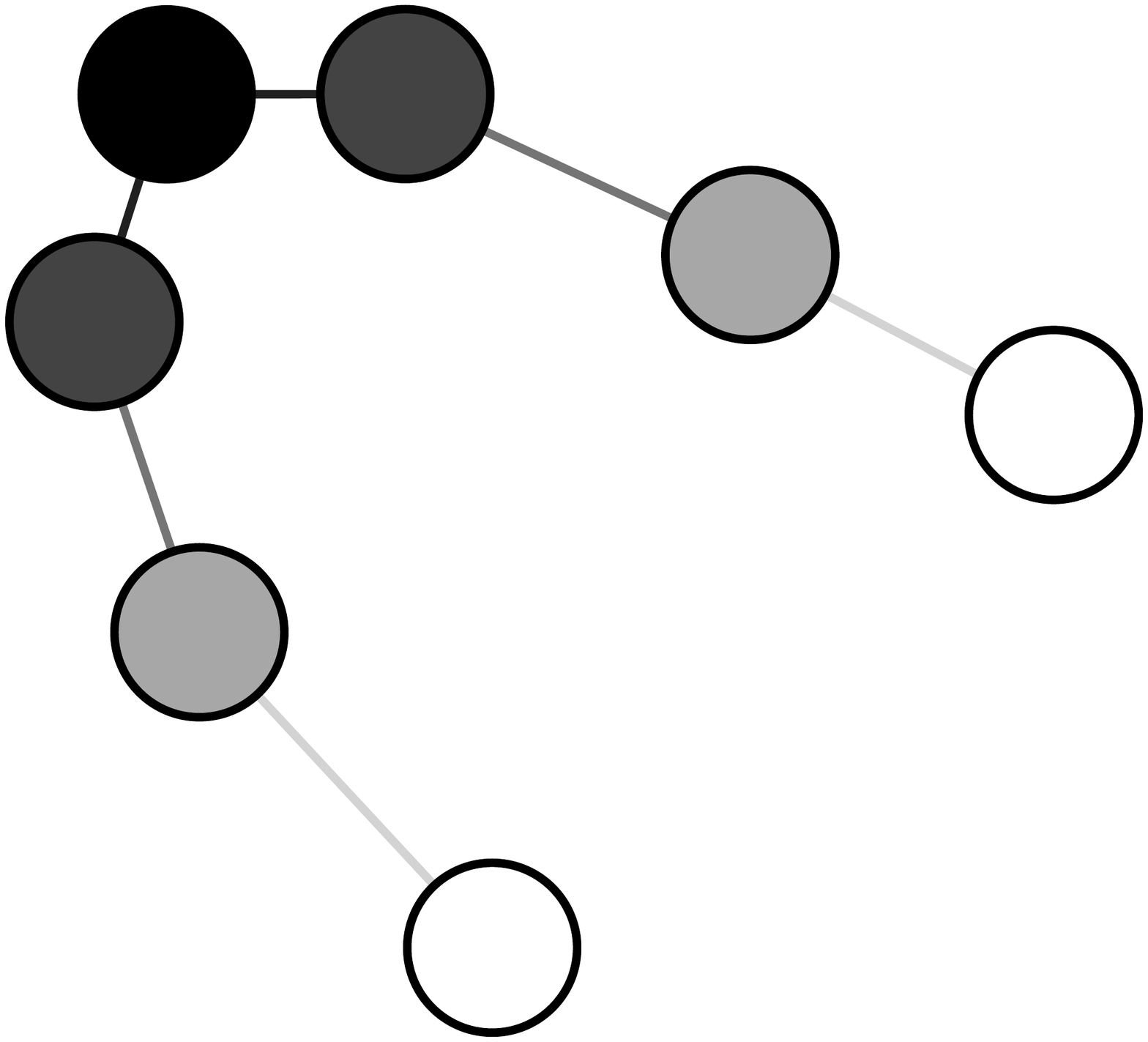}
  \caption{Random walk dynamics on a line of vertices under different uniform scalings $\TT=\gamma \II$. Random walk starts with the state vector concentrated on a single vertex. Vertex shading indicates the relative frequency it is visited by the random walk before convergence. Large value of $\gamma$ (left) leads to slower spreading of the probability density than smaller $\gamma$ (right).
}
  \label{fig:UniformScale}
\end{figure}

Intuitively, uniform scaling can be understood as rescaling time by $1/\gamma$. Under this scaling ($\TT = \gamma \II$), the solution of random walk dynamics becomes:
\begin{eqnarray*}
\ttheta(t) & = & \ttheta(0)\cdot e^{(\DD_{\WW}-\WW) (\DD_{\WW}\TT)^{-1} t} \\
& = & \ttheta(0)\cdot e^{(\DD_{\WW}-\WW) (\DD_{\WW})^{-1} \frac{t}{\gamma}}
\end{eqnarray*}

\figref{fig:UniformScale} illustrates how uniform scaling  affects the state vector of the random walk. It shows a simple network composed of a line of vertices, with the state vector of the random walk initially concentrated at the middle vertex. As the dynamic process evolves in time, the random walk visits neighboring vertices. The shading of a vertex corresponds to the visit frequency of the random walk (darker is higher). The figure shows uniform scaling of the form $\TT=\gamma \II$ with larger $\gamma$ to the left, and smaller $\gamma$ to the right. The larger the value of $\gamma$, the slower the process evolves. In the left-most figure, the state vector has non-negligible density only for the immediate neighbors of the middle vertex where the random walk starts. For smaller values of $\gamma$, the random walk has a chance to reach farther vertices in the same period of time.
In other words, a bigger ``time delay'' slows down the random walk.

Uniform scaling is a useful transformation that enables the generalized Laplacian to include arbitrary time delay factors $\TT^{\prime}$. The trick is to rescale $\TT$ to meet the condition  $\tau_i \geq 1$ by making $\TT=\frac{\TT^{\prime}}{\max_i{\tau_i}}$ without affecting any other matrix properties, as we will later see from some special operators under the framework. 

\paragraph{Non-uniform scaling}
The non-uniform scaling enables us to use the $\TT$ parameter to control the time delay at each vertex. Non-uniform scaling is written as
\begin{equation}
    X = Y Q \;,
\end{equation}
but the diagonal matrix $Q$ can have different entries. Unlike uniform scaling, this scaling does not preserve many matrix properties. We will discuss this in more detail after introducing a few properties of the generalized Laplacian.

\begin{figure}
    \includegraphics[width=0.32\textwidth]{figures/tau1}
    \includegraphics[width=0.32\textwidth]{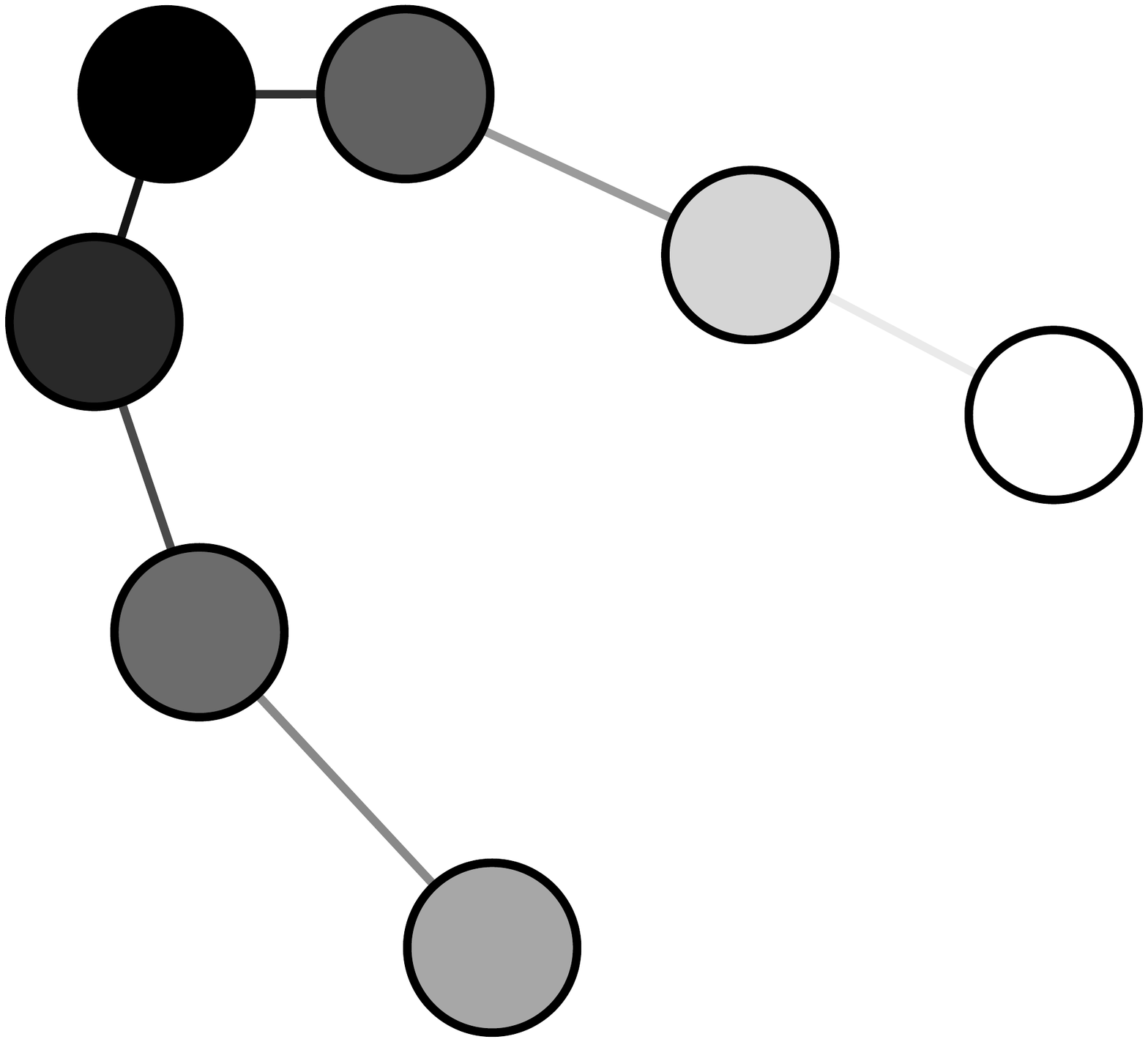}
    \includegraphics[width=0.32\textwidth]{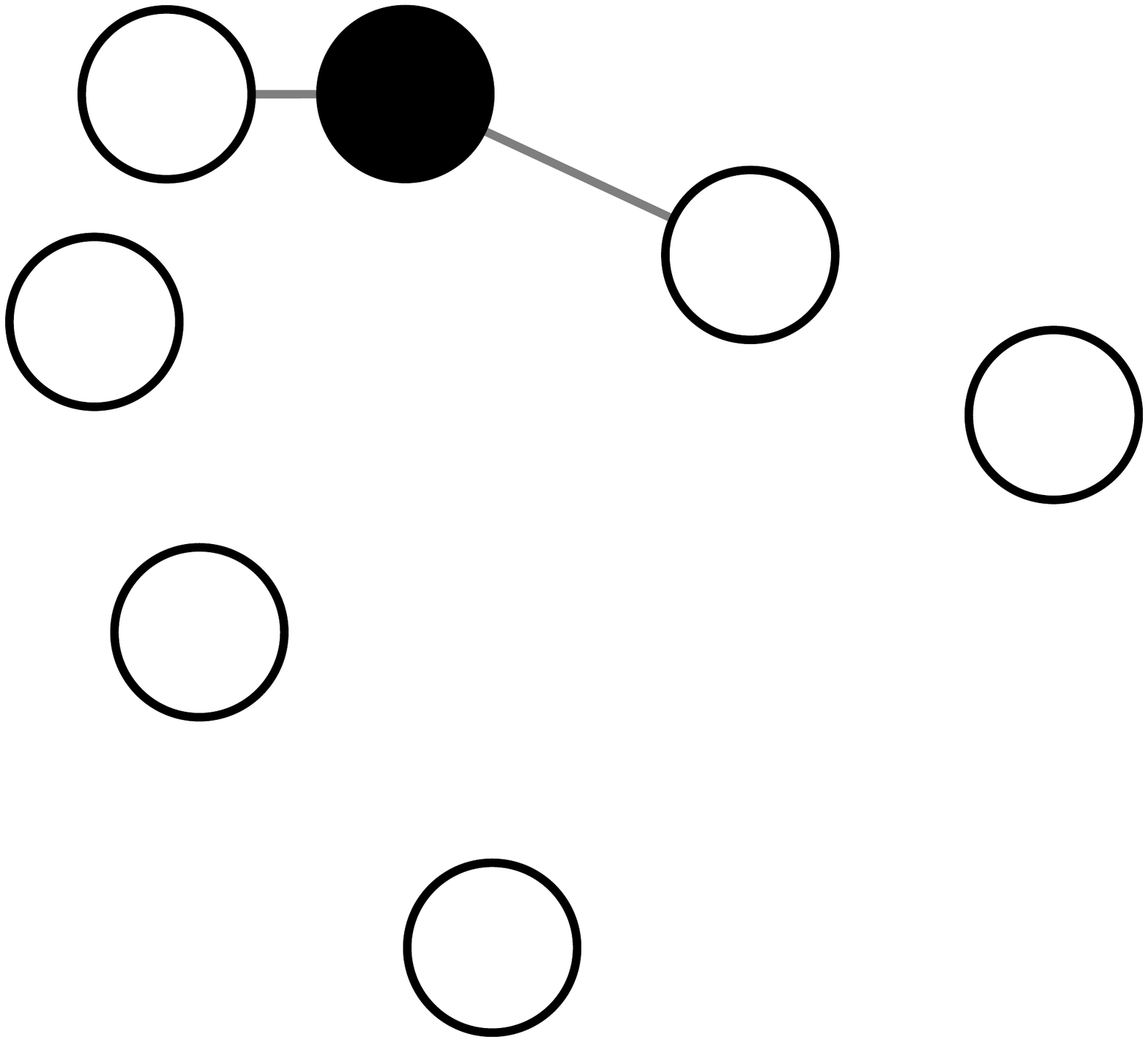}
  \caption{The same dynamics after different scaling at a single vertex $v$ (top right vertex). We start from  a state vector concentrated in the left top vertex. The gray scale of each vertex indicates the relative frequency visited by the random walk after certain time (before convergence). The left is the same dynamics from the middle of \figref{fig:UniformScale}. Notice here how bigger $\tau_v$ (middle) leads to slower propagation of probability density on the right branch. With a even bigger $\tau_v$ (right), however, the random walk gets trapped at vertex $v$.}
  \label{fig:nonUniformScale}
\end{figure}

Under the generalized Laplacian framework, non-uniform scaling can be understood as rescaling the mean waiting time at each vertex $i$ by $\tau_i$. Different $\TT$ lead to different dynamics, with state vectors specified by:
\begin{equation}
\ttheta(t) =\ttheta(0)\cdot e^{- (\DD_{\WW}-\WW)\DD_{\WW}^{-1} \TT^{-1} t} \;.
\end{equation}
Non-uniform scaling does not affect the trajectory of the random walk, i.e., the sequence of vertices, $i_0 , i_1, ...., i_t$,  visited by the random walk during some time interval does not depend on $\TT$. What changes with $\TT$ is only the waiting time at each vertex, i.e., the time the walk spends on a vertex before a transition. As \figref{fig:nonUniformScale} shows, non-uniform scaling can produce rather complicated dynamics.

\subsection{Reweighing Transformations}
\label{sec:rewighting-transformations}
The last parameterization of the generalized Laplacian we explore is one that transforms the adjacency matrix of a graph $\AA$ to the interaction matrix $\WW$. Given an adjacency matrix of a graph $\AA$, the choice of $\WW$ is a rather flexible design option. In fact, we can manipulate the adjacency matrix in any way (we can even come up completely new matrices) as long as the result is still an positive semi-definite and symmetric matrix, for any perceived ``interaction graph" of dynamics.


In this paper, we limit our attention to scaling transformations of the original adjacency matrix $\AA$, as defined in the previous subsection. Notice that this scaling is casted on $\AA$ inside of the generalized Laplacian, which will also change the degree matrix into $\DD_{\WW}$. To differentiate this transformation from the scaling of the generalized Laplacian as a whole, we call it the \emph{reweighing transformation}. Whereas the scaling transformation changes the delay time at each vertex, the reweighing transformation changes the trajectory of a dynamic process.



As described in Section~\ref{sec:background}, a biased random walk with transition probability $P_{ij} \propto b_i a_{ij}$ is equivalent to an unbiased random walk on an ``interaction graph", represented by a reweighed adjacency matrix:
\begin{equation}
 w_{ij} = b_i a_{ij} b_j \;.
\end{equation}
For intuitions, we impose the constrains that $b_i > 0$ . This transformation allows the generalized Laplacian to model many different types of dynamic processes by transforming them into a unbiased  walk on the interaction graph.
Several important special cases are introduced below.

\subsection{Special cases}
The simple parameterization in terms of $\TT$ and $\WW$ allows the generalized Laplacian to represent a variety of dynamic processes, including those described by the Laplacian and normalized Laplacian, as well as a continuous family of new operators that are not as well studied. It also contains certain operators for modeling epidemics. The consideration of this family of operators is also partially motivated by recent experimental work in understanding network centrality \cite{Ghosh11physrev,Lerman12pre}.


\paragraph{Normalized Laplacian}
If the interaction matrix is the original adjacency matrix of the graph $\WW=\AA$, and vertex delay factor is simply the identity matrix $\TT=\II$, then we recover the \emph{symmetric normalized Laplacian}:
$$\LL=\II -\DD_{\AA}^{-1/2}\AA\DD_{\AA}^{-1/2}.$$
The ``random walk" and ``consensus" formulation of this dynamic process correspond to the unbiased random walk and consensus processes described in Section~\ref{sec:background}:
$\LL^{RW}=\II -\AA\DD_{\AA}^{-1}$ and $\LL^{CON}=\II -\DD_{\AA}^{-1}\AA$.

\paragraph{(Scaled) Graph Laplacian}
When $\WW=\AA$, $\TT =d_{max}\DD_{\AA}^{-1}$,  the generalized Laplacian operator corresponds to the (scaled) graph {\em Laplacian}
$$\LL=\frac{1}{d_{max}}  (\DD - \AA).$$
This operator is often used to describe \emph{heat diffusion} processes \cite{Chung07pnas}, where $\LL$ is replacing the continuous Laplacian operator $\nabla^2$.

Notice that by setting $\TT =d_{max}\DD_{\AA}^{-1}$, the diagonal matrix $\TT\DD_{\WW}$ becomes effectively a scalar. As a result, different similarity transformation (other values of $\rho$ in Eq.~\ref{eq:spreading-operator}) lead to identical linear operators, meaning the ``random walk" and ``consensus" formulations are exactly the same as the symmetric formulation.

\paragraph{Replicator}
\label{eq:rep}
Let $\vv_{\AA}$ be the eigenvector of $\AA$ associated with its largest eigenvalue $\lambda_{max}$: $\AA\vv_{\AA}=\lambda_{max}\vv_{\AA}$.
We can then construct a diagonal matrix $\VV_{\AA} $ whose elements are the components of the eigenvector $\vv_{\AA}$.
Let us scale the adjacency matrix according to $\WW = \VV_{\AA}\AA\VV_{\AA}$ and use it as the interaction matrix.
Setting the vertex delay factor to identity, the spreading operator is:
$$\LL=\II - \DD_{\WW}^{-1/2}\WW\DD_{\WW}^{-1/2}=\II-\frac{1}{\lambda_{max}}\AA\;,$$
where the entries in $\DD_{\WW}$ simplifies as ${d_{\WW}}_i = \sum_j {v_{\AA}}_i a_{ij}{\vv_{\AA}}_j =  {v_{\AA}}_i\sum_j  a_{ij}{\vv_{\AA}}_j = \lambda_{max} {v_{\AA}}_i^2$.
This operator is known as the replicator matrix $\RR$, and it models \emph{epidemic diffusion}  on a graph \cite{Lerman12pre}. It is simply the normalized Laplacian of the interaction graph $\VV_{\AA}\AA\VV_{\AA}$ \cite{Smith13spectral}, by reweighing the adjacency graph $\AA$ with the eigenvector centralities of each vertex.

Using the random walk formulation, an URW on $\VV_{\AA}\AA\VV_{\AA}$ is equivalent to a maximum entropy random walk on the original graph $\AA$ \cite{burda_localization_2009,lambiotte_flow_2011}. Its solution is
\begin{equation}
 \theta_i (t+1) = \sum_j \dfrac{{v_{\AA}}_iw_{ij}}{\lambda_{max}{v_{\AA}}_j}\theta_j (t) \;.
\end{equation}
This means that both dynamics have exactly the same state vector $\theta$ at each time step. In particular, the stationary distributions are both
\begin{equation}
 \pi_i = \dfrac{{v_{\AA}}_i^2}{\sum_i {v_{\AA}}_i^2} \;.
\end{equation}

The consensus formulation of the Replicator gives a maximum entropy agreement dynamics:
$$\LL^{CON}=\II-\frac{1}{\lambda_{max}}\VV_{\AA}^{-1}\AA\VV_{\AA}\;.$$

\paragraph{Unbiased Laplacian}
Reweighing each edge by the inverse of the square root of the endpoint degrees gives the what is known as the normalized adjacency matrix $\WW = \DD_{\AA}^{-1/2}\AA\DD_{\AA}^{-1/2}$.
Then, the degree of vertex $i$ of the reweighted graph is ${ d_{\WW}}_i
= \sum_{j\in V} \WW[i,j]$.
With $\TT={d_{\WW}}_{max}{D}_{\WW}^{-1}$ we define the \emph{unbiased Laplacian matrix}:
$$\LL = \frac{1}{{d_{\WW}}_{max}}({\DD_{\WW} - \WW)}.$$

Unbiased Laplacian is an example of the degree based biased random walk with $P_{ij} \propto d_i^{-1/2} a_{ij}$ (\secref{sec:background}). An URW on the reweighed adjacency matrix $\WW$ is equivalent to a BRW on the original adjacency matrix of the following dynamics
\begin{equation}
 \theta_i (t+1) = \sum_j \dfrac{d^{-1/2}_ia_{ij}}{\sum_k d^{-1/2}_ka_{kj}}\theta_j(t) \;.
\end{equation}

The stationary distribution for this class of BRWs in general is
\begin{equation}
 \pi_i = \dfrac{\sum_i d_i^{\beta} a_{ij} d_j^{\beta}}{\sum_{ij} d_i^{\beta} a_{ij} d_j^{\beta}} \;.
\end{equation}

Just like the (scaled) graph Laplacian, the diagonal matrix $\TT\DD_{\WW}$ of the unbiased Laplacian is effectively a scalar. As a result, the ``random walk" and ``consensus" formulations are exactly the same as the symmetric formulation. This is a result of our design choice in vertex scaling factors $\TT$, whose motivation will be explained later.

These four special cases are related to each other through various transformations introduced earlier in this section, which are captured by the following diagram.
\begin{center}
\begin{tikzpicture}
  \matrix (m) [matrix of math nodes,row sep=3em,column sep=4em,minimum width=2em]
  {
     Normalized\quad Laplacian & Laplacian \\
     Replicator &  Unbiased\quad Laplacian\\};
  \path[-stealth]
    (m-1-1) edge node [above] {scaling} (m-1-2)
    (m-1-1) edge node [below] {reweighing+scaling} (m-2-2)
    (m-1-1) edge node [left] {reweighing} (m-2-1)
    (m-1-2) edge node [right] {reweighing+scaling} (m-2-2);
\end{tikzpicture}
\end{center} 

\section{Generalized Centrality}
\label{sec:centrality}
The generalized Laplacian allows us to study properties of networks, such as vertex centrality. Recall that centrality is used to capture how ``central" or important a vertex is in a network. In the context of dynamical systems, we aim for the following desirable properties when defining a proper centrality measure:
\begin{itemize}
 \item Centrality should be a per-vertex measure, with all values being positive scalars.
 \item Centrality of a vertex should be strongly related with the vertex state variable.
 \item Centrality should be independent of initial state vectors.
 \item Centrality should not change over time during the dynamical process.
\end{itemize}
The above conditions ensure that the centrality measure of a vertex is solely determined by the structure of the network and the interactions taking place on it. It follows our intuition that the importance of a vertex should depend on the structure of a network, not the specific initializations of the network, nor should it change over time.

The various centrality measures introduced in the past have lead to very different conclusions about the relative importance of vertices \cite{Page99thepagerank,Bonacich01,Ghosh11physrev,ghosh_rethinking_2012,Lerman12pre}. Among them, degree centrality, eigenvector centrality and page rank have been defined based on fundamentally different approaches. Our generalized Laplacian framework unifies these disparate centrality measures and shows them to be related to solutions of different dynamic processes on the network.

\subsection{Stationary Distribution of a Random Walk}
In random walks, a vertex has high centrality if it is visited frequently by the random walk. This is given by the state vector of the random walk. Equation~\ref{eq:dynamics} gives the weight distribution of the
dynamic process at any time
\begin{equation}
\label{eq:solution_RW}
\ttheta(t) = e^{ -\LL^{RW} t}\cdot\ttheta(0)
= \sum_{k=0}^{\infty}\frac{(-t)^{k}}{k!}{{\LL^{RW}}^{k}\ttheta(0)},
\end{equation}
where $\ttheta(0)$ is the state vector describing the initial distribution of the process.
The stationary distribution of the process is:
\begin{equation}
\label{eq:equilibrium}
\lim_{t \rightarrow \infty}\ttheta(t) = \ppi \quad with\quad \pi_i=\frac{{d_{\WW}}_{i}\tau_i}{\sum_{j} {d_{\WW}}_{j}\tau_j}\;,
\end{equation}
because
\begin{equation*}
    (\DD_{\WW}-\WW) (\TT\DD_{\WW})^{-1} \Pi =  (\DD_{\WW}-\WW) \overrightarrow{1} = \overrightarrow{0}\;,
\end{equation*}
with $\ppi$ being the vector with $\pi$ entries and $\Pi$ being the diagonal matrix with the same elements. By convention, $\ppi$ is the standard centrality measure in conservative processes, including random walks \cite{ghosh_rethinking_2012}.

By defining centrality as the stationary distribution of a random walk, the intuition of vertex importance becomes the total time a random walk spends at each vertex after convergence. This is proportional to both vertex degrees and vertex delay factors, which will be later interpreted as the volume measure. If $\LL^{RW}$ is a normalized Laplacian, this centrality measure is exactly the heat kernel page rank \cite{ChungLocalHeatKernal}, which is identical to degree centralities since $\WW=\AA$ and $\TT = \II$.

\subsection{Stationary Distribution in Consensus Dynamics}
In consensus processes, the state vector always converges to a uniform ``belief" state, where each vertex has the same value of the dynamic variable. As a result, the stationary distribution is not an appropriate measure of vertex centrality, since it deem all vertices to be equally important. However, the value of uniform ``belief" every vertex agrees to is
\begin{equation}
\pi_i = \dfrac{1}{\sum_j {d_{\WW}}_j \tau_j}\sum_i \theta_i(0) {d_{\WW}}_i \tau_i \;,
\end{equation}
where weight of vertex $i$ in this average is given by
\begin{equation}
\frac{{d_{\WW}}_{i}\tau_i}{\sum_{j} {d_{\WW}}_{j}\tau_j}\;.
\end{equation}

Intuitively, as a measure of importance, it make sense to define the centrality of a vertex in the consensus process as its contribution to the final agreement value. This consistency between ``consensus" and ``random walk" in centrality measures is more than a coincidence, and lead us to define \emph{generalized centrality}.

\subsection{Generalized Centrality}
As shown in \secref{sec:similarity-transformations}, the matrices connected through a similarity transformation represent the same linear operator up to a change of basis. For example, the relationship between ``consensus" and ``random walk" dynamics can be represented as the following diagram.
\begin{center}
\begin{tikzpicture}
  \matrix (m) [matrix of math nodes,row sep=3em,column sep=4em,minimum width=2em]
  {
     \left[\ttheta(0)\right]_{CON} & \left[\frac{\overrightarrow{1}}{\sum_j d_j\tau_j}\right]_{CON} \\
     \ttheta(0) & \frac{d_i\tau_i}{\sum_j d_j\tau_j} \\};
  \path[-stealth]
    (m-1-1) edge node [above] {${\LL^{CON}}^{\infty}$} (m-1-2)
    (m-2-1) edge node [below] {${\LL^{RW}}^{\infty}$} (m-2-2)
            edge node [left] {$(\DD\TT)^{-1}$} (m-1-1)
    (m-1-2) edge node [right] {$\DD\TT$} (m-2-2);
\end{tikzpicture}
\end{center}

The above equivalence applies to all state vectors at any time, including the stationary state. To verify, we solve the generalized     Laplacian. First, we rewrite the initial (at $t=0$) state vector in terms of  the eigenvectors of $\LL$ $\{\vv_1, \vv_2,...,\vv_n \}$ , which form an eigenbasis for the space spanned by $\LL$. Keep in mind here we have indexed the eigenvectors by their corresponding eigenvalue in ascending order $\lambda_1 < \lambda_2 < ...$, with the smallest $\lambda_1$ as the dominant eigenvalue.
\begin{align}
 \ttheta(0) = z_1 \vv_1 + z_2 \vv_2 +... +z_n \vv_n\;.
\end{align}
The new coordinates $\overrightarrow{z} = \{z_1, z_2,..., z_n\}$ are connected to the original coordinates under the standard basis through the change of basis matrix $\mathbb{V}$, whose columns are composed of $\{\vv_1, \vv_2,...,\vv_n \}$:
\begin{align*}
 \mathbb{V} \overrightarrow{z} = \ttheta(0)\;,\quad\quad \overrightarrow{z} =  \mathbb{V}^{-1}  \ttheta(0) = \mathbb{U}^{T} \ttheta(0)\;.
\end{align*}

Note that the matrices $\mathbb{V}$ and $\mathbb{U}^{T}$ are inverse of each other, and they form a dual basis. In fact, if matrix $B$ is diagonalizable, we have
\begin{align*}
		B =& \mathbb{V}\Lambda\mathbb{V}^{-1}\\
 \mathbb{V}^{-1}B =& \Lambda \mathbb{V}^{-1}\;,
\end{align*}
which means $\mathbb{U}^{T} = \mathbb{V}^{-1}$ is simply a row composition of left eigenvectors of $\LL$, which we will write as $\{\uu_1^T, \uu_2^T,...,\uu_n^T \}$. Now, we solve for the state vector at stationary, similar to what we did with random walks~\eqref{eq:solution_RW}.

\begin{align}
\label{eq:solution}
\ttheta(t) =& e^{ -\LL t}\cdot\ttheta(0) = \sum_{k=0}^{\infty}\frac{(-t)^{k}}{k!}{{\LL}^{k}\ttheta(0)}\nonumber\\
	   =& \sum_{k=0}^{\infty}\frac{(-t)^{k}}{k!}{{\LL}^{k} (z_1 \vv_1 + z_2 \vv_2 +... +z_n \vv_n)}\nonumber\\
	   =& \sum_i \sum_{k=0}^{\infty} \frac{(-t)^{k}}{k!}{{\lambda_i}^{k} z_i \vv_i}\nonumber\\
	   =& \sum_i z_i e^{-\lambda_i t}\vv_i \nonumber\\
	   =& \sum_i \uu_i^T\ttheta(0) e^{-\lambda_i t} \vv_i\nonumber\\
	   =& \mathbb{V} e^{-\Lambda t}\mathbb{U}^{T} \ttheta(0)\;,
\end{align}
where in the last step we used matrices to simplify the notation, with $\Lambda$ being the diagonal matrix with eigenvalues. One interesting observation is that by left multiplying both sides with $\mathbb{U}^T$, we have
\begin{align*}
 \mathbb{U}^T \ttheta(t) =  \mathbb{U}^T \mathbb{V} e^{-\Lambda t}\mathbb{U}^{T} \ttheta(0) =  e^{-\Lambda t}\mathbb{U}^{T} \ttheta(0)\;.
\end{align*}
Recall that $\mathbb{U}^{T} \ttheta$ is a vector in the eigenbasis $\mathbb{V}$. Applying the operator $\LL$ to any input vector simply re-scales it according to eigenvalues. Since the smallest eigenvalue of the generalized Laplacian is always $0$, we have
$$\uu_1^T \ttheta(t) = e^{-\lambda_1 t}\uu_1^T\ttheta(0) = \uu_1^T\ttheta(0)\;,$$
which states that the state vector is conserved along the direction of the dominant eigenvector $\vv_1$.

The state vector reaches a stationary distribution $\pi$
\begin{align}
\pi=\lim_{t \rightarrow \infty} \ttheta(t) = \lim_{t \rightarrow \infty} e^{\lambda_1 t}\ttheta(t) =  z_1 (\dfrac{e^{\lambda_1}}{e^{\lambda_1}})^t\vv_1 + z_2 (\dfrac{e^{\lambda_1}}{e^{\lambda_2}})^t\vv_2 +... +z_n (\dfrac{e^{\lambda_1}}{e^{\lambda_n}})^t \vv_n \approx z_1 \vv_1\;,
\end{align}
Since all terms vanishes as $t \rightarrow \infty$, and the stationary state vector $\pi$ only depends on $\vv_1$. According to the conditions we listed at the start of the section, $z_1\vv_1$ qualifies as a time invariant, initialization independent vertex measure.

\tabref{tab:centrality} summarizes the properties of the stationary distributions and centralities associated with different similarity transformation of the generalized Laplacian. $[\ttheta]_{\rho}$ represents the vector under the basis specified by the $\rho$ parameter, with the random walk as the standard basis.

\begin{table}
\caption{Stationary state vectors of different formulations of the generalized Laplacian.}
\begin{minipage}{\textwidth}\begin{center}
\setlength{\tabcolsep}{.3em}
\bgroup\def\arraystretch{1.5}
	\begin{tabular}{cccccc}
	\hline\noalign{\vspace {-.5cm}}\hline
	Formulations &$[\ttheta(0)]_{\rho}$ &${\uu_1}[i]$ &$z_1$  & ${\vv_1}[i]$ &  $[\pi_i]_{\rho}$\\
	\hline
	$\LL^{SYM}$ 	   &$(DT)^{-1/2}\ttheta(0)$ &$\frac{\sqrt{d_i\tau_i}}{\sqrt{\sum_j d_j\tau_j}}$	&$\frac{1}{\sqrt{\sum_j d_j\tau_j}}$	
		   &$\frac{\sqrt{d_i\tau_i}}{\sqrt{\sum_j d_j\tau_j}}$ 	&$\frac{\sqrt{d_j\tau_j}}{\sum_j d_j\tau_j}$\\
	$\LL^{RW}$ &$\ttheta(0)$ &$\frac{1}{\sqrt{\sum_j d_j\tau_j}}$	&$\frac{1}{\sqrt{\sum_j d_j\tau_j}}$					
		   &$\frac{d_i\tau_i}{\sqrt{\sum_j d_j\tau_j}}$		&$\frac{d_j\tau_j}{\sum_j d_j\tau_j}$ \\
	$\LL^{CON}$&$(DT)^{-1}\ttheta(0)$ &$\frac{d_i\tau_i}{\sqrt{\sum_j d_j\tau_j}}$	 &$\frac{1}{\sqrt{\sum_j d_j\tau_j}}$	 
		   &$\frac{1}{\sqrt{\sum_j d_j\tau_j}}$			&$\frac{1}{\sum_j d_j\tau_j}$\\
	\hline\noalign{\vspace {-.5cm}}\hline
	\end{tabular}\egroup
\end{center}\end{minipage}
\label{tab:centrality}
\end{table}

The spectral theorem states that any symmetric real matrix has an orthonormal basis $\VV$ which consists of its eigenvectors. Under the generalized Laplacian framework, the symmetric formulation with $\rho =0$ falls into this category. In the above table, we have chosen the normalization of the orthonormal basis $\sqrt{\sum_j d_j\tau_j}$ as the common normalization for all formulations.

As the table shows, similarity transformations of the same operator give the same the state vector $\ttheta$, as long as the input and output vectors are properly transformed into the correct basis. They represent the same linear operator and the same dynamics in different coordinate systems. Since centrality is determined by the dynamic process on a given network, it should be unified across these similarity transformations. In theory, any coordinate system can be set as the standard. Here, following the intuitions described earlier, we define the unnormalized stationary state vector of the random walk as the \emph{generalized centrality}:
$$c_i = {d_{\WW}}_i\tau_i\;.$$

Another motivation behind this definition is to establish a direct connection between centrality and community measures, as we will later demonstrate with the notion of \emph{generalized volume}~\eqref{eq:hS}.

\subsection{Transformations and Special Cases}
Generalized centrality includes many well known centrality measures as special cases. Below, we summarize the induced special cases discussed in the previous subsection.

\paragraph{Normalized Laplacian} $\WW=\AA$ and $\TT =\II$, and hence the generalized centrality reduces to degree centrality $c_i = d_i$.

\paragraph{(Scaled) Graph Laplacian}  $\WW=\AA$ and $\TT =d_{max}\DD_{\AA}^{-1}$, hence the generalized centrality measure here is uniform with $c_i = d_{max}$. This intuition is easier to see if one considers the unnormalized Laplacian as a consensus operator, as it is often used to calculate the unweighted average of vertex states \cite{olfati-saber_consensus_2007}.

\paragraph{Replicator}  $\WW = \VV_{\AA}\AA\VV_{\AA}$ and $\TT =\II$. Recall that $\vv_{\AA}$ is the eigenvector of $\AA$ associated with the largest eigenvalue $\lambda_{\max}$. The generalized centrality in this case is $c_i = \lambda_{\max}{\vv_{\AA}}_i^2$, which corresponds to the stationary distribution of a maximal-entropy random walk on the original graph $\AA$. Note that $\vv_{\AA}$, also known as the \emph{eigenvector centrality}, was introduced by Bonacich \cite{Bonacich01} to explain the importance of actors in a social network based on the importance of the actors to which they were connected.

\paragraph{Unbiased Laplacian} $\WW = \DD_{\AA}^{-1/2}\AA\DD_{\AA}^{-1/2}.$ and $\TT={d_{\WW}}_{max}{\DD}_{\WW}^{-1}$. Similar to the (Scaled) Graph Laplacian, the generalized centrality measure here is uniform with $c_i = {d_{\WW}}_{max}$.

\paragraph{Other transformations} Besides the above special cases, we can use any transformation introduced in the last section for new dynamics, and the corresponding generalized centrality will be immediately apparent. Scaling transformations change $\tau_i$ terms, while reweighing transformations change ${d_{\WW}}_i$. Similarity transform has no effect on generalized centrality by definition.

\section{Generalized Community Quality}
\label{sec:community}
Now we investigate the effect of different dynamics on network communities. A community is a subset of vertices that interact more with each other according to the rules of a dynamic process than with outside vertices. A \emph{quality function} measures the degree to which this interaction is happening within communities. Here in the context of dynamical systems, we use a set of properties to constrain our choice of quality functions:
\begin{itemize}
 \item Community quality should be a global measure of interactions.
 \item Community quality should be invariant of initial state vectors.
 \item Community quality should should not change over time during the dynamical process.
 \item Community quality of a subset should be strongly correlated to the change of state variable of member vertices.
\end{itemize}
The above conditions ensure that the quality function is solely determined by the choice of communities, network structure and the interactions between vertices. We assume that the underlying network structure remains static as the dynamics unfolds. There is a catch, however, by simply dividing each vertex into its own community, we would have a optimal but trivial community division. Therefore, we need additional constrains on the size of the communities.

A closely related problem in geometry is the isoperimetric problem, which relates the circumference of a region to its area. Isoperimetric inequalities lie at the heart of the study of expander graphs in graph theory. In graphs, area translate into the size of vertex subsets and the circumference translate into the size of their boundary \cite{Chung1997Spectral}. In particular, we will focus on the graph bisection (cut) problem, which restricts the number of communities to two. For bisections, the constrains on community sizes becomes a balancing problem.

Just like centrality, various community measures used in previous literature lead to very different conclusions about community structure \cite{Fortunato10, Newman2006, Rosvall08, Zhu2014}. In this section, we will demonstrate that for graph bisection, some of them are essentially graph isoperimetric solutions under our generalized Laplacian framework, and more importantly, each one corresponds to a unified community measure for a class of similar operators including seemingly different formulations of ``consensus", ``symmetric" and ``random walk".

\subsection{Generalized Conductance}
\label {sec:Quality}
Recall that conductance is a community quality measure associated with unbiased random walks.
\begin{equation}
   \phi(S)  =  \frac{\cut_{\AA}(S,\bar{S})}{{\min(\vol(S),\vol(\bar{S}))}}\;,
\end{equation}
where $\vol(S) = \sum_{i\in S}d_i$ and $\cut_{\AA} = \sum_{i\in S,j\in \bar S} a_{i,j}$.

We generalize this notion with a claim that every dynamic process has an associated function that measures the quality of the cluster with respect to that process. Optimizing the quality function leads to cohesive communities, i.e., groups of vertices that ``trap'' the specific dynamic process for a long period of time.

Consider a dynamic process defined by a spreading operator $\LL=\TT^{-1/2}\DD_{\WW}^{-1/2}(\DD_{\WW}-\WW)\DD_{\WW}^{-1/2}\TT^{-1/2}$. We define the \emph{generalized conductance} of a set $S$ with respect to $\LL$ as:
\begin{align}
\label{eq:hS}
h_{\LL}(S)
=&\frac{\cut_{\WW}(S,\bar S)}{\min\left(\vol_{\LL}(S),\vol_{\LL}(\bar
  S)\right)}\nonumber\\
=&\frac{\sum_{i\in S,j\in \bar S} w_{i,j}}{\min\left(\sum_{i\in \mathrm{S}} {d_{\WW}}_i\tau_i,\sum_{i\in {\mathrm{\bar{S}}}}{d_{\WW}}_i\tau_i\right)}
\end{align}
\noindent We also define the minimum over all possible $S$ as the generalized conductance of the graph,
\begin{equation}
\label{eq:phiG}
\phi_{\LL}(G)=\min_{S\in V}h_{\LL}(S)\;.
 \end{equation}

Notice that we have generalized the volume measure of a set $S \subseteq V$ to $\vol_{\LL}(S) = \sum_{i\in S} c_i = \sum_{i\in S} {d_{\WW}}_i\tau_i$, which is the sum of generalized centrality of member vertices. This direct connection between generalized centrality and generalized conductance justifies our previous definition of generalized centrality (\secref{sec:centrality}).

Using the random walk perspective, the numerator measures the random jumps across communities, while the denominator ensures a balanced bisection. As previously pointed out, the presence of a good cut implies that it will take a random walk a long time to over come this boundary and before reaching its stationary distribution. This corresponds to a small numerator. The generalized volume can be interpreted as the total time a random walk stays within a community after convergence, as it is proportional to both vertex degrees and vertex delay factors. This interpretation of the denominator coincides with our definition of generalized centrality.

\subsection{Transformations and Special Cases}
Now with the generalized quality measures for communities defined, let us investigate how various transformations under the generalized Laplacian affect them, ultimately leading to different bisections found by our algorithms.

We can use any transformation to produce new dynamics, and the corresponding generalized conductance will be redefined according to Equation~\eqref{eq:hS},
\begin{align}
h_{\LL}(S)
=&\frac{\sum_{i\in S,j\in \bar S} w_{i,j}}{\min\left(\sum_{i\in \mathrm{S}} {d_{\WW}}_i\tau_i,\sum_{i\in {\mathrm{\bar{S}}}}{d_{\WW}}_i\tau_i\right)}
\end{align}

However, the effect of transformations on the resulting communities is not as obvious when compared with the generalized centrality. Below, we elaborate the effect of transformations on the generalized conductance measure in cases and examples.

First of all, the similarity transformation keeps both numerator and denominator the same, which make the quality function of the same communities identical. This ultimately leads to identical generalized conductances, which is the minimum over all possible bisections. Uniform scaling does change the denominator. However, because all possible bisections are scaled uniformly, the relative quality measure remain the same, leading to identical generalized conductances communities.

From the algorithmic perspective, both similarity and uniform scaling transformations preserve spectral properties of the operator matrix. Since spectrum is the only input information our spectral dynamics clustering algorithm~\ref{alg:spec} uses, we always expect to get the same solution after the transformations. This is not the case  with non-uniform scaling and reweighing transformations.

With non-uniform scaling, the numerator remains unchanged. It is each vertex's delay time change that scales the volume measures in the denominator, which in turn results in different optimal bisections becaues of the balance constrain. The following toy example demonstrates how non-uniform scaling affects balanced cuts.

\begin{figure}
    \includegraphics[width=0.32\textwidth]{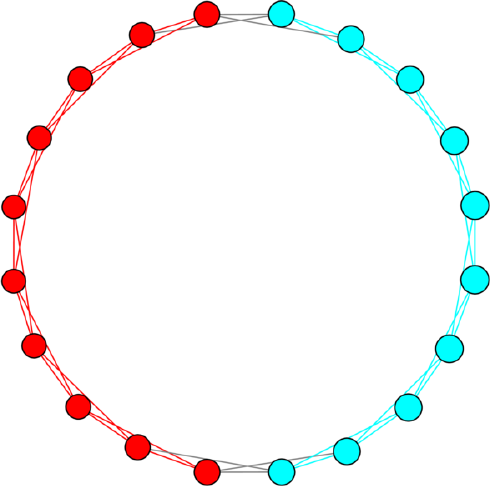}
    \includegraphics[width=0.32\textwidth]{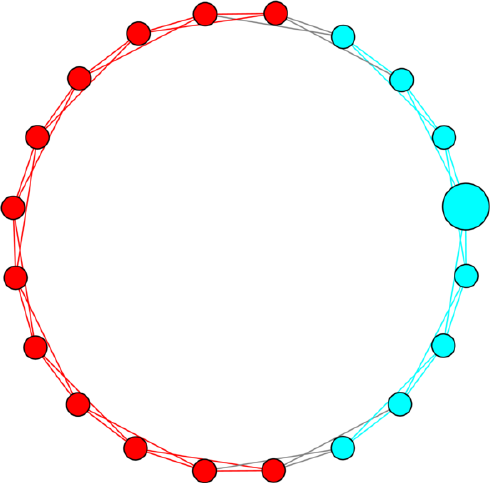}
    \includegraphics[width=0.32\textwidth]{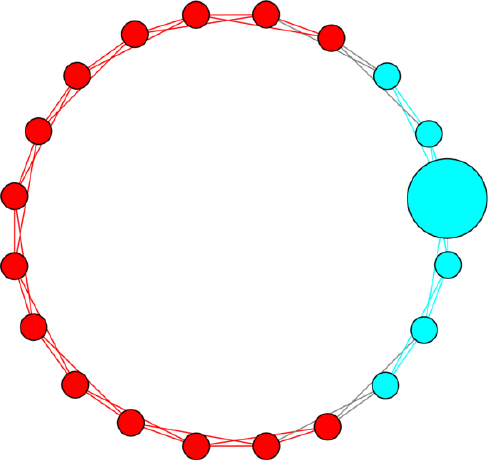}
  \caption{Optimal bisections with different non-uniform scalings}
\end{figure}

The reweighing transformation is the most complex of all, changing both the numerator and denominator in Equation~\eqref{eq:hS}.
This trade off between cut and balance can oftentimes be very complicated to analyze (as you will see with real world networks). For the purpose of illustration, we control the volume measures in the denominator by scaling the relevant vertex's delay time to almost $0$ in the following toy examples.
\begin{figure}
    \includegraphics[width=0.32\textwidth]{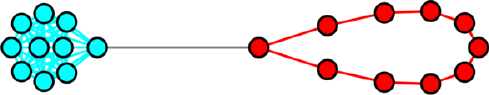}
    \includegraphics[width=0.32\textwidth]{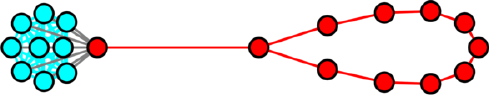}
    \includegraphics[width=0.32\textwidth]{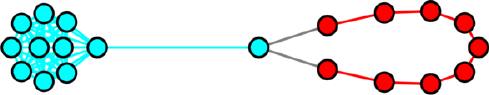}\\
    \includegraphics[width=0.43\textwidth]{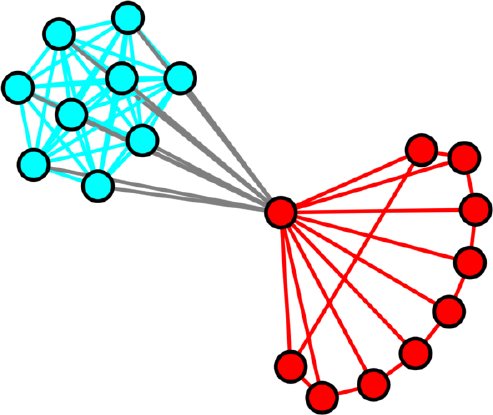}
    \includegraphics[width=0.43\textwidth]{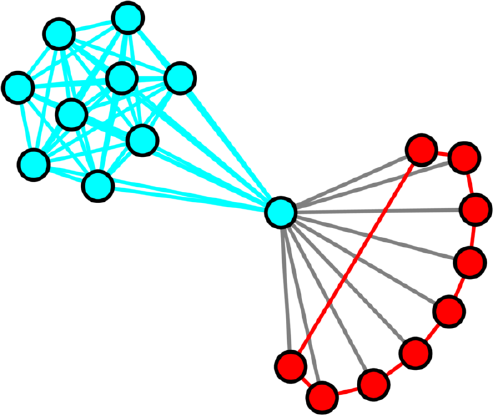}
  \caption{Optimal bisections with different interaction matrices. Here we have set the vertex delay factor of the vertices which changed its color to $\tau = 10^{-4}$ while all other vertices have $\tau=1$. This setup makes all bisections have almost identical denominators. It is the interaction matrix $\WW$ via different reweighing transformations changes the numerator (cut width), leading to different optimal bisections. Top row: different degree based biased random walks. From left to right, $\beta=1,5,-1$ respectively. Bottom row: left is the normalized Laplacian while the right is the replicator.}
\end{figure}

Finally, we summarize the induced special cases.
\paragraph{Normalized Laplacian} $\WW=\AA$ and $\TT =\II$, and hence
$h_{\LL}(S)$ is the conductance.

\paragraph{(Scaled) Graph Laplacian}  $\WW=\AA$ and $\TT =d_{max}\DD_{\AA}^{-1}$, hence
$$ h_{\LL}(S) = \frac{\cut_{\AA}(S,\bar S)}{\min(d_{\max}|S|,d_{max}|\bar{S}|)} = \frac{1}{d_{\max}}\cdot \frac{\sum_{i\in S,j\in \bar S} a_{i,j}}{\min(|S|,|\bar{S}|)},$$
This is the ratio cut scaled by $1/d_{\max}$.

\paragraph{Replicator}  $\WW = \VV_{\AA}\AA\VV_{\AA}$ and $\TT =\II$. Recall $\vv_{\AA}$ is the eigenvector of $\AA$ associated with the largest eigenvalue $\lambda_{\max}$.
The redefined cut size is $ \sum_{i\in S,j \in \bar{S}} w_{ij}=
 \sum_{i\in S,j \in \bar{S}}{v_{\AA}}_i a_{ij} {v_{\AA}}_j$. Therefore,
$$h_{\LL}(S) = \frac{\sum_{i\in \mathrm{S}, j\in {\mathrm{\bar S}}} {v_{\AA}}_i a_{ij}
				  {v_{\AA}}_j}{\lambda_{\max} \min\left(\sum_{i\in \mathrm{S}}{v_{\AA}}^{2}_i,\sum_{i \in {\mathrm{\bar S}}}{v_{\AA}}^{2}_i\right)}$$
Since the degree of a vertex in the interaction graph $\WW$ is ${d_{\WW}}_{i}=\sum_j w_{ij}=\lambda_{\max}{v_{\AA}}^{2}_i$, the generalized conductance of the Replicator is simply the conductance of the interaction graph \cite{Smith13spectral}.

\paragraph{Unbiased Laplaican} $\WW = \DD_{\AA}^{-1/2}\AA\DD_{\AA}^{-1/2}.$ and $\TT={d_{\WW}}_{max}{D}_{\WW}^{-1}$. The associated quality function is
$$h_{\LL}(S) = \frac{1}{{d_{\WW}}_{max}}\cdot \frac{\sum_{i\in S, j\in \bar S} \frac{a_{i,j}}{\sqrt{d_i d_j}}  }{\min(|S|,|\bar S|)}.$$
We call this quality function \emph{unbiased cohesion}.

Recall that the four special cases are related by transformations. By our analysis of each transformation, we expect their resulting optimal bisections to have the following relations:
\begin{center}
\begin{tikzpicture}
  \matrix (m) [matrix of math nodes,row sep=3em,column sep=4em,minimum width=2em]
  {
     Normalized\quad Laplacian & Laplacian \\
     Replicator &  Unbiased\quad Laplacian\\};
  \path[-stealth]
    (m-1-1) edge node [above] {share numerator} (m-1-2)
    (m-1-1) edge node [below] {reweighing+scaling both} (m-2-2)
    (m-1-1) edge node [left]  {reweighing both} (m-2-1)
    (m-1-2) edge node [right] {share denominator} (m-2-2);
\end{tikzpicture}
\end{center}

Notice that here the conductance for Laplacian and unbiased Laplacian share the same denominator even though they are related through both reweighing and scaling transformations. This is a result of their scaling cancelling out the reweighing effect on volumes (centralities). This is part of the motivation behind our design of the unbiased Laplacian operator for easier comparisons. We will be using these relationships for analyzing experimental results in the next section.

\subsection{Generalized Cheeger Inequality}
\label{sec:algorithms}
Given the generalized conductance measure, finding the best community bisection is still a combinatorial problem, which quickly becomes computationally intractable as the network grow in size. In this subsection we will generalized theorems for the classic Laplacian to our generalized setting, ultimately leading to efficient approximate algorithms with theoretical guarantees. For mathematical convenience we will use the symmetric formulation and assume that $\rho=0$ for $\LL$.
\label{sec:generalized-cheeger}
Cheeger inequality states that
$$ {\phi^2 (G)}/{2} \leq \lambda_2 \leq 2 \phi(G) $$
where $\lambda_2$ is the second smallest eigenvalue of the symmetric normalized
Laplacian, $\boldsymbol{\mathcal{L}}=\boldsymbol{I} -
\boldsymbol{D}^{-1/2}\boldsymbol{W}\boldsymbol{D}^{-1/2},$
and $\phi(G)$ is conductance.
The relationship between conductance and spectral properties of the Laplacian enables the use of its eigenvectors for partitioning graphs, particularly the nearest-neighborhood graphs and finite-element meshes
 \cite{Spielman96spectralpartitioning}.

In this section,  we generalize Cheeger's inequality to any spreading operator under our framework and its associated generalized conductance of the graph (given by Eq.~\ref{eq:phiG}). Our generalization of Cheeger's inequality comes with algorithmic  consequences. It leads to spectral partitioning algorithms that are efficient in finding low conductance cuts for a given operator.


\begin{theorem}(Generalized Cheeger Inequality)\\
\label{th:1}
Consider the dynamic process described by a (properly scaled) spreading operator $\LL=\TT^{-1/2}\DD_{\WW}^{-1/2}(\DD_{\WW}-\WW)\DD_{\WW}^{-1/2}\TT^{-1/2}$.
Let $\lambda_1 \leq \lambda_{2} \leq ... \leq \lambda_{n}$ be the
eigenvalues of $\LL$.
Then $\lambda_{1}=0$ and $\lambda_{2}$ satisfies the following inequalities:
$$\phi_{\LL}(G)^{2}/{2} \leq \lambda_{2} \leq 2 \phi_{\LL}(G)$$
where $\phi_{\LL}(G)$ is given by Eq.~\ref{eq:phiG}.
\end{theorem}

\begin{proof} We prove the theorem by following the approach for proving the classic Cheeger's inequality (see \cite{Chung1997Spectral}).

Let $(\tau_1,...,\tau_n)$ be the diagonal entries of $\TT$, and $\bvec{v}_1$ be the eigenvector associated with $\lambda_1$.
Note that $\bvec{v}_1 = \TT^{1/2}\DD_{\WW}^{1/2}\cdot \overrightarrow{1}$, where $\overrightarrow{1}$ denotes the column vector of all 1's, is an eigenvector of $\calL$ associated
with eigenvalue $\lambda_0 = 0$.
Let $\vol_{\LL}(S) = \sum_{i\in S} d_i\tau_i$ for $S\subseteq V$, where for clarity we abuse the notation $d_i$ and use it as ${d_{\WW}}_i$.
Suppose $\bvec{f}$ is the eigenvector associated with $\lambda_2$.
Then, $\bvec{f}\perp \bvec{v}_1$.
Consider vector $\bvec{g}$ such that $g[u]=\displaystyle {f[u]}/{\sqrt{d_{u}\tau_u}}$.
The fact that $\bvec{f}\perp \bvec{v}_1$ then implies
$\displaystyle \sum_{v}g[v] {d_{v}\tau_v}=0$.
Then,
\begin{eqnarray*}
\lambda_{2} & =  & \frac{\bvec{f}^T\calL
 \bvec{f}}{\bvec{f}^{T}\bvec{f}} =  \displaystyle
\frac{\sum_{u,v\in
    V}\left(\frac{f[u]}{\sqrt{d_{u}\tau_u}}-\frac{f[v]}{\sqrt{d_{v}\tau_v}}\right)^{2}w_{u,v}}{\sum_{v}f[v]^{2}}\\
& = & \frac{\sum_{u,v\in
    V}\left(g[u] - g[v]\right)^{2}w_{u,v}}{\sum_{v}g[v]^{2}d_v\tau_v}
\end{eqnarray*}
Instead of sweeping the vertices of $G$ according to the eigenvector
$\bvec{f}$ itself,  we sweep the vertices of the graph $G$ according to $\bvec{g}$  by ordering the vertices of $G$ so that
$$g[v_1]\ge g[v_2] \ge \cdots \ge g[v_n]$$
and consider sets  $S_i=\{ v_1, \cdots , v_i \}$ for all $1 \leq i \leq n$.

Similar to \cite{Chung1997Spectral}, we will eventually only consider the first ``half" of the sets $S_i$ during the sweeping: Let $r$ denote the largest integer such that
$ \vol_{\LL}(S_r) \leq \vol_{\LL}(V)/2$.
Note that
\begin{eqnarray*}
& &\sum_v {(g[v]- g[v_r])}^2 d_v\tau_v \\
& =& \sum_v g[v]^2 d_v\tau_v + g[v_r]^2 d_v\tau_v
\geq \sum_v g[v]^2 d_v\tau_v.
\end{eqnarray*}
where the first equation follows from $\sum_v g[v]
{d_v}\tau_v=0$.
We denote the positive and negative part of $g-g[v_r]$ as $g_{+}$ and $g_{-}$ respectively:
\begin{equation}
g_+[v]=\begin{cases}
    g[v] - g[v_r], & \text{if $g[v] \ge  g[v_r]$}.\\
    0, & \text{otherwise}.
    \end{cases}
\end{equation}

\begin{equation}
    g_-[v]=\begin{cases}
    |g[v] - g[v_r]|, & \text{if $g[v] \le  g[v_r]$}.\\
    0, & \text{otherwise}.
  \end{cases}
\end{equation}
Now
\begin{eqnarray*}
&&\lambda_2  =  \frac{\sum_{u, v\in V}  (g[u]- g[v])^2 w_{u,v}}{\sum_v g[v]^2 d_v\tau_v}\nonumber \\
&\ge& \frac{\sum_{u,v\in V} (g_+[u]-g_+[v])^2w_{u,v} +(g_-[u]
  -g_-[v])^2 w_{u,v} }{\sum_v (g_+[v]^2 +g_-[v]^2)
  d_v\tau_v} \\
&\geq & \min\left[
 \frac{\sum (g_+[u]-g_+[v])^2w_{u,v}}{\sum_v g_+[v]^2
   d_v\tau_v},
\frac{\sum (g_-[u]
  -g_-[v])^2 w_{u,v} }{\sum_v g_-[v]^2
  d_v\tau_v}
\right]
\end{eqnarray*}
Without loss of generality, we assume the first ratio is at most the
second ratio, and will mostly focus on the vertices $\{v_1,....,v_r\}$ in the first ``half" of the graph in the analysis below. Thus,
\begin{eqnarray*}
\lambda_2 &\geq &   \frac{\sum_{u,v} (g_+[u]-g_+[v])^2w_{u,v}}{\sum_v g_+[v]^2
   d_v\tau_v}
\\
& \geq &
 \frac{\left(\sum_{u,v} (g^2_+[u]-g^2_+[v])w_{u,v}\right)^2}{\left(\sum_v g_+[v]^2
   d_v\tau_v\right)\left(\sum_{u,v}(g_+[u]+g_+[v])^2w_{u,v}\right)}\\
\end{eqnarray*}
which follows from the Cauchy-Schwartz inequality.

We now separately analyze the numerator and denominator. To bound the denominator, we will use the following property of $\tau_i$:
Because $\LL$ is properly scaled, $\tau_i \geq 1$ for all $i\in V$.
Therefore,
\begin{eqnarray*}
\sum_{u,v}(g_+[u]+g_+[v])^2w_{u,v}  &\leq&  \sum_{u,v}
2(g^2_+[u]+g^2_+[v])w_{u,v} \nonumber \\
 &=&  2 \sum_{u\in V} g^2_+[u] d_u \\
 &\leq&   2 \sum_{u\in V} g^2_+[u] d_u\tau_u.
\end{eqnarray*}
Hence, the denominator is at most
$$2 \left(\sum_{u\in V} g^2_+[u] d_u\tau_u\right)^2.$$

To bound the numerator, we consider subsets  of vertices $S_i=\{ v_1, \cdots , v_i
\}$   for all $1 \leq i \leq r$ and define $S_0 = \emptyset$.
First note that
\begin{eqnarray}\label{eqn:GammaVol}
\vol_{\LL}(S_{i}) - \vol_{\LL}(S_{i-1}) = d_{v_i}\tau_{v_i}.
\end{eqnarray}
By the definition of $\phi_{\LL}(G)$,
  we know $\phi_{\LL}(G) \leq \min_{i} h_{\LL}(S_i)$ for all
  $1\leq i\leq r$, where recall the function $h_S(\calL)$ is
  defined by Eq. ~\ref{eq:hS}.
Since $\vol_{\LL}(S_i) \leq \vol_{\LL}(\bar{S_i})$ for all
$1\leq i\leq r$,
 we have
\begin{eqnarray}\label{eqn:cutsize}
\cut(S_i,\bar{S_i}) \geq\phi_{\LL}\cdot \vol_{\LL}(S_i)
\end{eqnarray}

By orienting vertices according to $v_1,...,v_n$, we can express the numerator
\begin{eqnarray*}
\mbox{Num} & = &
\left(\sum_{u,v} (g^2_+[u]-g^2_+[v])w_{u,v}\right)^2  \\
& = &
\left(\sum_{i< j} \left(\sum_{k=0}^{j-i-1}
g^2_+[v_{i+k}]-g^2_+[v_{i+k+1}]\right)w_{v_i,v_j}\right)^2\\
&& \mbox{Rewrite the difference as a telescoping series}\\
& = &
\left(\sum_{i=1}^{n-1}  \left(g^2_+[v_{i}]-g^2_+[v_{i+1}]\right)\cdot
\cut(S_i,\bar{S_i})\right)^2 \\
&& \mbox{Collecting $(v_i,v_{i+1})$ terms}
\end{eqnarray*}
\begin{eqnarray*}
& \geq & \left(\sum_{i=1}^{n-1}  \left(g^2_+[v_{i}]-g^2_+[v_{i+1}]\right)\cdot
\phi_{\LL}\cdot \vol_{\LL}(S_i)\right)^2 \\
&& \mbox{By Eqn: \ref{eqn:cutsize}}\\
& = &\phi_{\LL}^2\cdot\left(\sum_{i=1}^{n}
 g^2_+[v_{i}]\cdot \left(\vol_{\LL}(S_i) -\vol_{\LL}(S_{i+1})\right)
  \right)^2 \\
&& \mbox{By Eqn. \ref{eqn:GammaVol} and $g_+(v_n) = 0$}\\
&= &\phi_{\LL}(G)^2\cdot\left(\sum_{i=1}^{n}
 g^2_+[v_{i}]\cdot d_{v_i}\tau_i\right)^2.
\end{eqnarray*}

Combining the bounds for the numerator and the denominator, we obtain
$\lambda_2 \leq\phi_{\LL}^2/2$ as stated in the theorem.
The right hand side of the theorem follows from the same argument for
the standard Cheeger Inequality.
\end{proof}

\subsection{Spectral Partitioning for Generalized Conductance}
The generalized Cheeger inequality is essential for providing theoretical guarantees to greedy community detection algorithms. In this section, we extend traditional spectral clustering algorithm to the generalized Laplacian setting.

Given a weighted graph $G=(V,E,\AA)$ and a operator $\LL$, we can use the standard sweeping method in the proof of Theorem~\ref{th:1} to find a partition $(S,\bar{S})$. This procedure is described in Algorithm~\ref{alg:2}.

\begin{algorithm}
\caption{Spectral Dynamics Clustering $(G,\LL)$}
\label{alg:2}
\noindent {\textbf{Input}:} weighted network: $G = (V,E,\AA)$,  and spreading operator $\LL$ defined by the interaction matrix $\WW$ and the vertex delay factor $\TT$.\\
{\textbf{Output}} partition: $(S,\bar S)$ \\
{\textbf{Algorithm}}
\begin{itemize}
\item Find the eigenvector $\bvec{f}$ of  $\LL=\TT^{-1/2}\DD_{\WW}^{-1/2}(\DD_{\WW}-\WW)\DD_{\WW}^{-1/2}\TT^{-1/2}$ associated with the second smallest eigenvalue of $\LL$.

\item Let vector $\bvec{g}$ be $g[u] = f[u]/\sqrt{{d_{\WW}}_u\tau_u}$.

\item Order the vertices of $G$ into $(v_1,....,v_n)$ such that  $g[v_1] \geq g[v_2] \geq ...\geq g[v_n]$.

\item Sweeping: For each $S_i =\{v_1,...,v_i\}$, compute
$$ h_{\LL}(S_i) =
  \frac{\cut(S_i,\bar{S_i})}{\min\left(\vol_{\LL}(S_i),\vol_{\LL}(\bar
    S_i)\right)}.$$

\item Output the $S_i$ with the smallest $h_{\LL}(S_i)$.
\end{itemize}
\label{alg:spec}
\end{algorithm}

Before stating the quality guarantee of the above algorithm, we quickly discuss its implementation and running time. The most expensive step is the computation of the eigenvalue vector $f$ associated with the second smallest eigenvalue of $\LL$. While one can use standard numerical methods to find   an approximation of this eigenvector -- the analysis would depend on the separation of the second and the third eigenvalue of $\LL$. Since $\LL$ is a diagonally scaled normalized Laplacian matrix, one can use the nearly-linear-time Laplacian solvers (e.g., by Spielman-Teng \cite{SpielmanTengLinear} or Koutis-Miller-Peng \cite{KoutisMillerPeng}) to solve linear systems in $\LL$.

Following \cite{SpielmanTengLinear},  let us consider the following notion of spectral approximation of $\LL$: Suppose $\lambda_2(\LL)$ the second smallest eigenvalue of $\LL$.
For $\epsilon \geq 0$, $\bar{\bvec{f}}$ is an {\em $\epsilon$-approximate second eigenvector} of $\LL$ if $ \bar{\bvec{f}} \perp \DD_{\AA}^{1/2}\TT^{1/2}\cdot \overrightarrow{1}$, and
$$
\frac{{\bvec{\bar f}}^T\LL {\bvec{\bar f}}}{{\bvec{\bar f}}^T{\bvec{\bar f}}} \leq (1+\epsilon)\cdot \lambda_2(\LL).
$$

The following proposition follows directly from the algorithm and Theorem 7.2 of \cite{SpielmanTengLinear} (using the solver from  \cite{KoutisMillerPeng}).

\begin{proposition}\label{prop:approximateEig}
For any interaction graph $G = (V,E,\WW)$ and vertex scaling factor $\TT$,
and $\epsilon,p >0$, with probability at least $1-p$, one
can compute an {\em $\epsilon$-approximate
  second eigenvector} of operator $\LL$ in time
$$
O\left(|E|\log n\log\log n \log(1/p)\log(1/\epsilon)/\epsilon\right).
$$
\end{proposition}

To use this spectral approximation algorithm (and in fact any numerical   approximation to the second eigenvector of $\LL$) in our spectral  partitioning algorithm for the dynamics, we will need a strengthened  theorem of Theorem \ref{th:1}.

\begin{theorem}(Extended Cheeger Inequality with Respect to Rayleigh Quotient)\\
\label{th:1-approx}
For any interaction graph $G = (V,E,\WW)$ and vertex scaling factor $\TT$,
(whose diagonals are $(\tau_1,...,\tau_n)$),
  for any vector $\bvec{u}$ such that
  $\bvec{u} \perp \DD_{\AA}^{1/2}\TT^{1/2}\cdot \overrightarrow{1}$,
if we order the vertices of $G$ into $(v_1,....,v_n)$ such that
   $g[v_1] \geq ...\geq  g[v_n]$, where $\bvec{g} =(\DD\TT)^{-1/2}\cdot \bvec{u}$
then
$$\frac{\left(\min_i h_{\LL}(S_i)\right)^2}{2} \leq
\frac{{\bvec{u}}^T\LL {\bvec{u}}}{{\bvec{u}}^T{\bvec{u}}}, $$
where $\LL=\TT^{-1/2}\DD_{\WW}^{-1/2}(\DD_{\WW}-\WW)\DD_{\WW}^{-1/2}\TT^{-1/2}$ and $S_i =\{v_1,...,v_i\}$.
\end{theorem}
\begin{proof}
The theorem follows directly from the proof of Theorem \ref{th:1} if we replace vector $\bvec{f}$ (the eigenvector of associated with the second smallest eigenvalue of $\LL$) by $\bvec{u}$. This theorem is the analog of a theorem by Mihail \cite{mihail89} for Laplacian matrices.
\end{proof}

The next theorem then follows directly from Proposition \ref{prop:approximateEig}, Theorem \ref{th:1-approx} and the definition of $\epsilon$-approximate second eigenvector of $\LL$ that provide a guarantee of the quality of the algorithm of this subsection.

\begin{theorem}
For any interaction graph $G = (V,E,\WW)$ and vertex delay factor $\TT$,
(whose diagonals are $(\tau_1,...,\tau_n)$),
 one can compute in time $$O(|E|\log n\log\log
 n\log(1/\epsilon)/\epsilon)$$
a partition $(S,\bar{S})$ such that
\begin{eqnarray*}
 h_{\LL}(S) &=&
  \frac{\sum_{v\in S,u\in \bar S} w_{u,v}}{\min\left(\sum_{v\in S} {d_{\WW}}_v\tau_v,
   \sum_{v\in \bar S} {d_{\WW}}_v\tau_v\right)}\\
   &\leq & \sqrt{2(1+\epsilon)\lambda_2(\LL)}
   \end{eqnarray*}
where $\TT^{-1/2}\DD_{\WW}^{-1/2}(\DD_{\WW}-\WW)\DD_{\WW}^{-1/2}\TT^{-1/2}$, $w_{u,v}$ is the
$(u,v)^{th}$ entry of the interaction matrix $\WW$, and $\lambda_2(\LL)$
  is the second smallest eigenvalue of $\LL$.
Consequently,
\begin{eqnarray*}
&& h_{\LL}(S) \leq 2\sqrt{(1+\epsilon)\phi_{\LL}(G)}\\
&=&2(1+\epsilon)\sqrt{ \min_{S^*\in V} \frac{\sum_{v\in S^*,u\in \bar {S^*}} w_{u,v}}{\min\left(\sum_{v\in S^*} {d_{\WW}}_v\tau_v,  \sum_{v\in \bar{S^*}} {d_{\WW}}_v\tau_v\right)}}\;.
\end{eqnarray*}
\label{th:generalSpectral}
\end{theorem}

\section{Experiments}
\label{sec:experiments}
So far, we focused on investigating the mathematical properties of the generalized Laplacian matrix. Now armed with the theory on the interplay between dynamics and network centrality and community-quality measures, we systematically investigate the different perspectives on network structure that emerge from different dynamic processes. We will illustrate that even this simple class of processes can lead to divergent views on a collection of widely studied real-world networks, about who the central vertices are and what are the cohesive communities.

\begin{table}
\caption{Specifications of candidate datasets}
\begin{minipage}{\textwidth}\begin{center}
\setlength{\tabcolsep}{.3em}
\bgroup\def\arraystretch{1.5}
	\begin{tabular}{lccccc}
	\hline\noalign{\vspace {-.5cm}}\hline
	Name & \#vertices  & \#edges & diameter & clustering\\
	\hline
	Zachary's Karate Club &34	&78 &5 &0.588 \\
	College Football &115 	&613	&4 &0.403	\\
	House of Representatives & 434&51033&4 &0.882 \\
	Political Blogs & 1490&16714&9&0.21 \\
	Facebook Egonets&4039 &88234&17&0.303 \\
	Power Grid &4941 &6594&46&0.107 \\
	\hline\noalign{\vspace {-.5cm}}\hline
	\end{tabular}\egroup
\end{center}\end{minipage}
\label{tab:datasets}
\end{table}

Table~\ref{tab:datasets} lists the networks we study empirically, and their properties. We treat all networks as undirected in this paper.
These networks come from different domains and embody a variety of dynamical processes and interactions, from real-world friendships (Zachary karate club \cite{ZacharyKarateClub}), to online social networks (Facebook \cite{mcauley2012learning}), to electrical power distribution (Power Grid \cite{watts1998collective}), to co-voting records (House of Representatives \cite{Smith13spectral}) and hyperlinked weblogs on US politics (Political Blogs \cite{adamic2005political}), to games played between NCAA football teams \cite{GirvanBetween}.

In the following experiments, we use two figures to demonstrate the differences between the four special cases under the generalized Laplacian framework. The first is the \emph{centrality profile}, which gives the generalized centrality for each vertex under a given a dynamic operator. To improve visualization, vertices are ordered by their centrality according to the normalized Laplacian matrix, and they are rescaled to fall within the same range.

To investigate the interplay between dynamics and their corresponding communities, we apply Algorithm~\ref{alg:spec} to these networks. Leskovec et al. \cite{Leskovec08www} used community profiles to study results of network partitioning. Community profile shows the conductance of the best bisection of the network into two communities of size $k$ and $N-k$ as $k$ varies. They found that community profiles of real-world networks reveal a ``core and whiskers'' organization of the network, with a  giant core and many small communities, or whiskers, loosely connected to the core. They showed that conductance is generally minimized for small $k$, which suggests that conductance is a poor measure for resolving network structure within the giant core.
Instead of the community profile, we use the \emph{sweep profile} to study differences in bisecting the network using different dynamic operators. Given a dynamic operator, a sweep profile gives the generalized conductance~\eqref{eq:hS} value for a potential community of size $k$ during the sweep in Algorithm~\ref{alg:spec}. Sweep profiles provide us an interesting perspective on how different dynamics lead to different communities, and they are related to centrality profiles we have shown earlier. To improve visualization, we rescale sweep profiles to lie within the same range.

Keep in mind that Algorithm~\ref{alg:spec} only guarantees a solution within a certain bound of the optimal. The visualizations of network layouts are manually created for better demonstrating the different dynamics, using a combination of ``Yifan Hu" and ``Force Atlas" algorithms from the Gephi software package \cite{gephi}.

\subsection{Zachary's Karate Club}
The first network we study is a social network consisting of $34$ members of a karate club in a university, where undirected edges represent friendships \cite{ZacharyKarateClub}. The network is made up of two assortative blocks centered around the instructor and the club president, each with a high degree hub and lower-degree peripheral vertices. With a simple community structure, this network often serves as a starting benchmark for community detection algorithms. Its centrality/sweep profiles and visualizations of optimal bisections under each special case dynamic are given below.

\begin{figure}
\centering
 \begin{subfigure}[b]{0.45\textwidth}
      \includegraphics[width=0.95\linewidth]{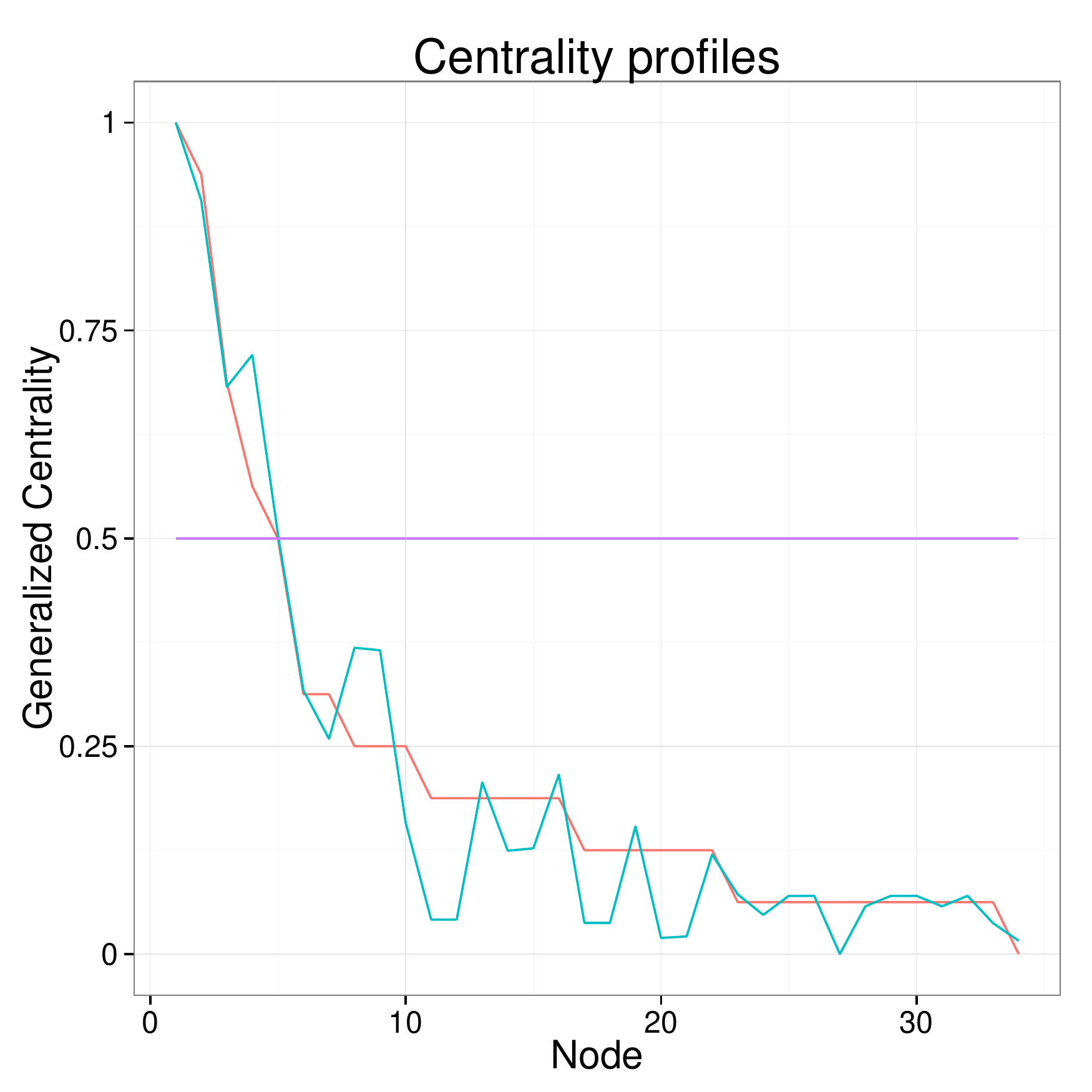}
      \caption{}\label{fig:1centrality}
 \end{subfigure}
 \begin{subfigure}[b]{0.45\textwidth}
      \includegraphics[width=0.95\linewidth]{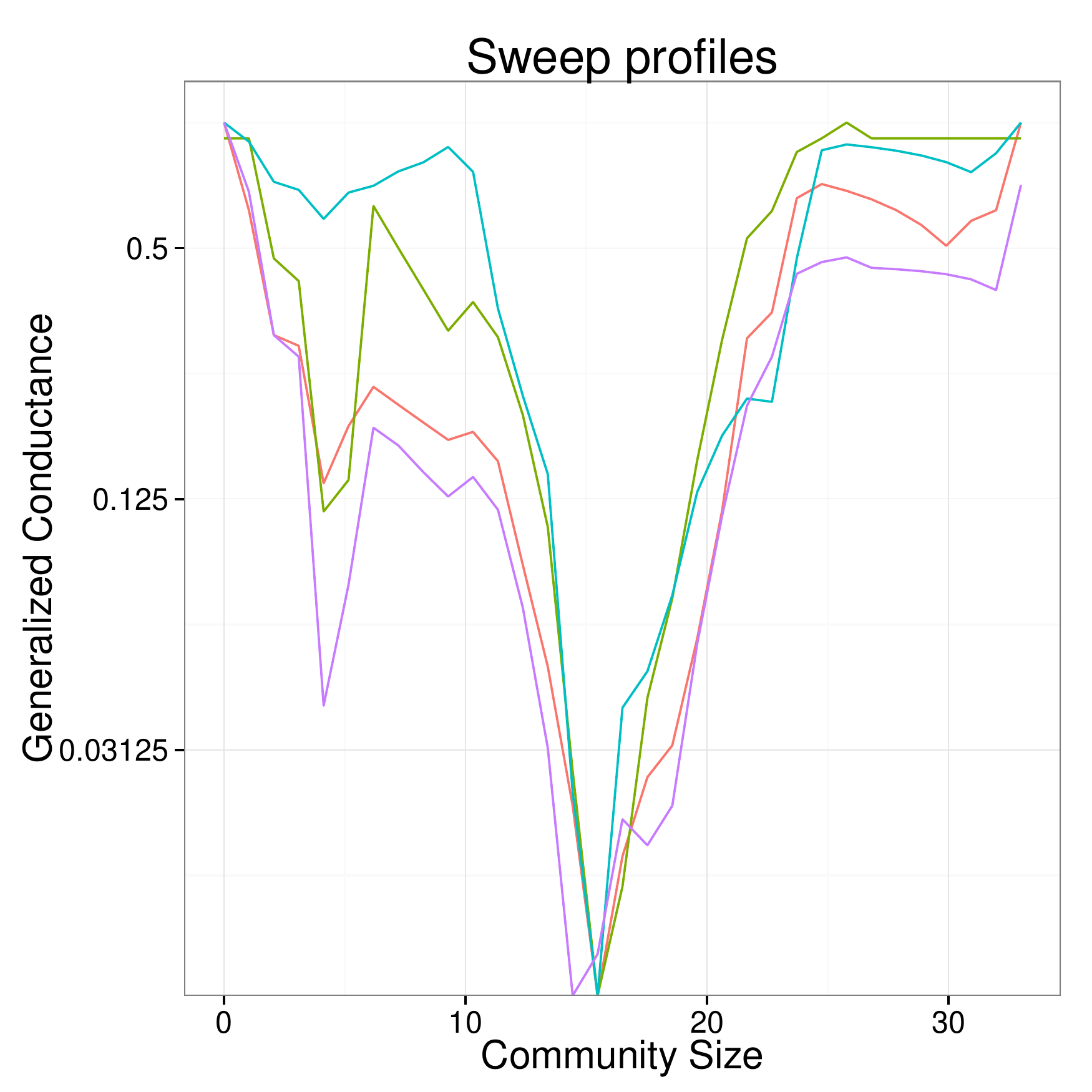}
      \caption{}\label{fig:1sweep}
 \end{subfigure}\\
 \includegraphics[width=0.98\textwidth]{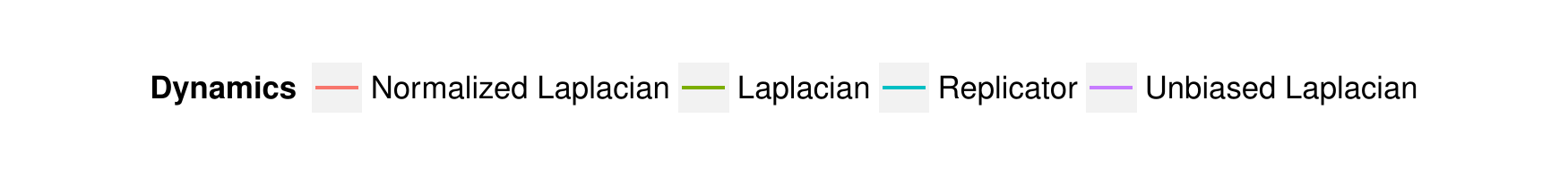}\\
 \begin{subfigure}[b]{0.32\textwidth}
      \includegraphics[width=0.95\textwidth]{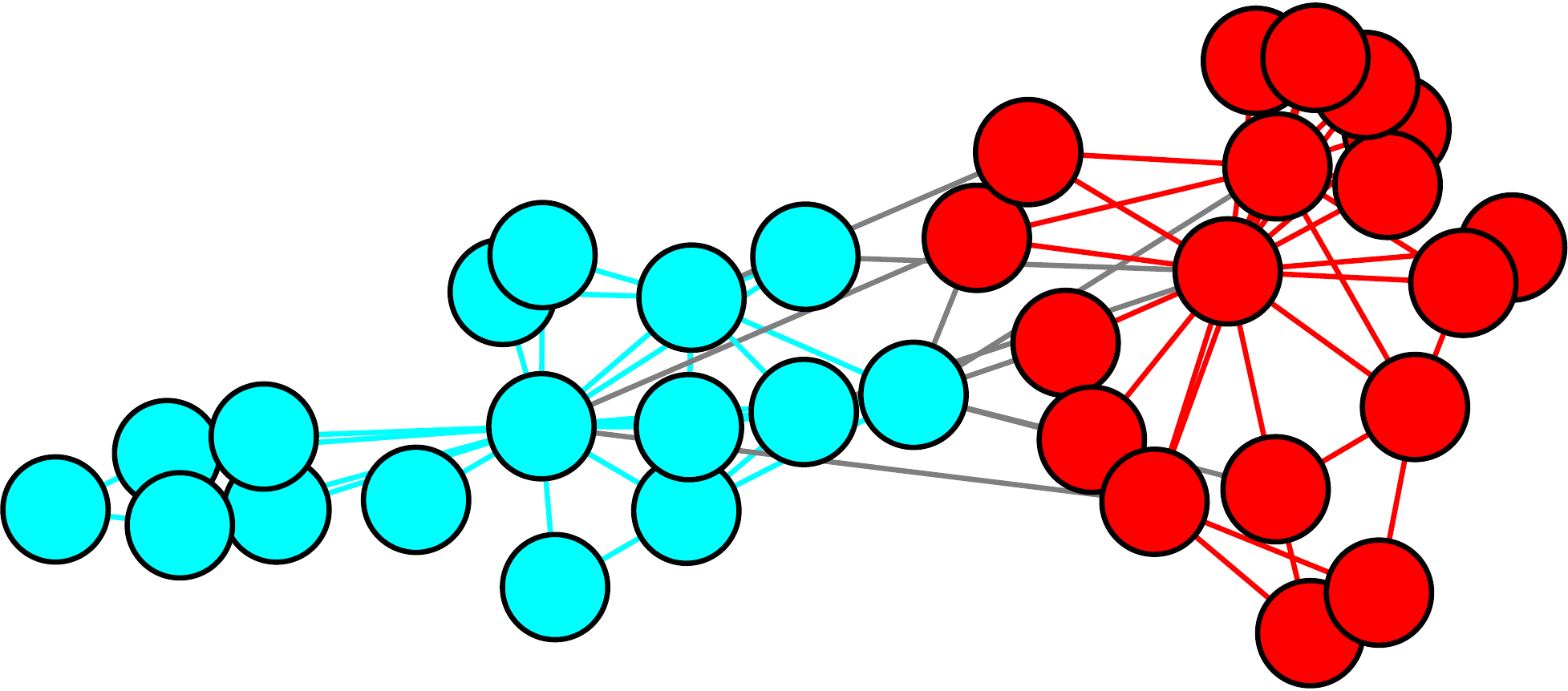}
      \caption{Normalized Laplacian}\label{fig:1norm}	
 \end{subfigure}
 \begin{subfigure}[b]{0.32\textwidth}
      \includegraphics[width=0.95\textwidth]{figures/karateOthers-crop}
      \caption{Laplacian}\label{fig:1lap}	
 \end{subfigure}\\
 \begin{subfigure}[b]{0.32\textwidth}
      \includegraphics[width=0.95\textwidth]{figures/karateOthers-crop}
      \caption{Replicator}\label{fig:1rep}	
 \end{subfigure}
 \begin{subfigure}[b]{0.32\textwidth}
      \includegraphics[width=0.95\textwidth]{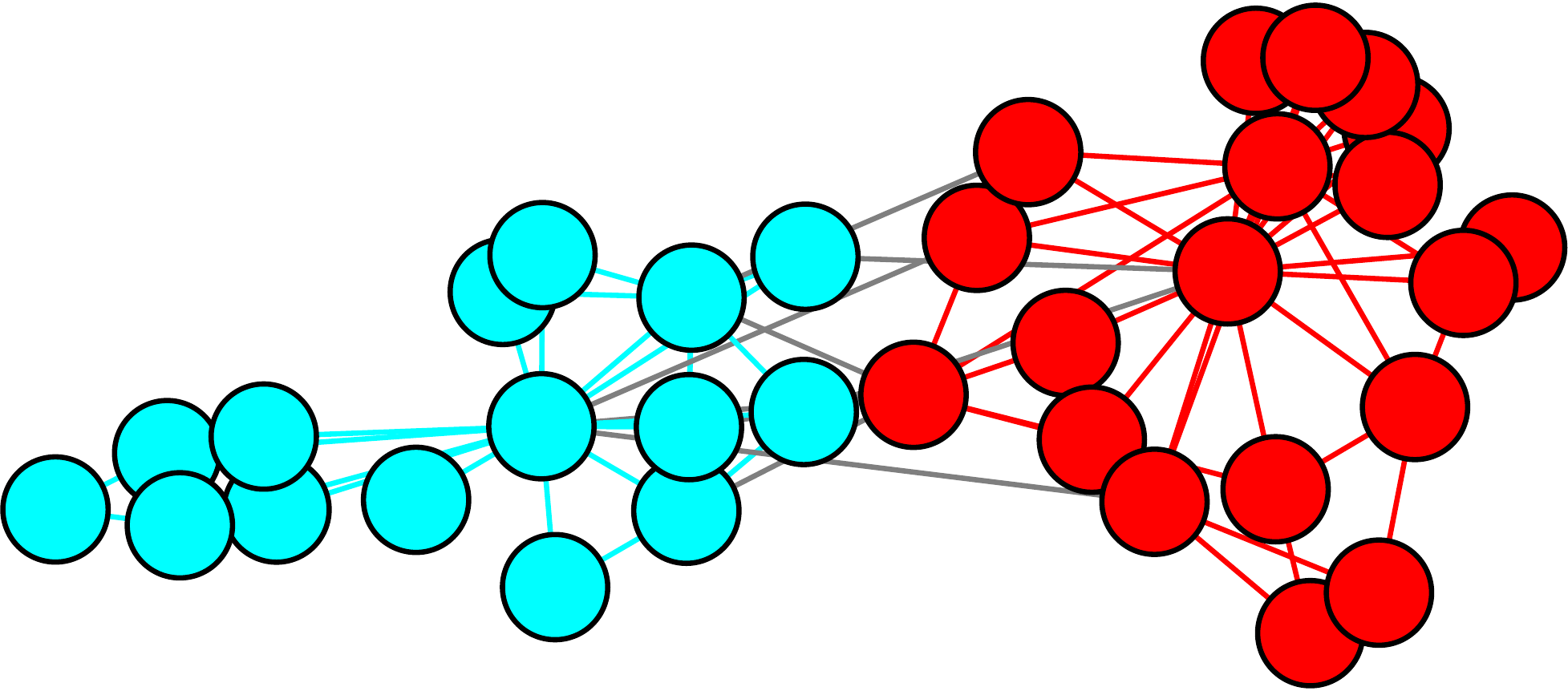}
      \caption{Unbiased Laplacian}\label{fig:1unb}	
 \end{subfigure}
 \begin{subfigure}[b]{0.32\textwidth}
      \includegraphics[width=0.95\textwidth]{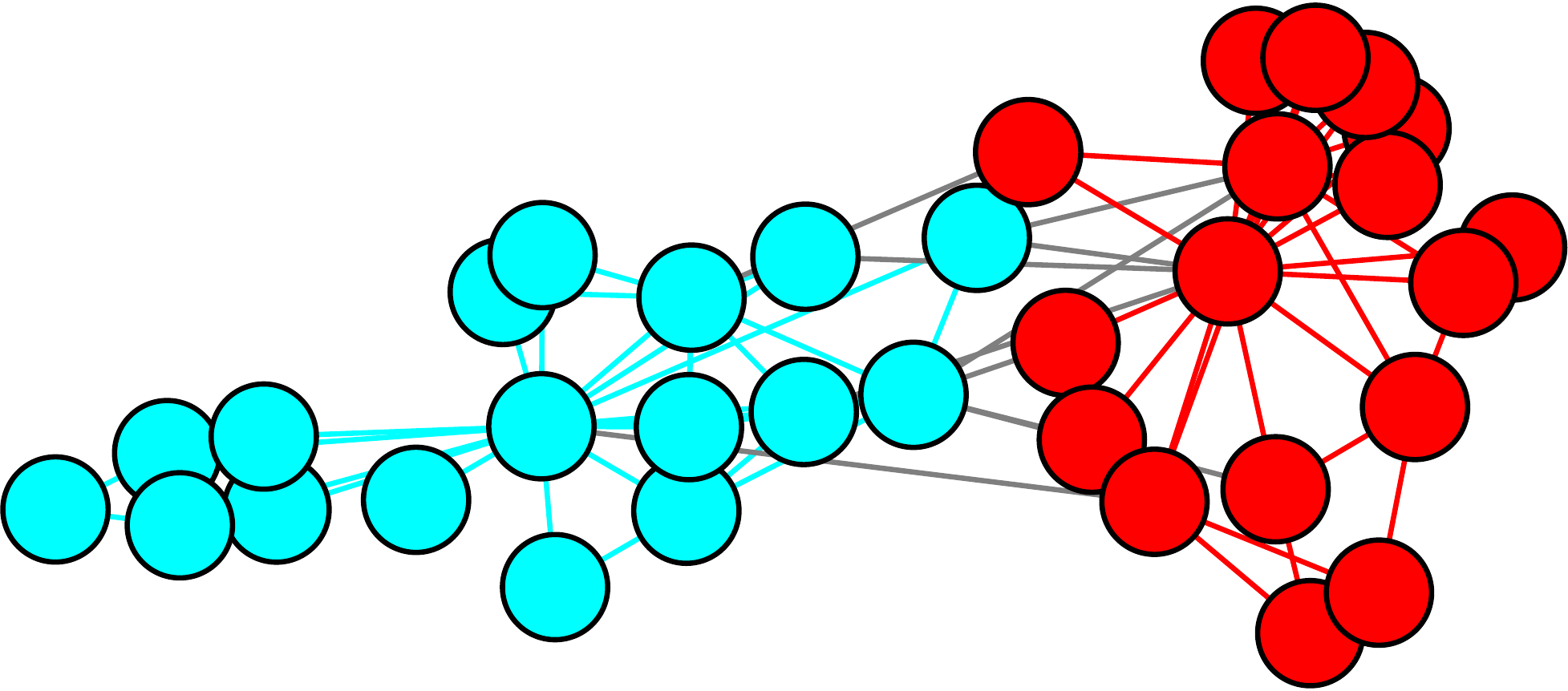}
      \caption{Ground truth communities}\label{fig:1true}	
 \end{subfigure}
 \caption{Centrality/sweep profiles and optimal bisections of Zachary's Karate Club}
 \textit{The visualizations are produced by Algorithm~\ref{alg:spec}, corresponding to normalized Laplacian, Laplacian, Replicator and Unbiased Laplacian respectively. Except for the last visualization gives the ground truth communities. Notice the mapping between the visualizations and their corresponding sweep profiles.}
\end{figure}

Just as many other community detection algorithms, the four special cases give almost identical optimal bisections for this simple network, all of which are very similar to the ground truth communities (\figref{fig:1true}). Further more, their centrality and sweep profiles are very similar as well (\figref{fig:1centrality}, \figref{fig:1sweep}). This is a excellent example showing that all good community measures capture the same fundamental idea of communities, those well-interacting subsets of vertices with relatively sparse connection in between. They do differ, however, in finer details of their mathematical definitions, as we will see in more complicated networks in the following subsections.

\subsection{College Football}
The second network represents American football games played between Division IA colleges during the regular season in Fall 2000 \cite{GirvanBetween}, where two vertices (colleges) are linked if they played in a game. Following the structure of the division, the network naturally breaks up into 12 smaller conferences, roughly corresponding to the geographic locations of the colleges. Most games are played within each conference which leading to densely connected local clusters. Its centrality/sweep profiles and visualizations of optimal bisections under each special case dynamic are given below.

\begin{figure}
 \centering
 \begin{subfigure}[b]{0.45\textwidth}
      \includegraphics[width=0.95\linewidth]{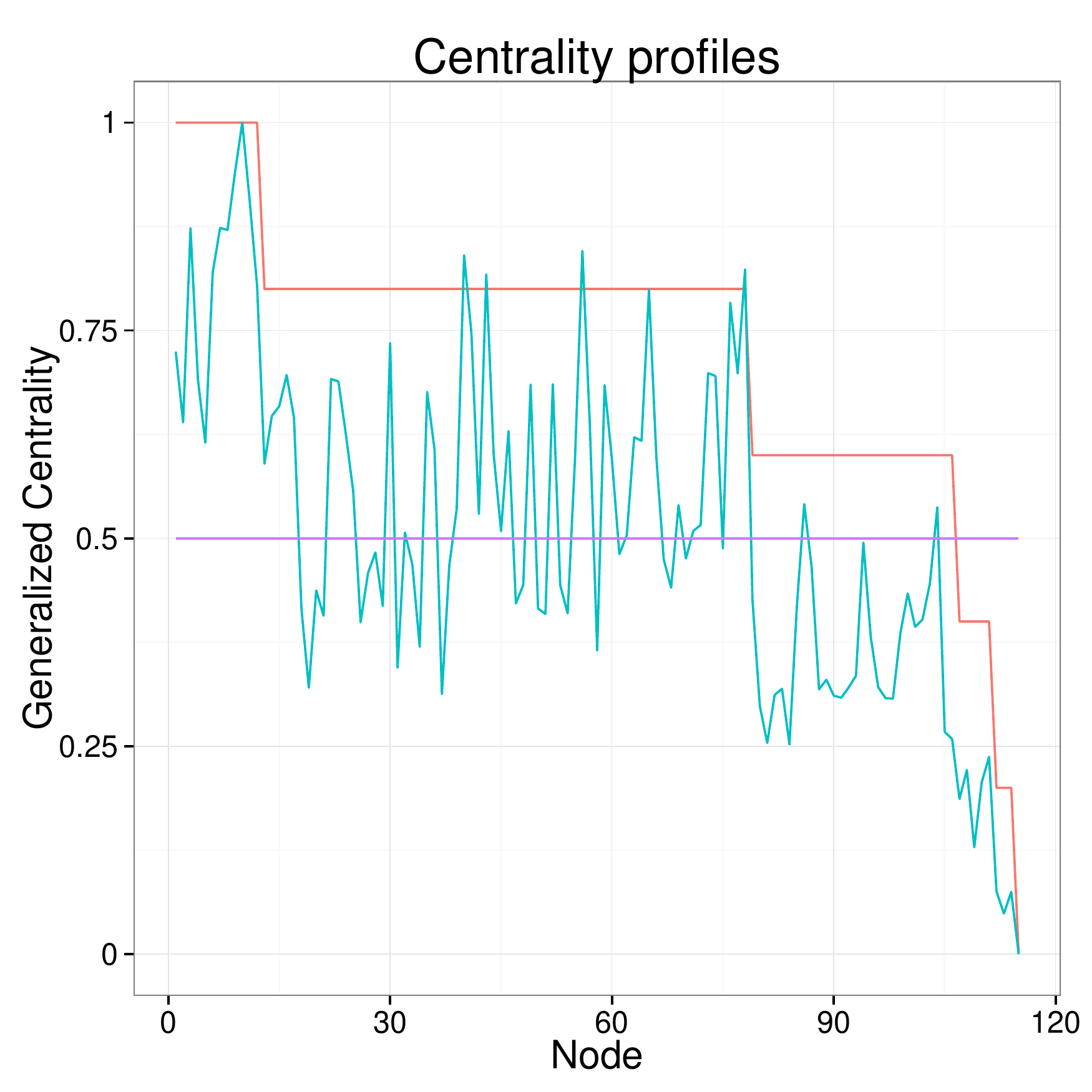}
      \caption{}\label{fig:6centrality}
 \end{subfigure}
 \begin{subfigure}[b]{0.45\textwidth}
      \includegraphics[width=0.95\linewidth]{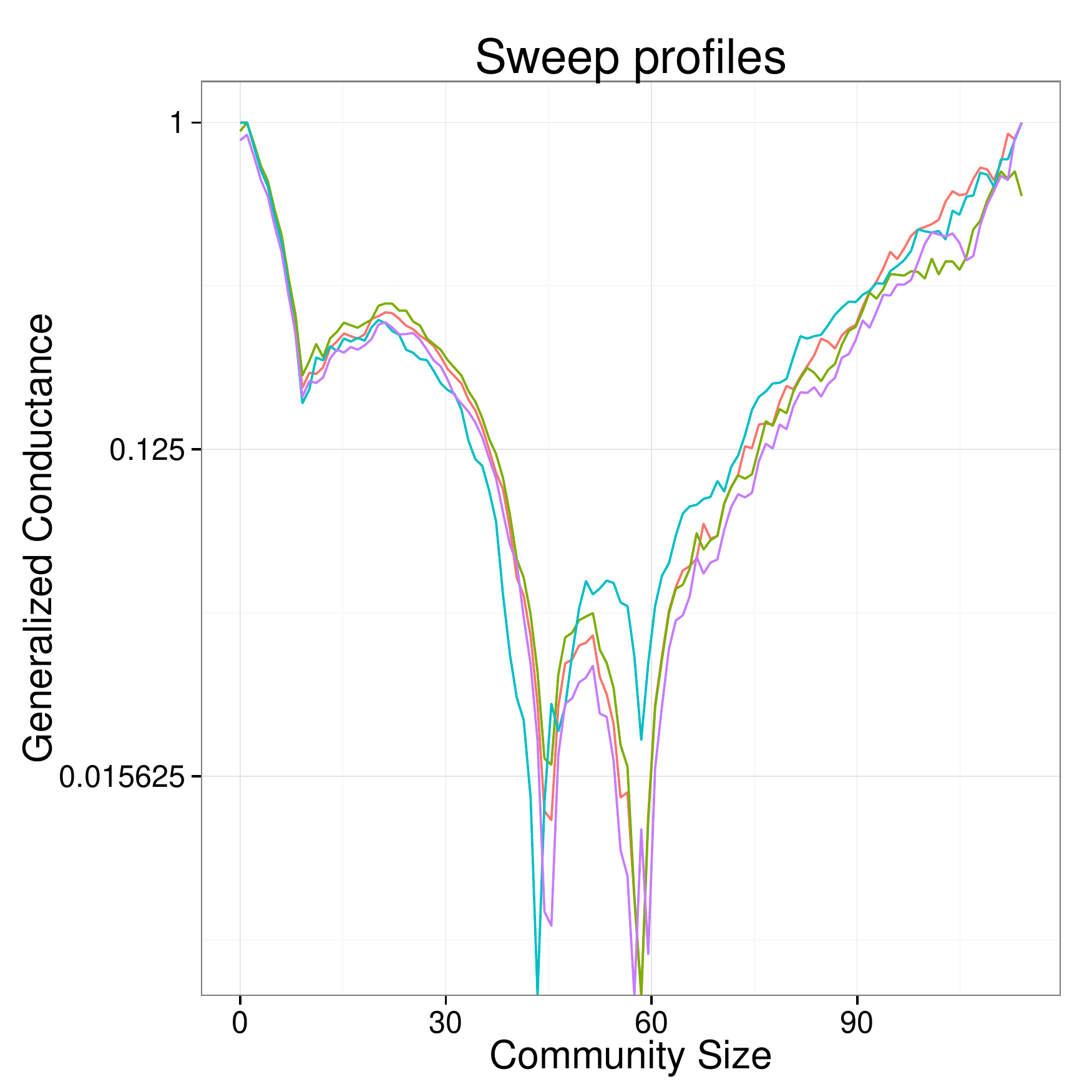}
      \caption{}\label{fig:6sweep}
 \end{subfigure}\\
 \includegraphics[width=0.98\textwidth]{figures/legendPro0}\\
 \begin{subfigure}[b]{0.32\textwidth}
      \includegraphics[width=0.95\textwidth]{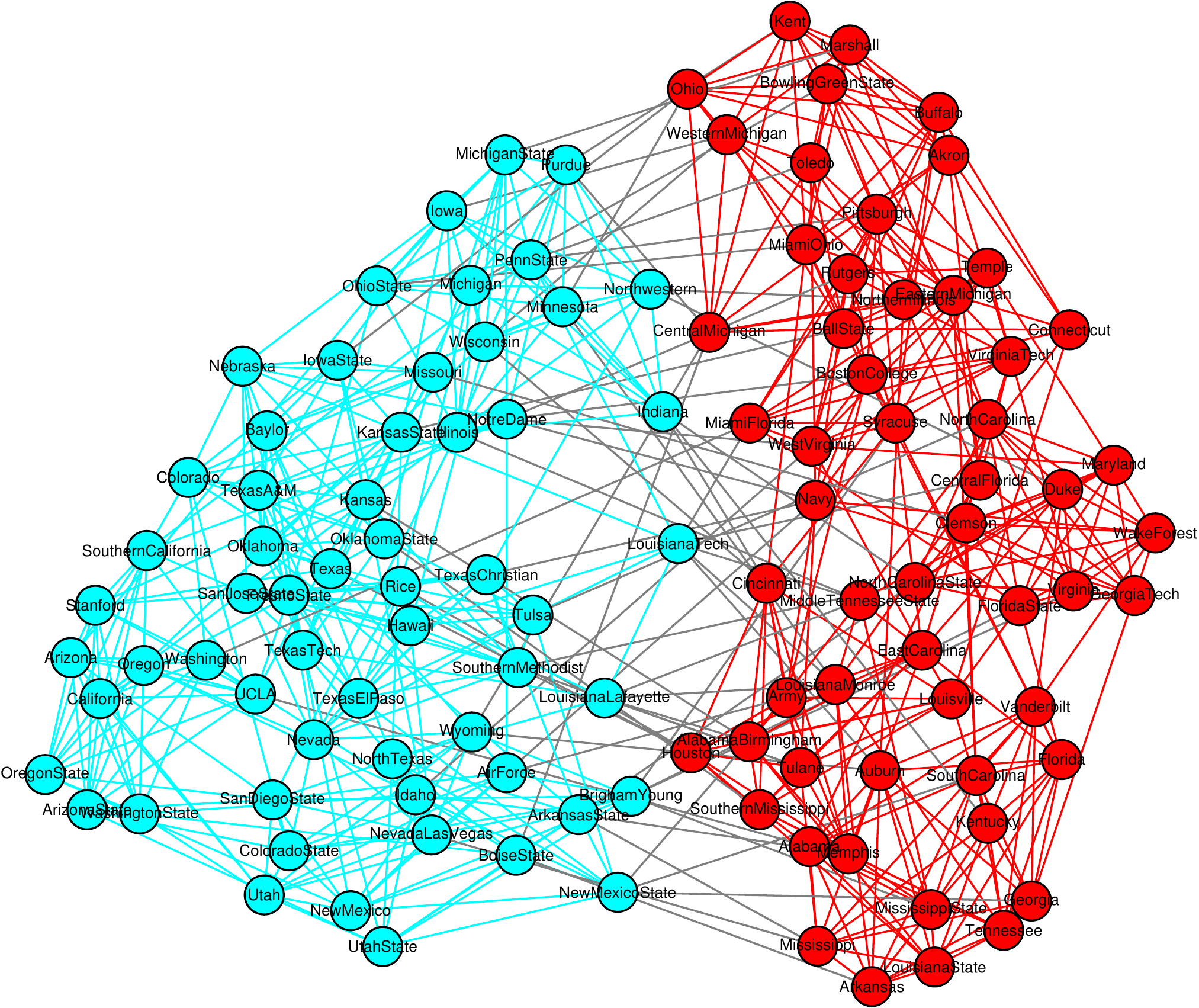}
      \caption{Normalized Laplacian}\label{fig:6norm}	
 \end{subfigure}
 \begin{subfigure}[b]{0.32\textwidth}
      \includegraphics[width=0.95\textwidth]{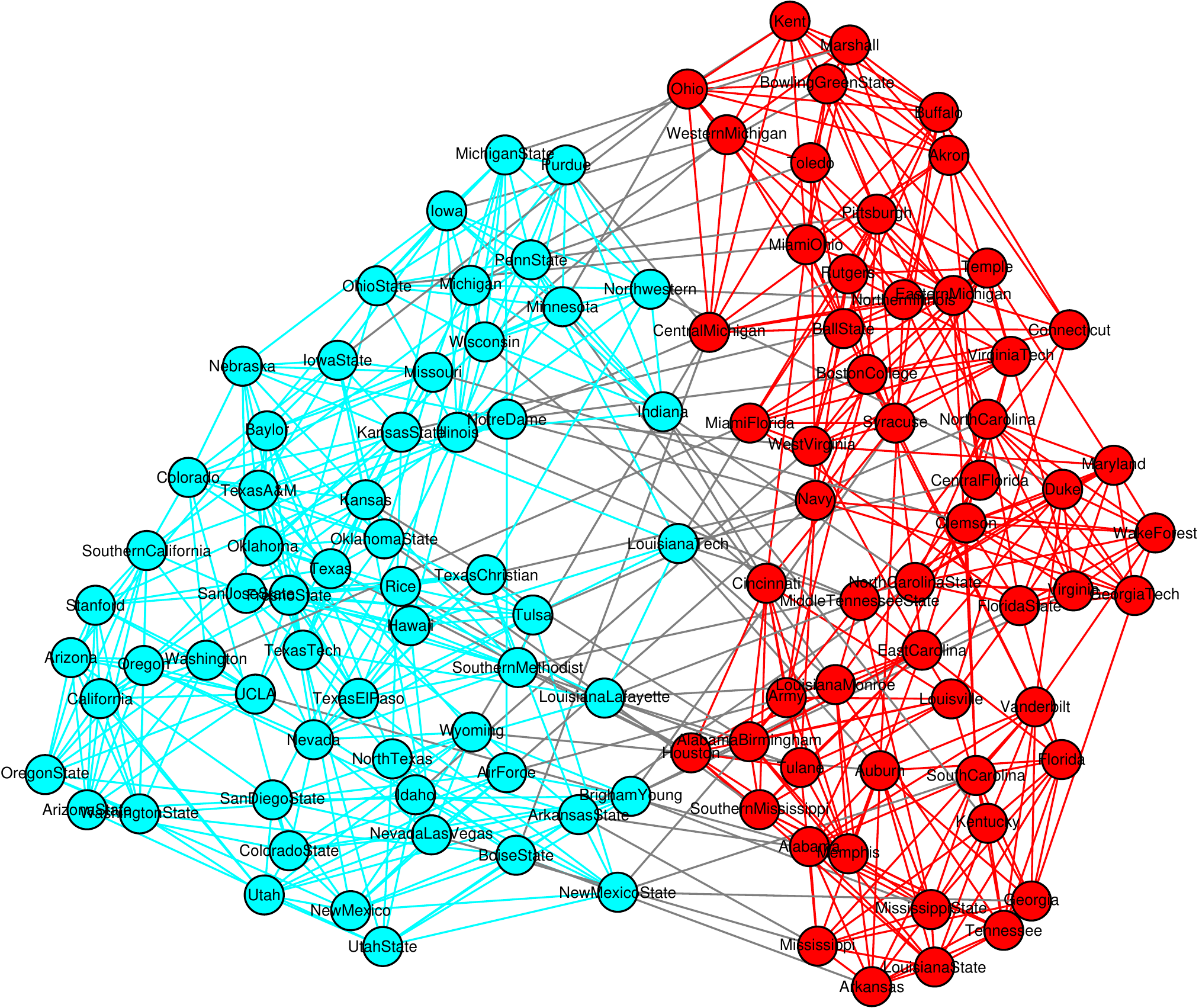}
      \caption{Laplacian}\label{fig:6lap}	
 \end{subfigure}
 \begin{subfigure}[b]{0.32\textwidth}
      \includegraphics[width=0.95\textwidth]{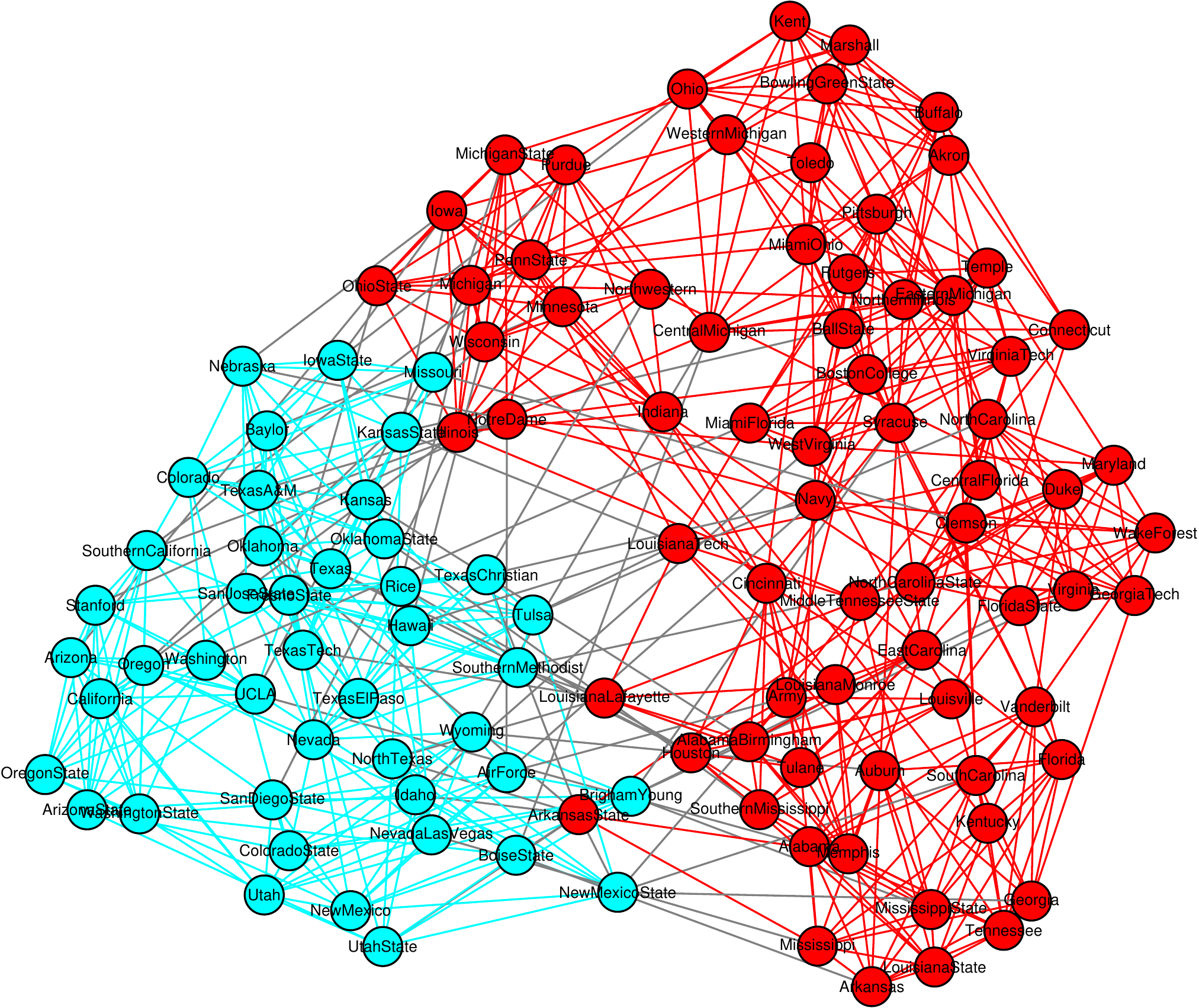}
      \caption{Replicator}\label{fig:6rep}	
 \end{subfigure}\\
 \begin{subfigure}[b]{0.32\textwidth}
      \includegraphics[width=0.95\textwidth]{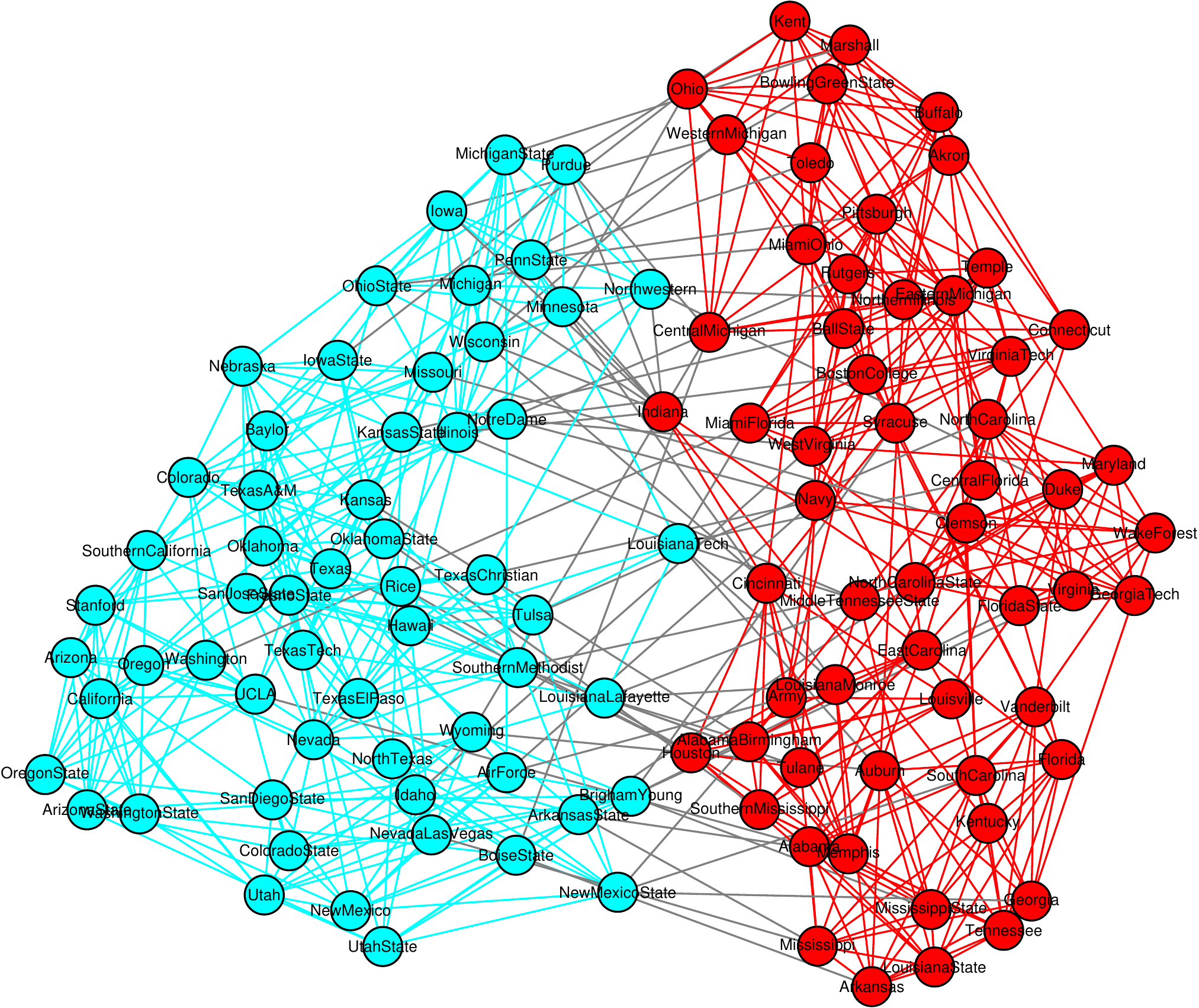}
      \caption{Unbiased Laplacian}\label{fig:6unb}	
 \end{subfigure}
 \begin{subfigure}[b]{0.32\textwidth}
      \includegraphics[width=0.95\textwidth]{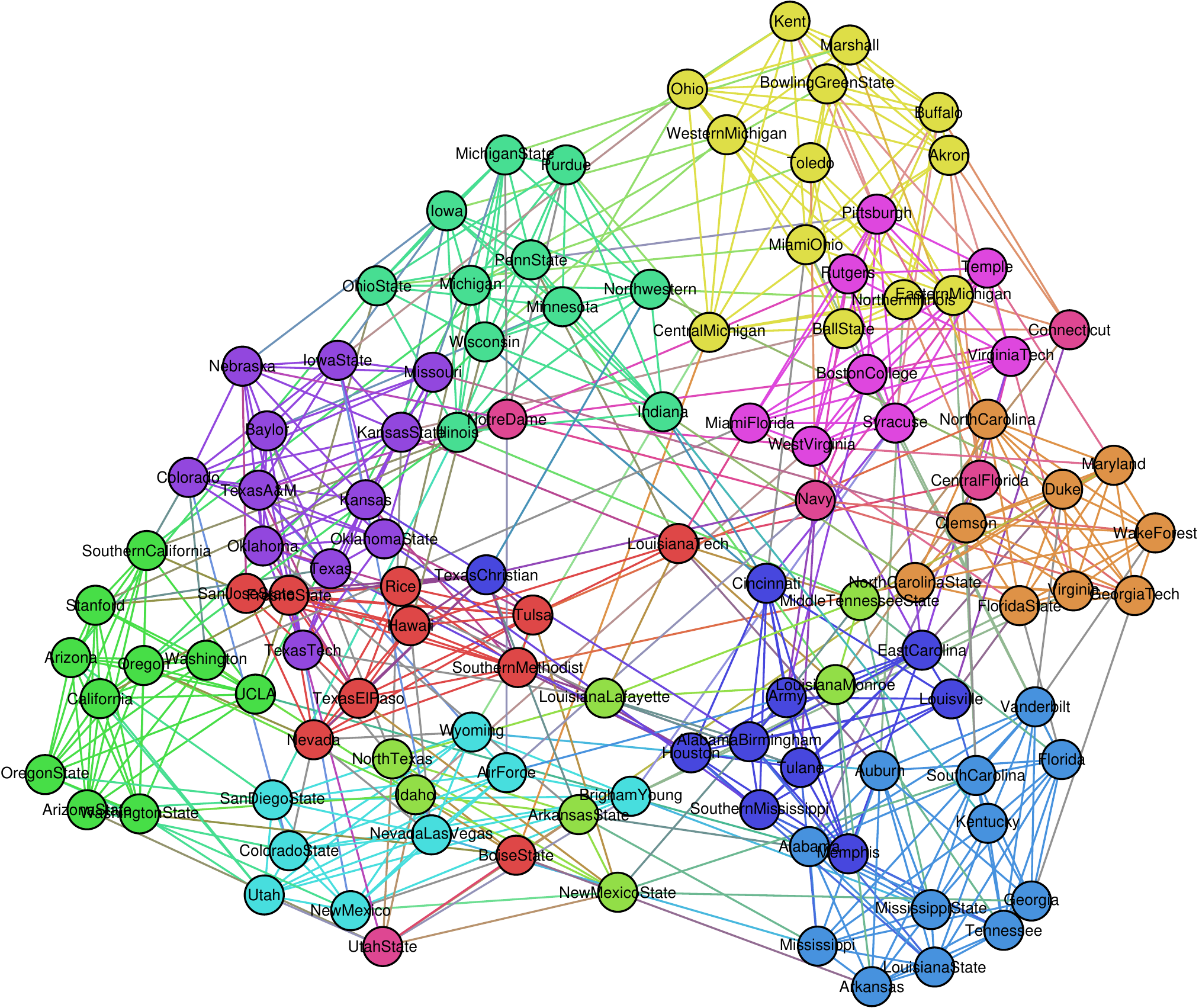}
      \caption{Ground truth communities}\label{fig:6true}	
 \end{subfigure}
  \begin{subfigure}[b]{0.32\textwidth}
      \includegraphics[width=0.95\textwidth]{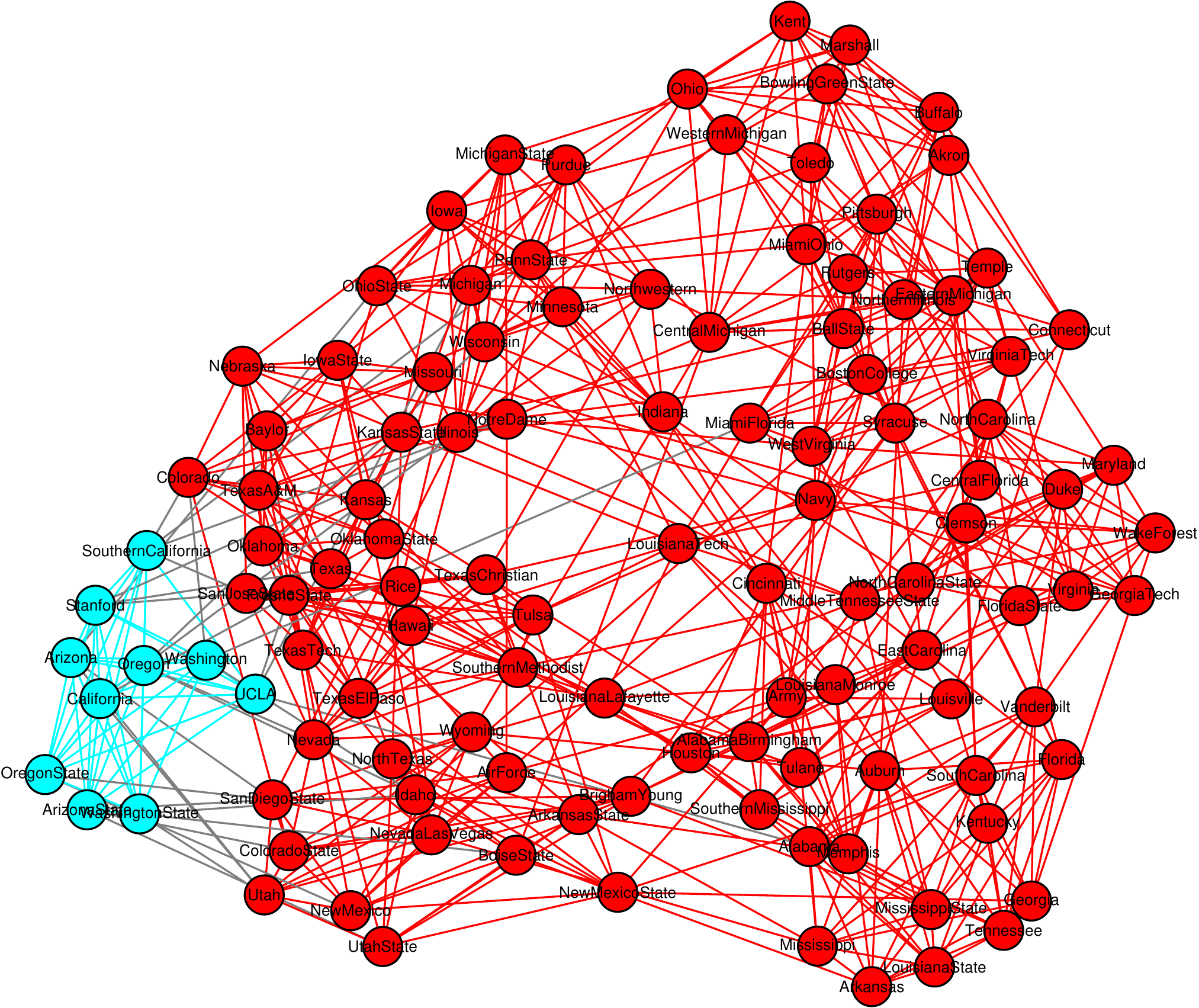}
      \caption{Pacific-10 Conference}\label{fig:6local}	
 \end{subfigure}
 \caption{Centrality/sweep profiles and optimal bisections of the College Football network.}
 \textit{The visualizations are produced by Algorithm~\ref{alg:spec}, corresponding to normalized Laplacian, Laplacian, Replicator and Unbiased Laplacian respectively. Except for the last two visualizations gives the ground truth communities for all 12 local conferences and the Pacific-10 conference in particular.}
\end{figure}

The centrality profiles show heavy tailed distributions, which corresponds to evenly spread out degrees across the network~\figref{fig:6centrality}. This is consistent with the reality of the network, where every football team plays roughly the same number of games each season.

Unlike Zachary's Karate Club, College Football starts to give us divergent community divisions under different dynamic operators. Most operators lead to a balanced east-west bisection (\figref{fig:6norm}, \figref{fig:6lap}, \figref{fig:6unb}). The replicator, however, places the ``swing" Big Ten Conference (contains mostly colleges in the midwest) into the west cluster (\figref{fig:6rep}). Upon further investigation, we discovered that while both bisections have almost the same cross community links, the seemingly more balanced division does lead to a slight imbalance in terms of links within each community. The the generalized centrality under the replicator magnified this imbalance, ultimately pushed the ``swing" conference to the east side.

In fact, the sweep profile \figref{fig:6sweep} clearly shows that all four special cases actually see both bisections as plausible solutions, with closely matched local optima. This phenomenon where different dynamics agrees on multiple local optima but favor different ones as the global solution is a repeating theme in the following examples. Following our observation in the Zachary's Karate Club, it means while different special cases of the generalized conductance can differ in finer details, they will agree on strong community structures that impact all dynamics in similar ways. \figref{fig:6local} further illustrates the point. All four special cases here agree on the first local optimum in the sweep profiles, and this local cluster corresponding to the Pacific 10 conference (it later becomes the Pacific 12).

\subsection{House of Representatives}
The House of Representatives network is built from the 98th United States House of Representatives voting data set \cite{CongressData}. Unlike the previously studied variants, here we use a special version taking account of all 908 votes. The resulting network is densely connected and has an unusually flat degree distribution.  Originally analyzed in an earlier version of the paper by Smith et al. \cite{Smith13spectral}, this network better differentiates some of the dynamics under our framework. Its centrality/sweep profiles and visualizations of optimal bisections under each special case dynamic are given below.

\begin{figure}
 \centering
 \begin{subfigure}[b]{0.45\textwidth}
      \includegraphics[width=0.95\linewidth]{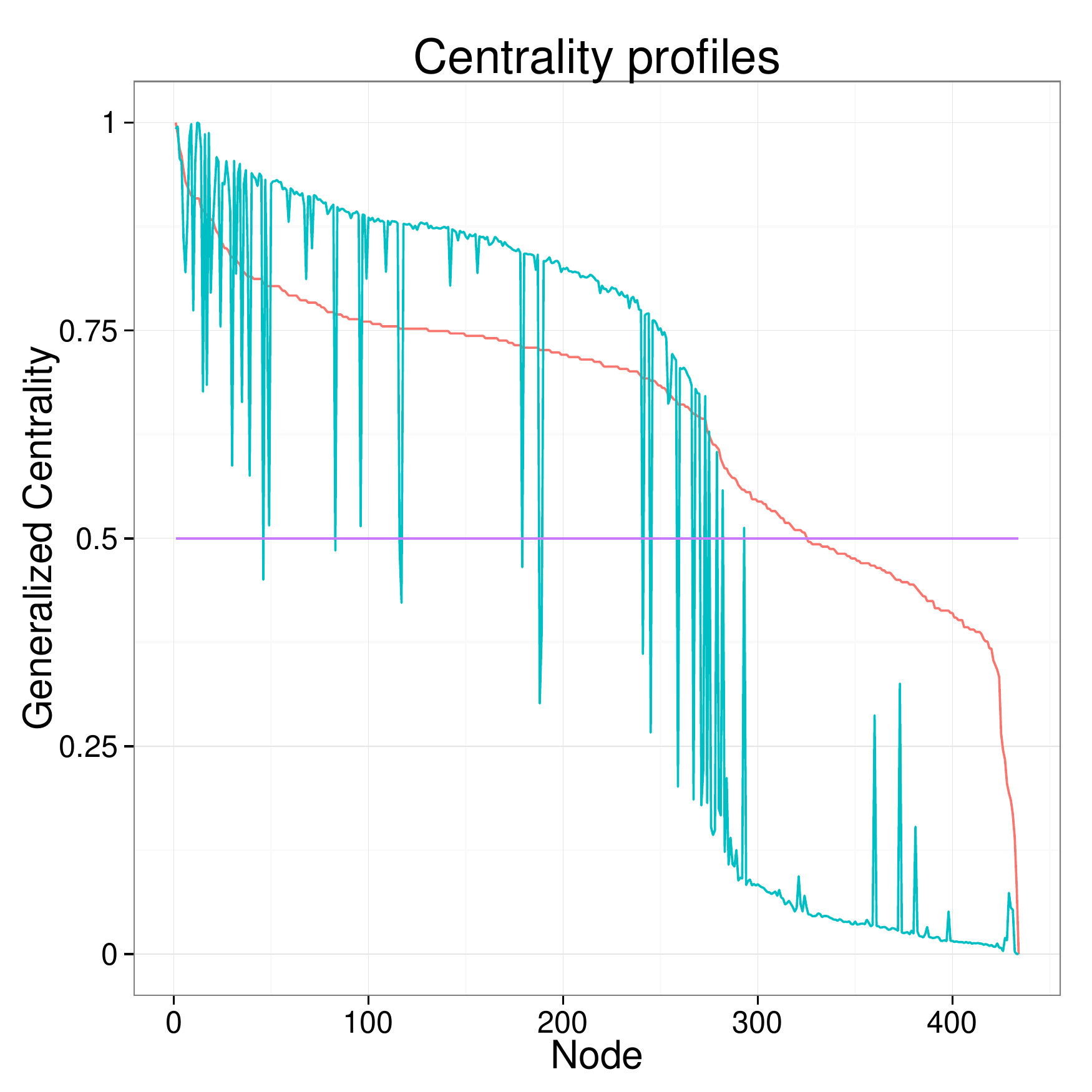}
      \caption{}\label{fig:2centrality}
 \end{subfigure}
 \begin{subfigure}[b]{0.45\textwidth}
      \includegraphics[width=0.95\linewidth]{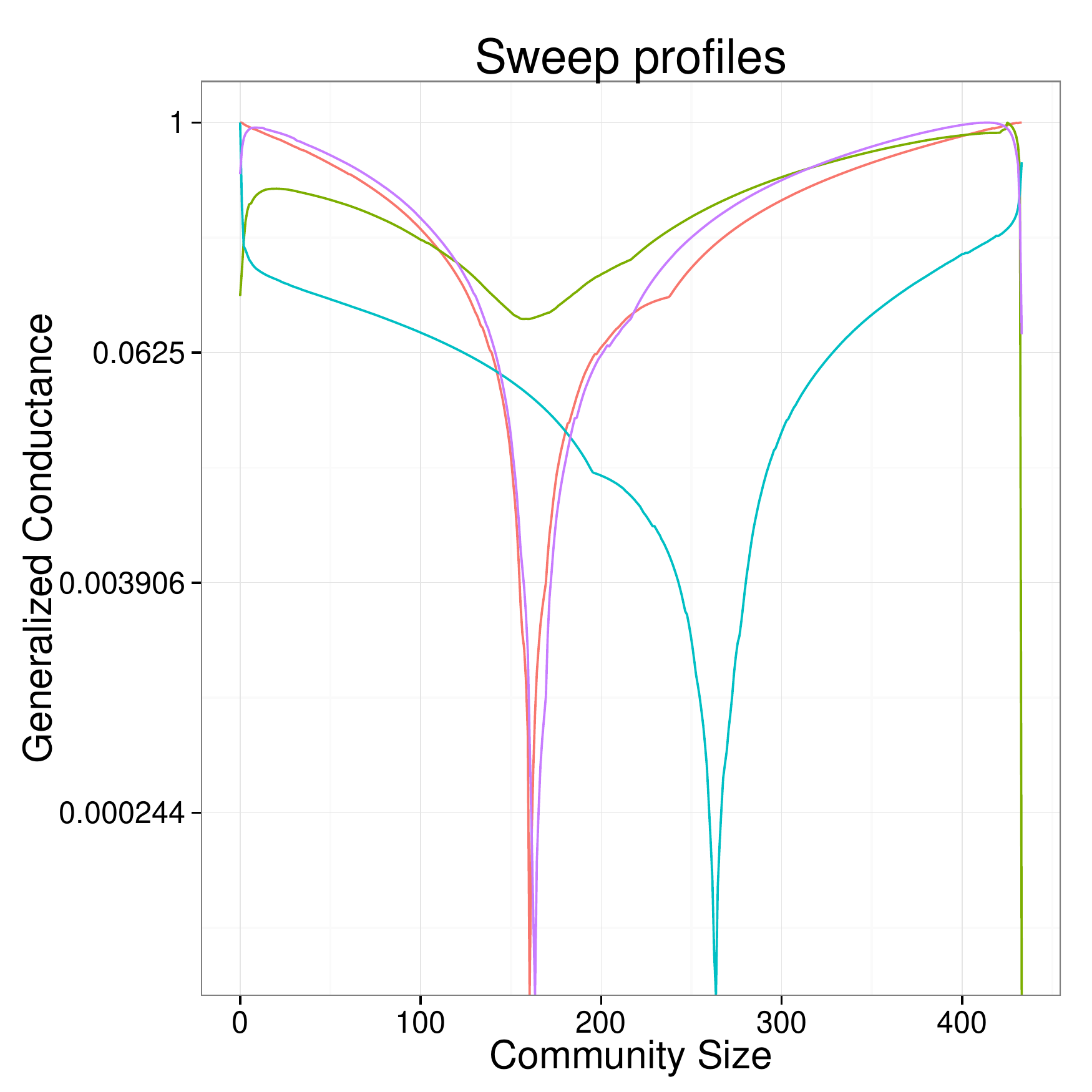}
      \caption{}\label{fig:2sweep}
 \end{subfigure}\\
 \includegraphics[width=0.98\textwidth]{figures/legendPro0}\\
 \begin{subfigure}[b]{0.32\textwidth}
      \includegraphics[width=0.95\textwidth]{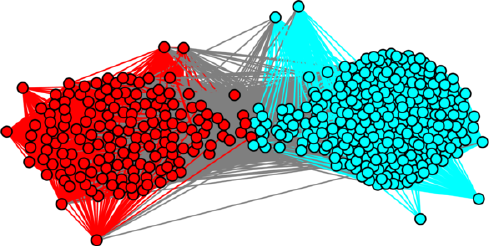}
      \caption{Normalized Laplacian}\label{fig:2norm}	
 \end{subfigure}
 \begin{subfigure}[b]{0.32\textwidth}
      \includegraphics[width=0.95\textwidth]{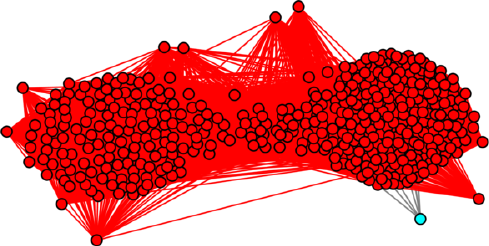}
      \caption{Laplacian}\label{fig:2lap}	
 \end{subfigure}\\
 \begin{subfigure}[b]{0.32\textwidth}
      \includegraphics[width=0.95\textwidth]{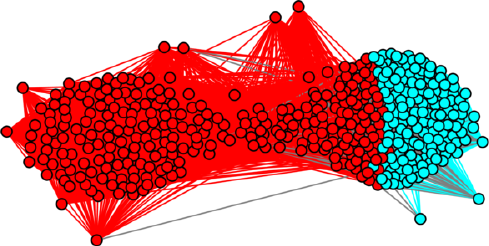}
      \caption{Replicator}\label{fig:2rep}	
 \end{subfigure}
 \begin{subfigure}[b]{0.32\textwidth}
      \includegraphics[width=0.95\textwidth]{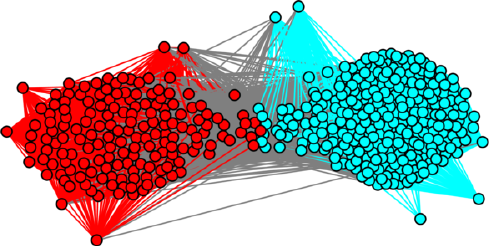}
      \caption{Unbiased Laplacian}\label{fig:2unb}	
 \end{subfigure}
 \begin{subfigure}[b]{0.32\textwidth}
      \includegraphics[width=0.95\textwidth]{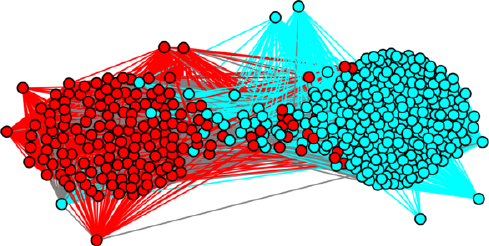}
      \caption{Ground truth communities}\label{fig:2true}	
 \end{subfigure}
 \caption{Centrality/sweep profiles and optimal bisections of the House of Representatives network}
 \textit{The visualizations are produced by Algorithm~\ref{alg:spec}, corresponding to normalized Laplacian, Laplacian, Replicator and Unbiased Laplacian respectively. Except for the last visualization gives the ground truth communities. Notice the mapping between the visualizations and their corresponding sweep profiles.}
\end{figure}

``House of Representatives" network is an excellent example of how centralities and communities are closely related under our framework. First, the centrality profile of the ``House of Representatives" network looks similar to that of the College Football, but quite different from the rest networks in (\tabref{tab:datasets}). Because we have taken into account all the votes, this network is very densely connected, and its degree distribution has an extremely fat tail as demonstrated by the red curve in (\figref{fig:2centrality}).

Since the degree distribution is relatively uniform, we expect variance of the cut size (numerator) in \eqref{eq:hS} to be relatively small. The exception here is the optimal bisection produced by the regular Laplacian (\figref{fig:2lap}), which is most prone to ``whiskers". For the other three special cases, the volume balance (denominator) is the determining factor in communities measures, and all produce fairly ``balanced" bisections according their own generalized volume measures.

Another observation of the centrality profile is that vertices are considered to be of different importance by normalized Laplacian, Replicator and unbiased Laplacian. In particular, centralities of normalized Laplacian might differs from that of unbiased Laplacian by the degree, but given its relative uniform distribution, leads to almost identical optimal bisections (\figref{fig:2norm}, \figref{fig:2unb}). The replicator, on the other hand, scales vertex centrality according to eigenvector centralities, which places more volume to the high degree vertices on the cyan cluster. The resulting optimal bisection is thus shifted right to achieve volume balance (\figref{fig:2rep}).

\subsection{Political Blogs}
The next example is a network of political blogs in the US assembled by Adamic and Glance \cite{adamic2005political}. Here we focus on the giant component, which consists of 1222 blogs and 19087 links between them. The blogs have known political leanings, and were labeled as either liberal or conservative. The network is assortative and has a highly skewed degree distribution. Its centrality/sweep profiles and visualizations of optimal bisections under each special case dynamic are given below.

\begin{figure}
 \centering
 \begin{subfigure}[b]{0.45\textwidth}
      \includegraphics[width=0.95\linewidth]{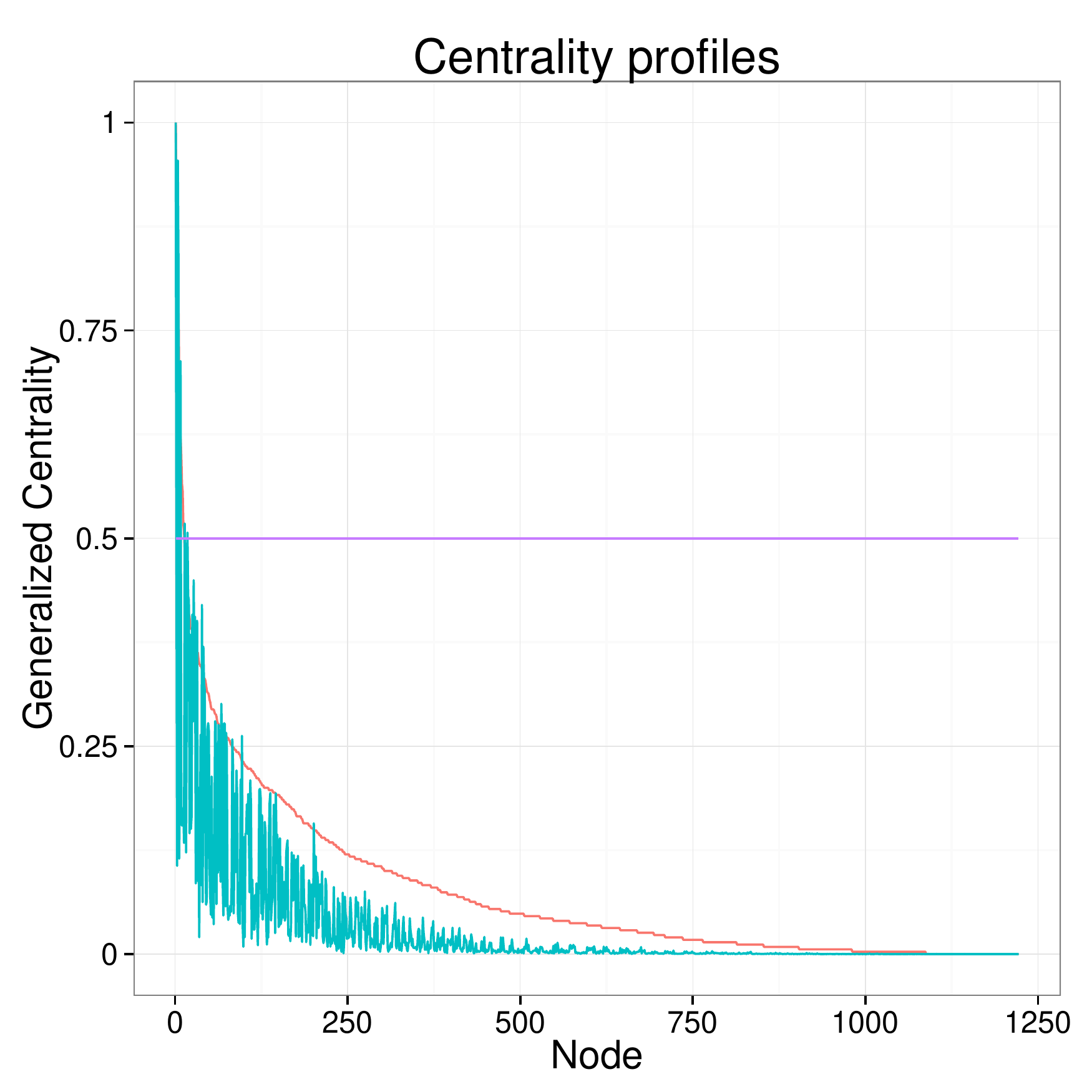}
      \caption{}\label{fig:3centrality}
 \end{subfigure}
 \begin{subfigure}[b]{0.45\textwidth}
      \includegraphics[width=0.95\linewidth]{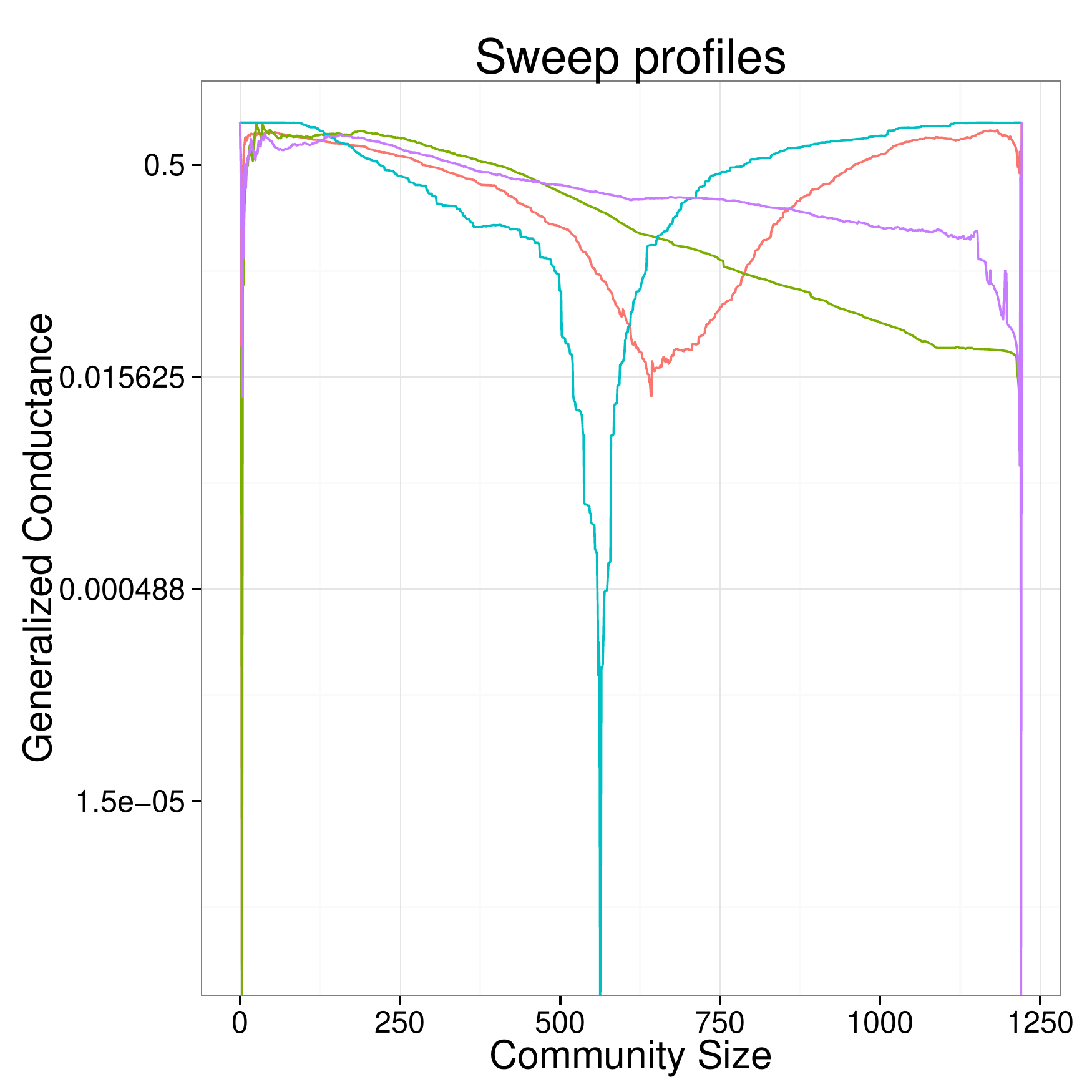}
      \caption{}\label{fig:3sweep}
 \end{subfigure}\\
 \includegraphics[width=0.98\textwidth]{figures/legendPro0}\\
 \begin{subfigure}[b]{0.32\textwidth}
      \includegraphics[width=0.95\textwidth]{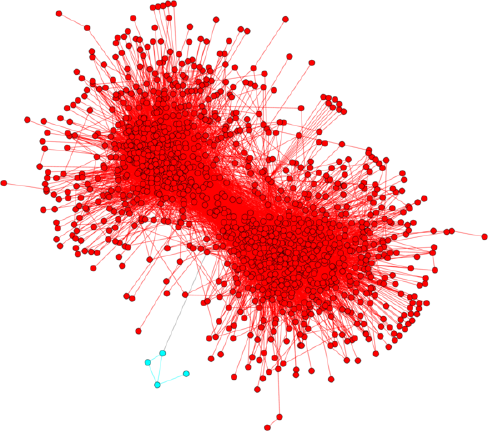}
      \caption{Normalized Laplacian}\label{fig:3norm}	
 \end{subfigure}
 \begin{subfigure}[b]{0.32\textwidth}
      \includegraphics[width=0.95\textwidth]{figures/pBlogsLaplacian-crop}
      \caption{Laplacian}\label{fig:3lap}	
 \end{subfigure}\\
 \begin{subfigure}[b]{0.32\textwidth}
      \includegraphics[width=0.95\textwidth]{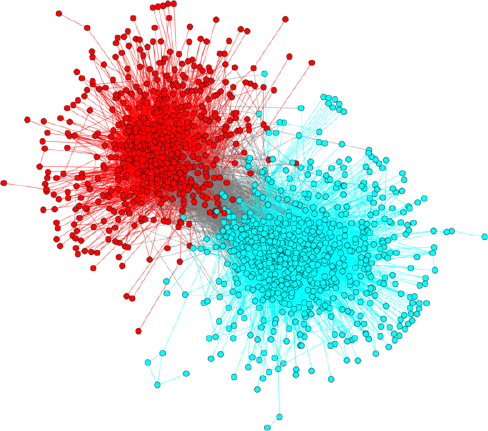}
      \caption{Replicator}\label{fig:3rep}	
 \end{subfigure}
 \begin{subfigure}[b]{0.32\textwidth}
      \includegraphics[width=0.95\textwidth]{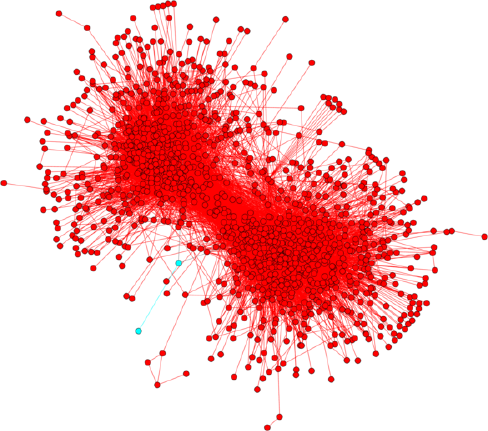}
      \caption{Unbiased Laplacian}\label{fig:3unb}	
 \end{subfigure}
 \begin{subfigure}[b]{0.32\textwidth}
      \includegraphics[width=0.95\textwidth]{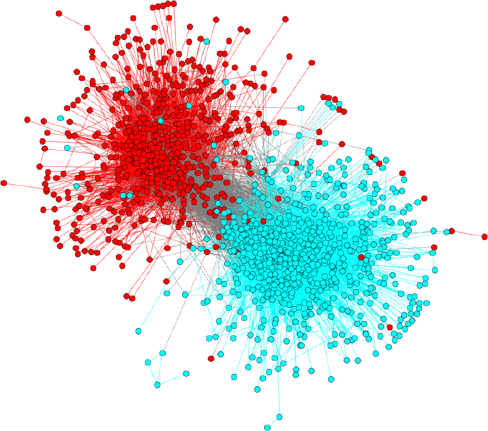}
      \caption{Ground truth communities}\label{fig:3true}	
 \end{subfigure}
 \caption{Centrality/sweep profiles and optimal bisections of the Political Blogs network.}
 \textit{The visualizations are produced by Algorithm~\ref{alg:spec}, corresponding to normalized Laplacian, Laplacian, Replicator and Unbiased Laplacian respectively. Except for the last visualization gives the ground truth communities. Notice the mapping between the visualizations and their corresponding sweep profiles.}
\end{figure}

The Political Blogs network demonstrates a common pitfall of greedy community detection algorithms. A lot of real world networks have power-law like skewed degree distributions, which often corresponds to a ``core-whiskers" structures. As Leskovec et.al. have shown in \cite{Leskovec08www}, such structures have ``whisker" cuts that are so cheap that balance constrains can be effectively ignored. The same happened here for three of our special cases, whose optimal bisections are highly unbalanced.

Unlike the House of Representatives, community measure in Political Blogs is dominated by the cut size (numerator). In particular, both the normalized Laplacian and the Laplacian share the same cut size measures, give the same solution (\figref{fig:3norm}, \figref{fig:3lap}), despite their differences in volume/centrality measures (see curves in \figref{fig:3centrality}). The Unbiased Laplacian produced a different whisker cut, because it has a reweighed cut size measure (\figref{fig:3unb}). Further investigation reveals that the Unbiased Laplacian cuts off a whisker from two highly connected vertices, which according to \eqref{eq:hS} greatly reduces the cut size.

The exception here is the Replicator (\figref{fig:3rep}). By reweighing the adjacency matrix by eigenvector centralities, the generalized volume measure now consider highly connected vertices near the core to be even more important (see the red curve in the centrality profile). The difference in generalized volume is now too drastic to be ignored. As a result, Replicator did not fall for the ``whisker" cuts and produced balanced communities.

\subsection{Facebook Egonets}
The Facebook Egonets dataset was collected using a Facebook app \cite{mcauley2012learning}. Here we use the combined network which merges all egonets of survey participants. Each egonet of a user is defined as local network consists of ``friends" on Facebook, representing the user's social circle. Facebook Egonets has many of the typical social network properties, including a heavy tail degree distribution. However, it also differs from traditional social networks because of the sampling bias in the data collection process, leading to lower clustering coefficient and a bigger diameter than what one might expect. Its centrality/sweep profiles and visualizations of optimal bisections under each special case dynamic are given below.

\begin{figure}
 \centering
 \begin{subfigure}[b]{0.45\textwidth}
      \includegraphics[width=0.95\linewidth]{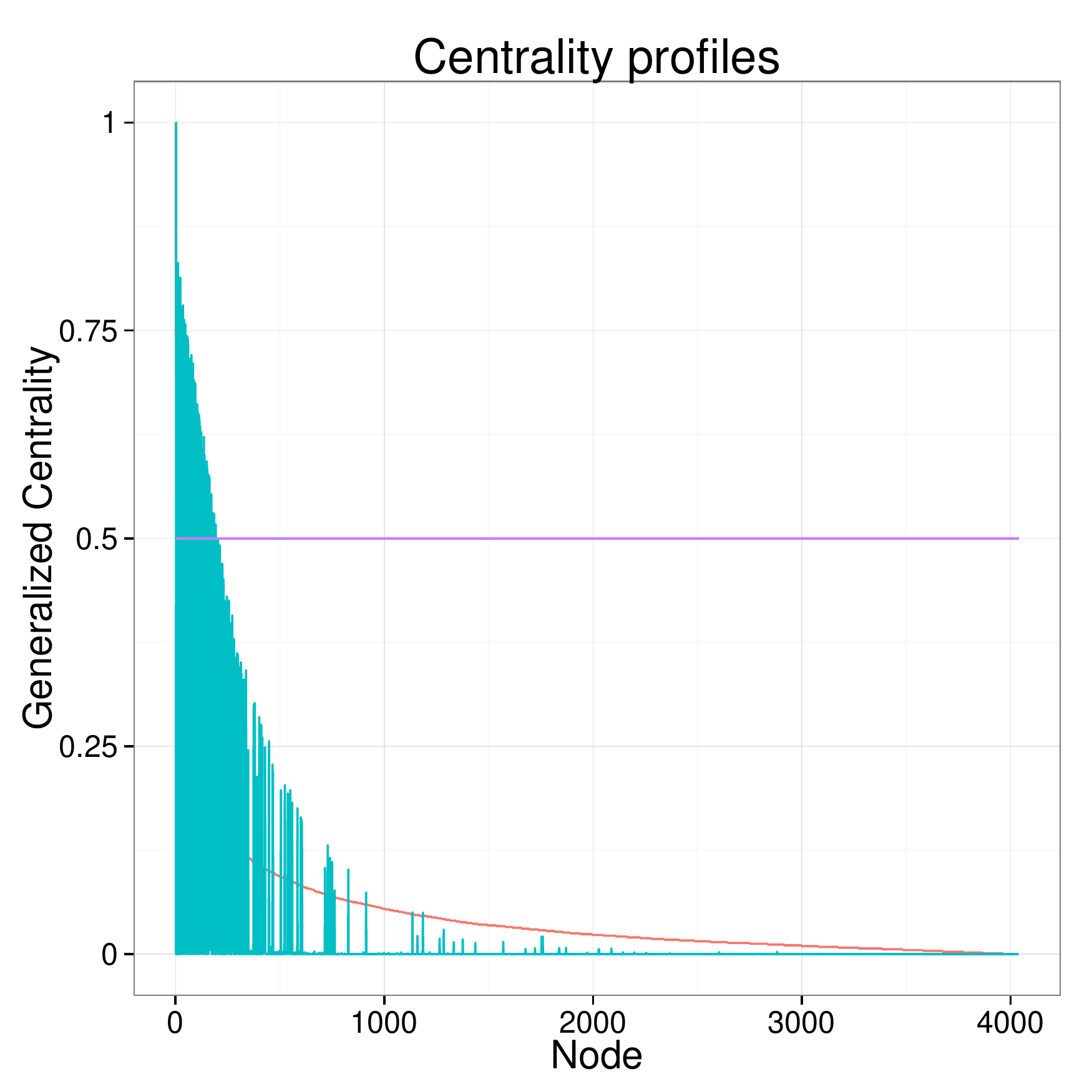}
      \caption{}\label{fig:4centrality}
 \end{subfigure}
 \begin{subfigure}[b]{0.45\textwidth}
      \includegraphics[width=0.95\linewidth]{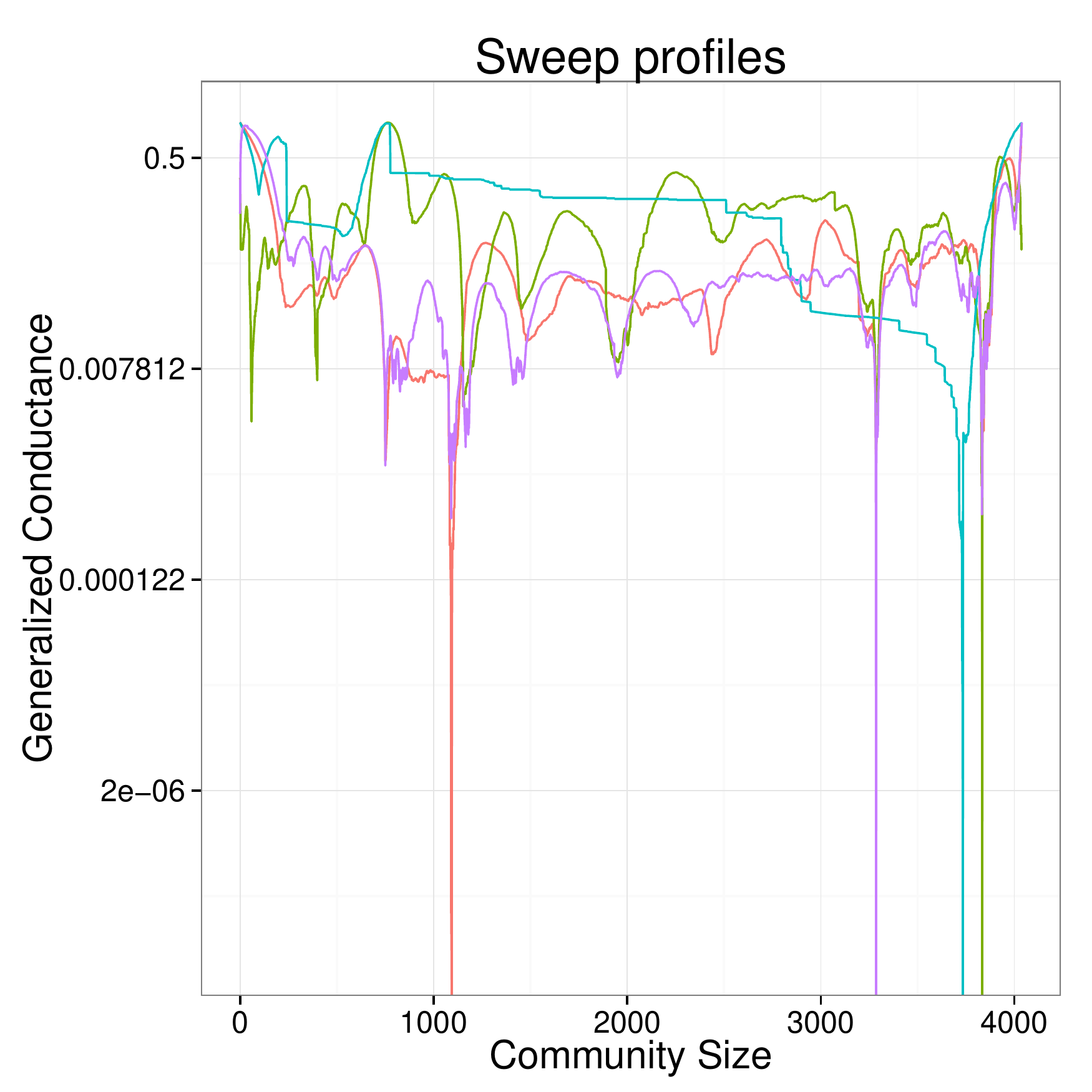}
      \caption{}\label{fig:4sweep}
 \end{subfigure}\\
 \includegraphics[width=0.98\textwidth]{figures/legendPro0}\\
 \begin{subfigure}[b]{0.32\textwidth}
      \includegraphics[width=0.95\textwidth]{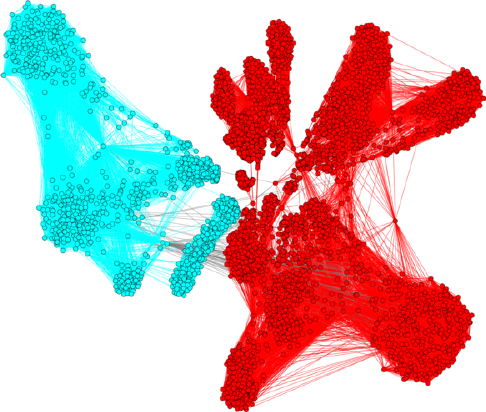}
      \caption{Normalized Laplacian}\label{fig:4norm}	
 \end{subfigure}
 \begin{subfigure}[b]{0.32\textwidth}
      \includegraphics[width=0.95\textwidth]{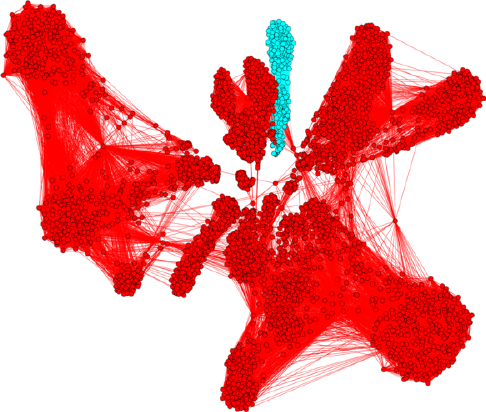}
      \caption{Laplacian}\label{fig:4lap}	
 \end{subfigure}\\
 \begin{subfigure}[b]{0.32\textwidth}
      \includegraphics[width=0.95\textwidth]{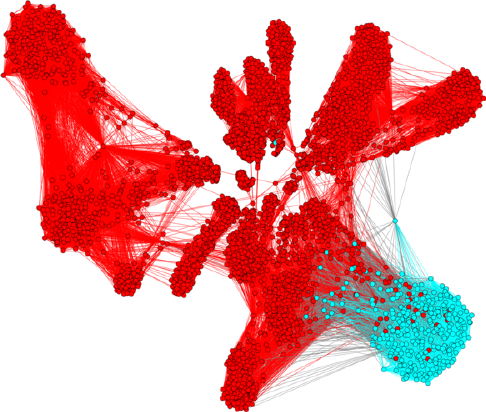}
      \caption{Replicator}\label{fig:4rep}	
 \end{subfigure}
 \begin{subfigure}[b]{0.32\textwidth}
      \includegraphics[width=0.95\textwidth]{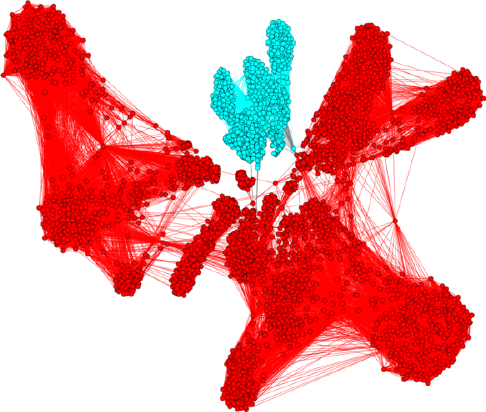}
      \caption{Unbiased Laplacian}\label{fig:4unb}	
 \end{subfigure}
 \caption{Centrality/sweep profiles and optimal bisections of the Facebook Egonets network.}
 \textit{The visualizations are produced by Algorithm~\ref{alg:spec}, corresponding to normalized Laplacian, Laplacian, Replicator and Unbiased Laplacian respectively. There is no predetermined ground truth communities in this network. Notice the mapping between the visualizations and their corresponding sweep profiles.}
\end{figure}

Like what happened with Political Blogs, the overall multi-core structure leads to unbalanced bisections. Due to its bigger size and a even more heterogeneous degree distribution (\figref{fig:4centrality}), all four special cases under the generalized Laplacian framework fall for local clusters, each in a different fashion. Again, the regular Laplacian, being the most susceptible to such problems, finds a smallest local community with the minimal cut size of 17 links (\figref{fig:4lap}). In contrast, the unbiased Laplacian which has the same volume measures, finds a superset of vertices as the optimal cut, with 40 inter community edges (\figref{fig:4unb}). The normalized Laplacian measures cut sizes the same way as the Laplacian, but its different volume measure leads to a much more balanced cut (\figref{fig:4norm}). Last but not least, the replicator finds a local core structure with an average degree of $85.7$ (\figref{fig:4rep}). This is consistent with what we observed on House of Representatives, where the
eigenvector centrality places more volume to the cyan cluster, and the resulting cut is actually much more balanced than it looks.

\subsection{Power Grid}
The last example is an undirected, unweighted network representing the topology of the Western States Power Grid of the United States \cite{watts1998collective}. Among the six datasets in \tabref{tab:datasets}, Power Grid is the largest network in terms of the number of vertices. However, it is extremely sparse with an average degree of $2.67$, leading to a homogeneous connecting pattern across the whole network without real cores or whiskers. Its centrality/sweep profiles and visualizations of optimal bisections under each special case dynamic are given below.

\begin{figure}
 \centering
 \begin{subfigure}[b]{0.45\textwidth}
      \includegraphics[width=0.95\linewidth]{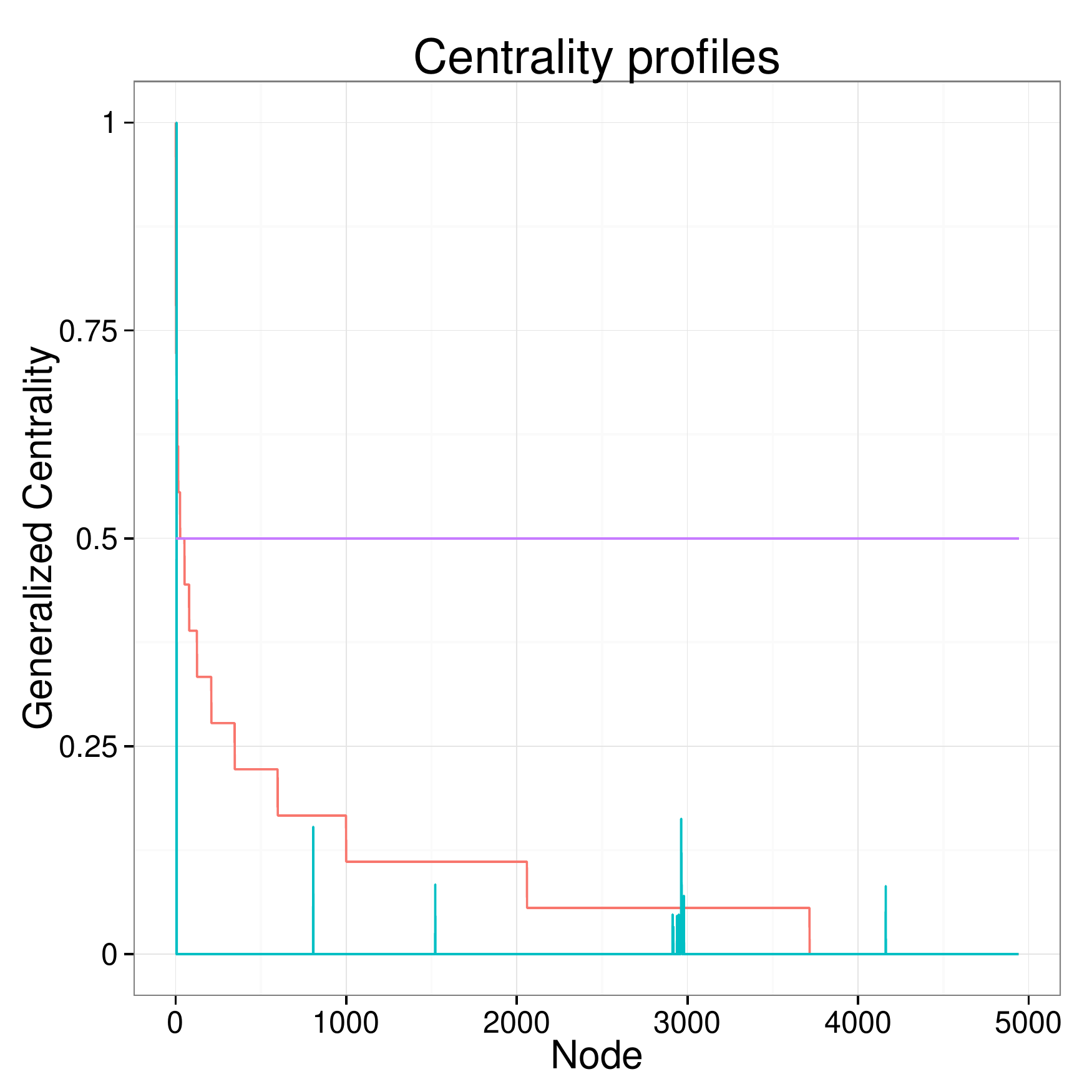}
      \caption{}\label{fig:5centrality}
 \end{subfigure}
 \begin{subfigure}[b]{0.45\textwidth}
      \includegraphics[width=0.95\linewidth]{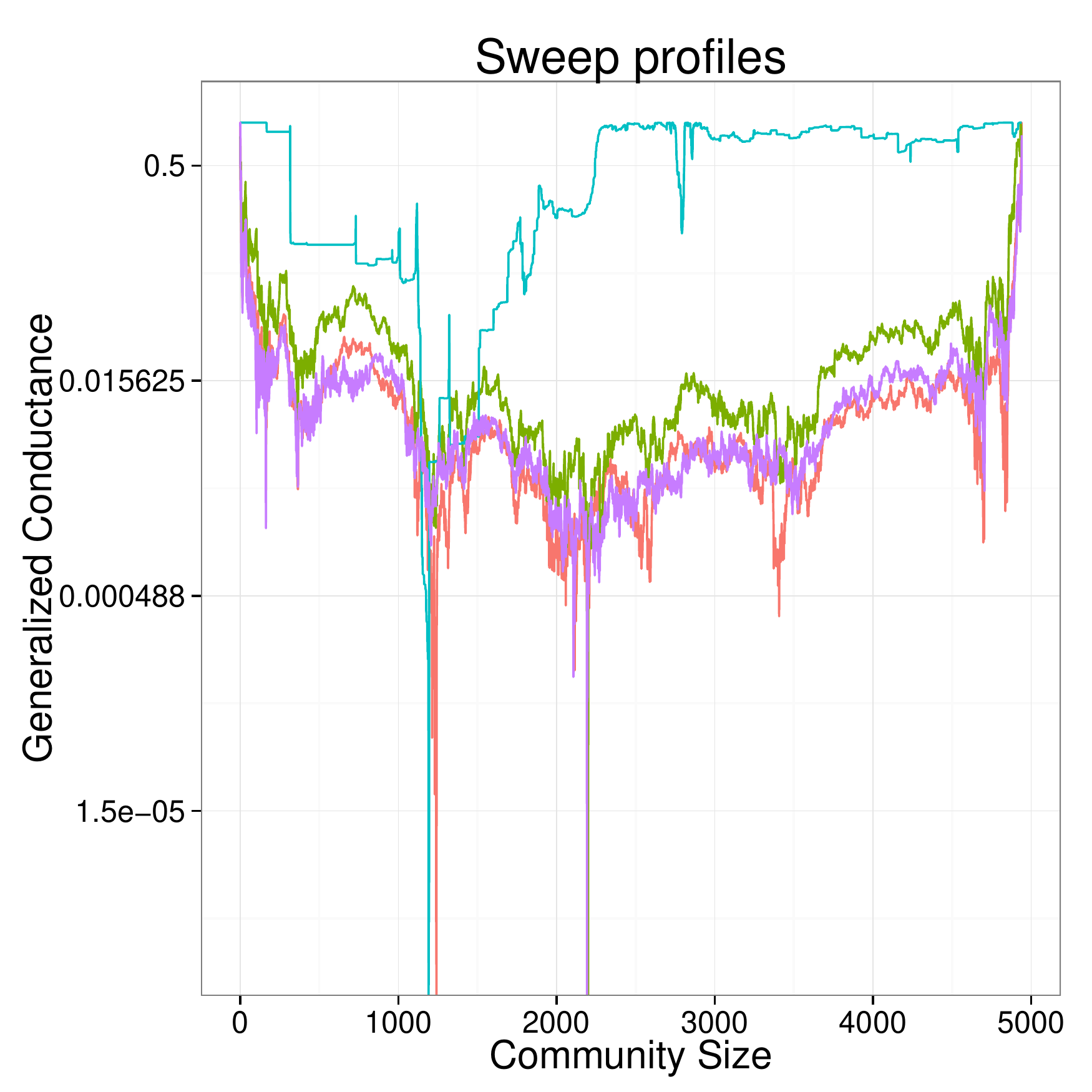}
      \caption{}\label{fig:5sweep}
 \end{subfigure}\\
 \includegraphics[width=0.98\textwidth]{figures/legendPro0}\\
 \begin{subfigure}[b]{0.32\textwidth}
      \includegraphics[width=0.95\textwidth]{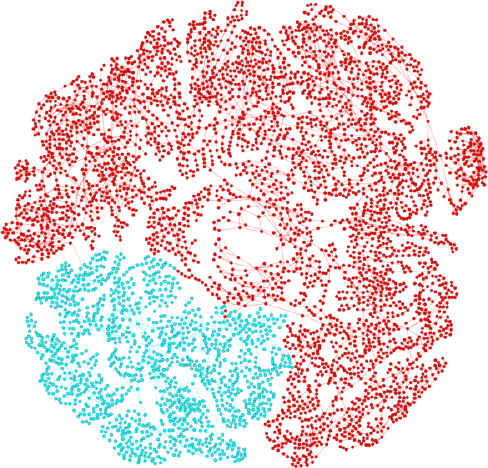}
      \caption{Normalized Laplacian}\label{fig:5norm}	
 \end{subfigure}
 \begin{subfigure}[b]{0.32\textwidth}
      \includegraphics[width=0.95\textwidth]{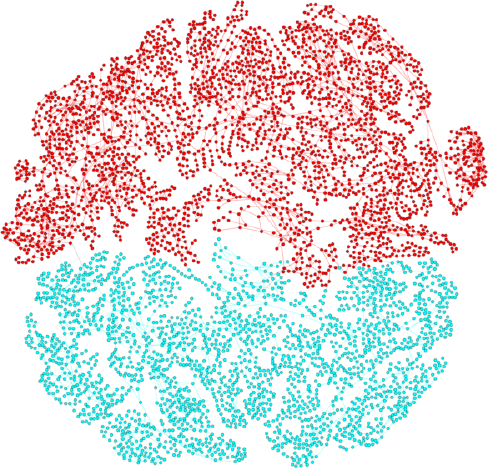}
      \caption{Laplacian}\label{fig:5lap}	
 \end{subfigure}\\
 \begin{subfigure}[b]{0.32\textwidth}
      \includegraphics[width=0.95\textwidth]{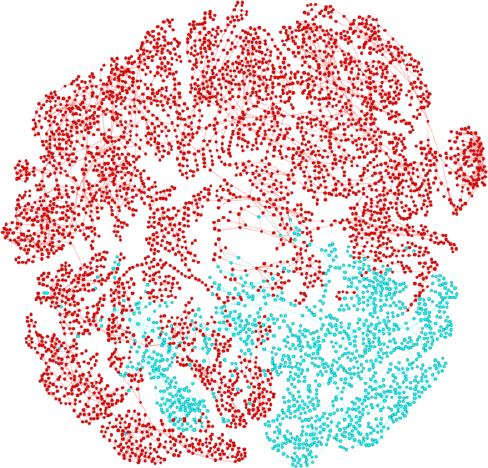}
      \caption{Replicator}\label{fig:5rep}	
 \end{subfigure}
 \begin{subfigure}[b]{0.32\textwidth}
      \includegraphics[width=0.95\textwidth]{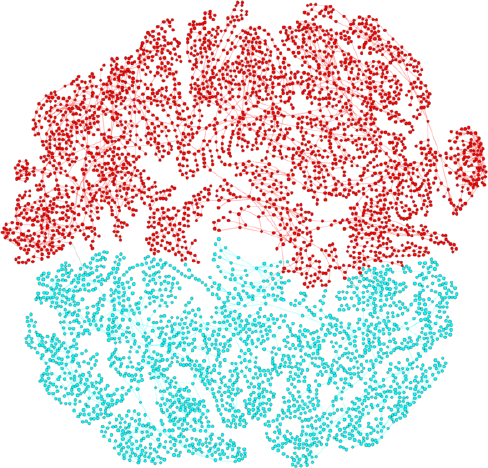}
      \caption{Unbiased Laplacian}\label{fig:5unb}	
 \end{subfigure}
 \caption{Centrality/sweep profiles and optimal bisections of the Power Grid network.}
 \textit{The visualizations are produced by Algorithm~\ref{alg:spec}, corresponding to normalized Laplacian, Laplacian, Replicator and Unbiased Laplacian respectively. There is no predetermined ground truth communities in this network. Notice the mapping between the visualizations and their corresponding sweep profiles.}
\end{figure}

The fat tails in centrality profiles indicate existence of high degree vertices \figref{fig:5centrality}. However, as the visualizations shows, these hub vertices do not usually link each other directly, resulting in negative degree assortativity \cite{newman_mixing_2003}. This is consistent with the geographic constrains when designing a power grid, as the final goal is to distribute power from central stations to end users. These important difference in overall structure prevented core or whiskers from appearing, and changes how different dynamics behave on Power Grid.

Replicator, which demonstrated the most consistent performance on social networks with core-whiskers structures, performs the worst on bisecting the Power Grid. In fact, the visualization shown in \figref{fig:5rep} is obtained by manually fixing negative eigenvector centrality entries in~\eqref{eq:rep} (the numeric error comes from the extreme sparse and ill-conditioned adjacency matrix).

The other three special cases all give reasonable results. Laplacian and Unbiased Laplacian share the same volume measure, and they have nearly identical solutions with well balanced communities (\figref{fig:5lap}, \figref{fig:5unb}). Their different cut size measures only lead to slightly different boundaries thanks to the homogeneous connecting pattern. Normalized Laplacian share the same cut size measure with the regular Laplacian, and its volume balance is usually more robust on social networks with core-whisker structures. On Power Grid, however, it opts for a smaller cut size at the cost of volume imbalance (\figref{fig:5norm}). It turns out the volume of the cyan cluster is compensated by its relative high average degree.

\section{Conclusion}
The generalized Laplacian framework presented in this paper can describe a variety of different dynamic processes taking place on a network, including random walks and epidemics, but also new ones, such as the unbiased Laplacian. We extended the relationships between the properties of centrality, conductance and the dynamic operator normalized Laplacian, to this more general class of dynamic processes. Each dynamic process leads to a distribution that  gives centrality of vertices with respect to that process. In addition, we show that the generalized conductance with respect to the dynamic process is related to the eigenvalues of the operator describing that process through a Cheeger-like inequality. We used these relationships to develop efficient algorithms for global community detection, in which vertices within the same partition interact more with each other via the dynamic process than with vertices in other partitions.

The generalized Laplacian framework also provides a systematic tool to analyze and compare different dynamic processes. By making the dynamic process explicit, we gain new insights to existing centrality measures and community quality measures. By connecting them using standard linear transformations, we discovered the equivalence among seemingly different dynamical systems, and we also have a better understanding of the different community structures each of them produces. In the future, we plan to investigate their differences based on how the vertex state variables change during the evolution of the dynamics. In the analysis of massive networks, it is also desirable to identify subsets of vertices whose induced sub-graphs have ``enough"  community structures without examining the entire network. Chung \cite{Chung07pnas, ChungLocalHeatKernal} derived a local version of the Cheeger's like inequality to the heat-kernel page rank to identify random walk-based local clusters. Similarly, our framework can be
adapted to such local clustering procedures.

While our framework is flexible enough to represent several important types of dynamic processes, it does not represent all possible processes, for example, those processes that even after a change of basis, do not conserve the total volume during the dynamical process. In order to describe such dynamics, an even more general framework is needed. We conjecture, however, that the more general operators will still obey the Cheeger-like inequality, and that other theorems presented in this paper can be extended to these types of processes.

\subsection*{Acknowledgments}
This work is partly supported by grants NSF CIF-1217605,  AFOSR-MURI  FA9550-10-1-0569, AFRL FA-8750-12-2-0186, DARPA W911NF-12-1-0034, NSF CCF-0964481, and  NSF CCF-1111270.

\bibliographystyle{nws}
\bibliography{rumig-synchronization,references}

\begin{thebibliography}{}

\bibitem[\protect\citename{Adamic \& Glance, }2005]{adamic2005political}
Adamic, Lada~A, \& Glance, Natalie. (2005).
\newblock {The political blogosphere and the 2004 US election: divided they
  blog}.
\newblock {\em Pages  36--43 of:} {\em {Proceedings of the 3rd international
  workshop on Link discovery}}.
\newblock ACM.

\bibitem[\protect\citename{Andersen {\em et~al.}\relax,
  }2007]{AndersenChungLang}
Andersen, R., Chung, F., \& Lang, K. (2007).
\newblock {Using {PageRank} to Locally Partition a Graph}.
\newblock {\em Internet math.}, {\bf 4}(1), 1--128.

\bibitem[\protect\citename{Andersen \& Peres, }2009]{AndersenPeres}
Andersen, Reid, \& Peres, Yuval. (2009).
\newblock {Finding Sparse Cuts Locally Using Evolving Sets}.
\newblock {\em Pages  235--244 of:} {\em {STOC}}.
\newblock ACM.

\bibitem[\protect\citename{{Arenas} {\em et~al.}\relax, }2006]{Arenas2006Sync}
{Arenas}, A., {D{\'i}az-Guilera}, A., \& {P{\'e}rez-Vicente}, C.~J. (2006).
\newblock {Synchronization Reveals Topological Scales in Complex Networks}.
\newblock {\em Physical review letters}, {\bf 96}(11), 114102.

\bibitem[\protect\citename{Bastian {\em et~al.}\relax, }2009]{gephi}
Bastian, Mathieu, Heymann, Sebastien, \& Jacomy, Mathieu. (2009).
\newblock {\em {Gephi: An Open Source Software for Exploring and Manipulating
  Networks}}.

\bibitem[\protect\citename{Bonacich \& Lloyd, }2001]{Bonacich01}
Bonacich, Phillip, \& Lloyd, Paulette. (2001).
\newblock {Eigenvector-like measures of centrality for asymmetric relations}.
\newblock {\em Social networks}, {\bf 23}(3), 191--201.

\bibitem[\protect\citename{Borgatti, }2005]{Borgatti05}
Borgatti, S. (2005).
\newblock {Centrality and network flow}.
\newblock {\em Social networks}, {\bf 27}(1), 55--71.

\bibitem[\protect\citename{Burda {\em et~al.}\relax,
  }2009]{burda_localization_2009}
Burda, Z., Duda, J., Luck, J.~M., \& Waclaw, B. (2009).
\newblock {Localization of the Maximal Entropy Random Walk}.
\newblock {\em Physical review letters}, {\bf 102}(16), 160602.

\bibitem[\protect\citename{Chung, }2007]{Chung07pnas}
Chung, Fan. (2007).
\newblock {The heat kernel as the pagerank of a graph}.
\newblock {\em Proceedings of the national academy of sciences}, {\bf 104}(50),
  19735--19740.

\bibitem[\protect\citename{Chung, }2009]{ChungLocalHeatKernal}
Chung, Fan. (2009).
\newblock {A Local Graph Partitioning Algorithm Using Heat Kernel Pagerank}.
\newblock {\em Internet mathematics}, {\bf 6}(3), 315--330.

\bibitem[\protect\citename{Chung, }1997]{Chung1997Spectral}
Chung, Fan R.~K. (1997).
\newblock {\em {Spectral Graph Theory (CBMS Regional Conference Series in
  Mathematics, No. 92)}}.
\newblock American Mathematical Society.

\bibitem[\protect\citename{Delvenne {\em et~al.}\relax,
  }2008]{delvenne_stability_2008}
Delvenne, J.-C., Yaliraki, S.~N., \& Barahona, M. (2008).
\newblock {Stability of graph communities across time scales}.
\newblock {\em {ArXiv} e-prints}, Dec.

\bibitem[\protect\citename{Fortunato, }2010]{Fortunato10}
Fortunato, Santo. (2010).
\newblock {Community detection in graphs}.
\newblock {\em Physics reports}, {\bf 486}(Jan.), 75--174.

\bibitem[\protect\citename{Fronczak \& Fronczak, }2009]{fronczak_biased_2009}
Fronczak, Agata, \& Fronczak, Piotr. (2009).
\newblock {Biased random walks in complex networks: The role of local
  navigation rules}.
\newblock {\em Physical review e}, {\bf 80}(1).

\bibitem[\protect\citename{Ghosh \& Lerman, }2011]{Ghosh11physrev}
Ghosh, Rumi, \& Lerman, Kristina. (2011).
\newblock {Parameterized centrality metric for network analysis}.
\newblock {\em Physical review e}, {\bf 83}(6), 066118.

\bibitem[\protect\citename{Ghosh \& Lerman, }2012]{ghosh_rethinking_2012}
Ghosh, Rumi, \& Lerman, Kristina. (2012).
\newblock {Rethinking Centrality: The Role of Dynamical Processes in Social
  Network Analysis}.
\newblock {\em {CoRR}}, {\bf abs/1209.4616}.

\bibitem[\protect\citename{Girvan \& Newman, }2002]{GirvanBetween}
Girvan, M., \& Newman, M. E.~J. (2002).
\newblock {Community structure in social and biological networks}.
\newblock {\em Proceedings of the national academy of sciences}, {\bf 99}(12),
  7821--7826.

\bibitem[\protect\citename{G{\'o}mez-Garde{\~n}es \& Latora,
  }2008]{gomez-gardenes_entropy_2008}
G{\'o}mez-Garde{\~n}es, J., \& Latora, V. (2008).
\newblock {Entropy rate of diffusion processes on complex networks}.
\newblock {\em Physical review e}, {\bf 78}(6), 065102.

\bibitem[\protect\citename{Kannan {\em et~al.}\relax, }2004]{KananVampelaVetta}
Kannan, Ravi, Vempala, Santosh, \& Vetta, Adrian. (2004).
\newblock {On Clusterings: Good, Bad and Spectral}.
\newblock {\em J. acm}, {\bf 51}(3), 497--515.

\bibitem[\protect\citename{Katz, }1953]{katz1953}
Katz, Leo. (1953).
\newblock {A new status index derived from sociometric analysis}.
\newblock {\em Psychometrika}, {\bf 18}(1), 39--43.

\bibitem[\protect\citename{Kempe {\em et~al.}\relax, }2003]{Kempe03}
Kempe, David, Kleinberg, Jon, \& Tardos, Eva. (2003).
\newblock {Maximizing the spread of influence through a social network}.
\newblock {\em Pages  137--146 of:} {\em {KDD '03}}.
\newblock ACM.

\bibitem[\protect\citename{Koutis {\em et~al.}\relax, }2010]{KoutisMillerPeng}
Koutis, Ioannis, Miller, Gary~L., \& Peng, Richard. (2010).
\newblock {Approaching Optimality for Solving SDD Linear Systems}.
\newblock {\em Pages  235--244 of:} {\em {FOCS}}.
\newblock IEEE.

\bibitem[\protect\citename{Krause, }2008]{krause_compromise_2008}
Krause, Ulrich. (2008).
\newblock {Compromise, consensus, and the iteration of means}.
\newblock {\em Elemente der mathematik}, {\bf 63}, 1--8.

\bibitem[\protect\citename{Lambiotte {\em et~al.}\relax,
  }2011]{lambiotte_flow_2011}
Lambiotte, R., Sinatra, R., Delvenne, J.-C., Evans, T.~S., Barahona, M., \&
  Latora, V. (2011).
\newblock {Flow graphs: Interweaving dynamics and structure}.
\newblock {\em Physical review e}, {\bf 84}(1), 017102.

\bibitem[\protect\citename{Lerman \& Ghosh, }2012]{Lerman12pre}
Lerman, Kristina, \& Ghosh, Rumi. (2012).
\newblock {Network Structure, Topology and Dynamics in Generalized Models of
  Synchronization}.
\newblock {\em Physical review e}, {\bf 86}(026108).

\bibitem[\protect\citename{Leskovec {\em et~al.}\relax, }2008]{Leskovec08www}
Leskovec, Jure, Lang, Kevin~J., Dasgupta, Anirban, \& Mahoney, Michael~W.
  (2008).
\newblock {Statistical Properties of Community Structure in Large Social and
  Information Networks}.

\bibitem[\protect\citename{Ling {\em et~al.}\relax, }2013]{ling_effects_2013}
Ling, Xiang, Hu, Mao-Bin, Ding, Jian-Xun, Shi, Qing, \& Jiang, Rui. (2013).
\newblock {Effects of target routing model on the occurrence of extreme events
  in complex networks}.
\newblock {\em The european physical journal b}, {\bf 86}(4).

\bibitem[\protect\citename{Lov{\'a}sz, }1993]{LovaszS}
Lov{\'a}sz, L{\'a}szl{\'o}. (1993).
\newblock {\em {Random Walks on Graphs: A Survey}}.
\newblock Pages  353--397.

\bibitem[\protect\citename{McAuley \& Leskovec, }2012]{mcauley2012learning}
McAuley, Julian~J, \& Leskovec, Jure. (2012).
\newblock {Learning to Discover Social Circles in Ego Networks.}
\newblock {\em Pages  548--556 of:} {\em {NIPS}},  vol. 272.

\bibitem[\protect\citename{Mihail, }1989]{mihail89}
Mihail, M. 1989 (Oct).
\newblock {Conductance and convergence of Markov chains-a combinatorial
  treatment of expanders}.
\newblock {\em Pages  526--531 of:} {\em {Foundations of Computer Science,
  1989., 30th Annual Symposium on}}.

\bibitem[\protect\citename{{Motter} {\em et~al.}\relax, }2005]{Motter2005Sync}
{Motter}, A.~E., {Zhou}, C., \& {Kurths}, J. (2005).
\newblock {Network synchronization, diffusion, and the paradox of
  heterogeneity}.
\newblock {\em Physical review e}, {\bf 71}(1), 016116.

\bibitem[\protect\citename{Newman, }2003]{newman_mixing_2003}
Newman, M.~E. (2003).
\newblock {Mixing patterns in networks}.
\newblock {\em {\textbackslash}pre}, {\bf 67}(2), 026126.

\bibitem[\protect\citename{Newman, }2006]{Newman2006}
Newman, M. E.~J. (2006).
\newblock {Finding community structer in networks using the eigenvectors of
  matrices}.
\newblock {\em Physical review e}, {\bf 74}(3).

\bibitem[\protect\citename{Olfati-Saber {\em et~al.}\relax,
  }2007]{olfati-saber_consensus_2007}
Olfati-Saber, Reza, Fax, J.~Alex, \& Murray, Richard~M. (2007).
\newblock {Consensus and Cooperation in Networked Multi-Agent Systems}.
\newblock {\em Proceedings of the {IEEE}}, {\bf 95}(1), 215--233.

\bibitem[\protect\citename{Page {\em et~al.}\relax, }1999]{Page99thepagerank}
Page, Lawrence, Brin, Sergey, Motwani, Rajeev, \& Winograd, Terry. (1999).
\newblock {\em The pagerank citation ranking: Bringing order to the web}.

\bibitem[\protect\citename{Poole, }n.d.]{CongressData}
Poole, Keithe.
\newblock {\em {Voteview Website}}.
\newblock \url{http://voteview.com/house98.htm}.

\bibitem[\protect\citename{Porter {\em et~al.}\relax, }2009]{porter}
Porter, M.~A, Onnela, J.~P, \& Mucha, P.~J. (2009).
\newblock {Communities in networks}.
\newblock {\em Notices of the american mathematical society}, {\bf 56}(9),
  1082--1097.

\bibitem[\protect\citename{Rosvall \& Bergstrom, }2008]{Rosvall08}
Rosvall, Martin, \& Bergstrom, Carl~T. (2008).
\newblock {Maps of random walks on complex networks reveal community
  structure}.
\newblock {\em Proceedings of the national academy of sciences}, {\bf 105}(4),
  1118--1123.

\bibitem[\protect\citename{Shi \& Malik, }2000]{ShiMalik00}
Shi, J., \& Malik, J. (2000).
\newblock {Normalized Cuts and Image Segmentation}.
\newblock {\em Ieee transactions on pattern analysis and machine intelligence},
  {\bf 22}(8), 888--905.

\bibitem[\protect\citename{{Smith} {\em et~al.}\relax, }2013]{Smith13spectral}
{Smith}, L.~M., {Lerman}, K., {Garcia-Cardona}, C., {Percus}, A.~G., \&
  {Ghosh}, R. (2013).
\newblock {Spectral clustering with epidemic diffusion}.
\newblock {\em Physical review e}, {\bf 88}(4), 042813.

\bibitem[\protect\citename{Spielman \& Teng,
  }1996]{Spielman96spectralpartitioning}
Spielman, Daniel~A., \& Teng, Shang-Hua. (1996).
\newblock Spectral partitioning works: Planar graphs and finite element meshes.
\newblock {\em Pages  96--105 of:} {\em Ieee focs}.

\bibitem[\protect\citename{Spielman \& Teng, }2004]{SpielmanTengLinear}
Spielman, Daniel~A., \& Teng, Shang-Hua. (2004).
\newblock {Nearly-linear Time Algorithms for Graph Partitioning, Graph
  Sparsification, and Solving Linear Systems}.
\newblock {\em Pages  81--90 of:} {\em {STOC}}.
\newblock ACM.

\bibitem[\protect\citename{Spielman \& Teng, }2007]{SpielmanTengClustering}
Spielman, Daniel~A., \& Teng, Shang-Hua. (2007).
\newblock {Spectral partitioning works: Planar graphs and finite element
  meshes}.
\newblock {\em Linear algebra and its applications}, {\bf 421}(2-3), 284--305.

\bibitem[\protect\citename{Watts \& Strogatz, }1998]{watts1998collective}
Watts, Duncan~J, \& Strogatz, Steven~H. (1998).
\newblock {Collective dynamics of {\lq}small-world{\rq}networks}.
\newblock {\em nature}, {\bf 393}(6684), 440--442.

\bibitem[\protect\citename{Zachary, }1977]{ZacharyKarateClub}
Zachary, Wayne~W. (1977).
\newblock {An Information Flow Model for Conflict and Fission in Small Groups}.
\newblock {\em Journal of anthropological research}, {\bf 33}(4), 452--473.

\bibitem[\protect\citename{Zhu {\em et~al.}\relax, }2014]{Zhu2014}
Zhu, Yaojia, Yan, Xiaoran, \& Moore, Cristopher. (2014).
\newblock {Oriented and degree-generated block models: generating and inferring
  communities with inhomogeneous degree distributions}.
\newblock {\em Journal of complex networks}, {\bf 2}(1), 1--18.

\end{thebibliography}

\renewcommand{\thesection}{\Alph{section}}
\appendix
\section{Appendices}
In the appendix, we will mathematically analyze the convergence of the dynamics defined by the generalized Laplacian framework (Eqn. (\eqref{eq:spreading-operator})). The existence of a communities lead to bottlenecks which prevent dynamics from converging rapidly from the initial distribution. 
For mathematical convenience and better intuitions, we shall use different basis throughout the Appendix. As we have shown previously, the same theoretical result can be applied to any value of $\rho$ with a simple change of basis. For clarity we again abuse the notation $d_i$ and use it as ${d_{\WW}}_i$.

\subsection{Convergence of the dynamic process}
\label{sec:appendix}
We first examine the evolution of the  dynamics defined by a network $G = (V,E,\AA)$ and a dynamic operator $\LL$. Particularly, we are interested in estimating the rate that the dynamic process converges to its its own notion of stationary distribution from some initial distribution, and its relation to the generalized conductance $h_{\sbvec{\calL}}$ that we defined (Eqn. (\ref{eq:hS})). We will use the random walk formulation with $\LL=(\DD_{w}-\WW) (\DD_{w}\TT)^{-1}$ in this subsection.

For an starting vector $\bvec{\mu}$, let $\ttheta_{\bvec{\mu}}(t)[1]$ denote the value of the state variable of vertex $i\in V$ at time $t$ when the initial vector is $\bvec{\mu}$. In other words,  $\ttheta_{\bvec{\mu}}(t) = (\theta_{\bvec{\mu}}(t)[1],...,\theta_{\bvec{\mu}}(t)[n])$ is the
solution of
\begin{equation}\label{eqn:dynamics2}
\ttheta_{\bvec{\mu}}(t) =e^{- \LL t}\cdot \bvec{\mu}
= \sum_{k=0}^{\infty}\frac{(- t)^{k}}{k!}{\LL^{k}\bvec{\mu}},
\end{equation}

If the dynamic process converges when starting from $\bvec{\mu}$ and $\sum_i \mu[i] = 1$, regardless of $\bvec{\mu}$, the random walk always converges to
\begin{eqnarray*}
\lim_{t \rightarrow \infty}\ttheta_{\bvec{\mu}}(t) = \ppi = \left(\pi_1,...,\pi_n\right) = \frac{1}{{\sum_j {d}_j\tau_j}}\left({d}_1\tau_1,...,{d}_n\tau_n\right).
\end{eqnarray*}

We recall that the {\em volume} of a subset $S\subseteq V$ is $\vol_{\LL}(S)=\sum_{j \in S}{d}_j\tau_j$.
Let $\bvec{\mu}_i(S)=\frac{d_i\tau_i}{\vol_{\LL}(S)} $ for $i\in S$ and is $0$ otherwise.

Like in the traditional Laplacian-conductance framework, we will establish that the existence of a community $S\subset V$ with small  $h_{\LL}(S)$ underscores why the dynamics may not converge rapidly from the initial distribution $\bvec{\mu}$.

To this end, let $\boldsymbol{\Theta}_{\bvec{\mu}}(t,S)=\sum_{i\in S}\theta_{\bvec{\mu}}(t)[i]$ be the total probability density in the set $S$.
Let $\bvec{\chi}(S)$ be the indicator vector of $S$, such that ${\bvec{\chi}}_i(S) = 1$ for $i\in S$ and ${\bvec{\chi}}_i(S) = 0$ for $i\not\in S$.
In addition, we have $\bvec{\mu}(S) = (\DD_{w}\TT) \bvec{\chi}(S)/\vol_{\LL}(S)$.

In the analysis of this subsection, we will denote $\KK_{t}={e^{- \LL t}}$, and use the property that
$\KK_{t}=\KK_{t/2}\KK_{t/2}= \KK_{t/2}  (\DD_{w}\TT) \KK^T_{t/2} (\DD_{w}\TT)^{-1} $, where $\KK^T_t$ is the similar consensus formulation of the random walk $\KK_t$ at the same time $t$. By definition, we also have $\KK_t\LL = \LL\KK_t$.

We first establish a lemma showing that while the total probability $\boldsymbol{\Theta}_{\bvec{\mu}}(t,S)$ is propagating out of subset $S$ during every step of the dynamic process, the rate at which it happens (the derivative) is bounded by the generalized conductance $h_{\LL}(S)$ of $S$ given by Eqn.~(\ref{eq:hS}).

\begin{lemma}
\label{th:3}
For $S\subset V$  with $\vol_{\LL}(S)\leq \vol_{\LL}(V)/2$ we have
$\frac{\partial \boldsymbol{\Theta}_{\bvec{\mu}}(t,S)}{\partial t}  \leq 0$.
Moreover,
$$\left|\frac{\partial \boldsymbol{\Theta}_{\bvec{\mu}}(t,S)}{\partial t} \right|\leq h_{\LL}(S)\;.$$
\end{lemma}

\begin{proof}
By definition,
$\boldsymbol{\Theta}_{\bvec{\mu}}(t,S) =  \bvec{\chi}(S)^T \KK_t \bvec{\mu}(S)$. As
$d \KK_t/dt = -\LL \KK_t$, we have

\begin{eqnarray*}
\frac{\partial \boldsymbol{\Theta}_{\bvec{\mu}}(t,S)}{\partial t}  & =& - \bvec{\chi}(S)^T \LL \KK_t \bvec{\mu}(S) \\
& = & -\bvec{\chi}(S)^T  \KK_{t/2} (\DD_{w}-\WW) (\DD_{w}\TT)^{-1} \KK_{t/2} \bvec{\mu}(S) \\
& = & -\bvec{\chi}(S)^T  \KK_{t/2} (\DD_{w}-\WW) \KK^T_{t/2} \bvec{\chi}(S)/\vol_{\LL}(S) \\
& = & -\frac{1}{\vol_{\LL}(S)} \bvec{\chi}(S)^T \KK_{t/2} (\DD_{w}-\WW)\KK^T_{t/2} \bvec{\chi}(S) \\
&\le& 0\;,
\end{eqnarray*}
where the last inequality follows from $\DD_{w}-\WW$ is
positive-semi-definite. We can similarly show that
$\displaystyle \frac{\partial^2\boldsymbol{\Theta}_{\bvec{\mu}}(t,S)}{\partial t^2}  \geq 0$.
Thus,
\begin{eqnarray*}
\left|\frac{\partial \boldsymbol{\Theta}_{\bvec{\mu}}(t,S)}{\partial t}\right| &\leq&
-\left.\frac{\partial\boldsymbol{\Theta}_{\bvec{\mu}}(t,S)}{\partial t} \right| t=0  \\
& = & \frac{\bvec{\chi}(S)^T  (\DD_{w}-\WW) (\DD_{w}\TT)^{-1} (\DD_{w}\TT)\bvec{\chi}(S)}{\vol_{\LL}(S)}\\
& = & \frac{\cut(S,\bar{S})}{\min(\vol(S),\vol(\bar{S}))}\\
& = & h_{\LL}(S)\;.
\end{eqnarray*}
\end{proof}

While the dynamic process always converges to $\ppi$, regardless of $\bvec{\mu}$, Lemma \ref{th:3} effectively puts a limit to how fast this convergence could happen if the starting probability $\bvec{\mu}$ is specified in certian parts of the graph. If there exists a subset $S$ with small generalized conductance that contains the starting seeds, the dynamic process will converge slowly.

\subsection{Generalized Cheeger inequality using vertex state variables}
\label{sec:CheegerTheta}
Following the intuition of the proof of the extended Cheeger inequality (Theorem \ref{th:1}) and the traditional  Laplacian analysis  \cite{LovaszS,SpielmanTengClustering,Chung1997Spectral,ChungLocalHeatKernal}, we consider a sweeping process based on the vertex state variables instead of the eigenvectors. In this subsection, we will use the symmetric formulation and assume that $\LL= (\DD_{w}\TT)^{-1/2} (\DD_{w}-\WW) (\DD_{w}\TT)^{-1/2}$.

Suppose $\ttheta_{\bvec{\mu}}(t)$ is now the same vertex state vector with starting vector $\bvec{\mu}$ under the symmetric basis. We order the vertices of $G$ so that
\begin{equation}\label{eq:sweepOrder}
\frac{\theta_{\bvec{\mu}}(t)[w_1]}{\sqrt{d_{w_1}\tau_{w_1}}} \ge \frac{\theta_{\bvec{\mu}}(t)[w_2]}{\sqrt{d_{w_2}\tau_{w_2}}}  \ge \cdots \ge \frac{\theta_{\bvec{\mu}}(t)[w_n]}{\sqrt{d_{w_n}\tau_{w_n}}}.
\end{equation}
Let $S_{\bvec{\mu}}(t)[i] = \{w_1,\cdots, w_i\}$.
Let $h^*_{\bvec{\mu}}(t)
=\min_i \{h_{\LL}(S_{\bvec{\mu}}(t)[i]): i\in[1:n]\}$.
As the dynamic process evolves, both $\ttheta_{\bvec{\mu}}(t)$ and $h^*_{\bvec{\mu}}(t)$ change over time.

Here we will focus on a simpler initial vector in the form of $\bvec{\mu}^0=\frac{1}{\sqrt{d_u\tau_u}}\chi_u$, i.e. the dynamic process starts from a single vertex $u$. Notice that the total probability in random walk formulation has been scaled by $1/\sqrt{d_u\tau_u}$ under the symmetric basis. We will also simplify subscripts and use $\ttheta_{u}(t) = \ttheta_{\chi_u}(t)$ and $h^*_{u}(t) = h^*_{\chi_v}(t)$.

If the dynamic process converges when starting from $\chi_u$, we have under the symmetric basis,
\begin{eqnarray*}
\lim_{t \rightarrow \infty}\ttheta_{\bvec{\mu}^0}(t) = z_u\ppi = \frac{1}{\vol_{\LL}(V)}\left(\sqrt{{d}_1\tau_1},...,\sqrt{{d}_n\tau_n}\right)\;,
\end{eqnarray*}
where the stationary state is independent of the initial condition, and we can rewrite the constant $z_u = z$.

\begin{theorem}[Generalized Cheeger inequality using vertex state variables]
\label{le:4}
For a graph $G=(V,E,\WW)$, for $t\ge0$, for all $u,v\in V$,
 \begin{equation}
\left|\theta_{u}(t)[v]-\pi_{u}[v]\right|\le  \frac{1}{\sqrt{d_u\tau_u d_v\tau_v}} e^{- t\frac{\left({h^*_{u}}^2(t)+{h^*_{v}}^2(t)\right)}{4}}
\end{equation}
\end{theorem}

\begin{proof}
If we again let $\KK_{t}={e^{- \LL t}}$ under the symmetric basis,
we have $\KK_{t}-\KK_{\infty}=(\KK_{t/2}-\KK_{\infty})^2 $ and
\begin{eqnarray*}
& &\hspace*{-0.4in}\left|\theta_{u}(t)[v]-\pi_u[v]\right|^2
= \left|\bvec{\chi}_v^{T} ({\KK_t}-{\KK_\infty} ) \bvec{\chi}_u\right|^2  \\
&= &\left|\bvec{\chi}_v^{T} \left(\sum_w({ \KK_{t/2}}-{ \KK_{\infty}})\bvec{\chi}_w \bvec{\chi}^T_w({ \KK_{t/2}}-{ \KK_{\infty}})\right)\bvec{\chi}_u\right|^2 \\
&\le& \left|\left(\sum_w\left(\bvec{\chi}_u^T\left({\KK_{t/2}}-{\KK_{\infty}}\right)\bvec{\chi}_w\right)^2\right)\cdot
	    \left(\sum_w\left(\bvec{\chi}_v^T\left({\KK_{t/2}}-{\KK_{\infty}}\right)\bvec{\chi}_w\right)^2\right)\right| \\
&&\mbox{by Cauchy-Schwartz}  \\
&\le&\left|\left(\bvec{\chi}_u^T\left({ \KK_{t/2}}-{ \KK_{\infty}}\right)\left({ \KK_{t/2}}-{\KK_{\infty}}\right)\bvec{\chi}_u\right)
     \cdot \left(\bvec{\chi}_v^T\left({ \KK_{t/2}}-{ \KK_{\infty}}\right)\left({ \KK_{t/2}}-{\KK_{\infty}}\right)\bvec{\chi}_v\right)\right| \\
& =  &\left|{\theta_{u}(t)[u]}-\pi[u]\right|\cdot \left|{\theta_{v}(t)[v]}-\pi[v]\right|
\end{eqnarray*}
Now it is enough to show that
$\left|\theta_{u}(t)[u]-\pi_u[u]\right|\le  \frac{1}{\sqrt{d_u\tau_u}} e^{- t\cdot\frac{h^*_{u}(t)}{2}}$.
\begin{eqnarray*}
\frac{\partial \theta_{u}(t)[u]}{\partial t}
&=& - \bvec{\chi}_u^T{\KK_{t/2}(\DD_{w}\TT)^{-1/2} (\DD_{w}-\WW) (\DD_{w}\TT)^{-1/2}\KK_{t/2}\bvec{\chi}_{u}}\nonumber \\
&=& - \sum_{i\leq j}w_{ij}\left(\frac{\theta_{u}(t/2)[i]}{\sqrt{d_i\tau_i}} -\frac{\theta_{u}(t/2)[j]}{\sqrt{d_j\tau_j}} \right)^2\;,
\end{eqnarray*}
where we followed the similar derivations in Lemma~\ref{th:3}, and used the fact that for any vector $f$,
$$
f^T(\DD_{w}-\WW)f = \sum_{i\leq j} w_{ij}(f[i]-f[j])^2\;.
$$

Also, because for all $w \in V$, $\pi_u[w]  = z\sqrt{{d}_w\tau_w} $, we have
\begin{eqnarray*}
\theta_{u}(t)[u]-\pi_u[u] = {\sum_w (\theta_{u}(t/2)[w]-\pi_u[w])^2} =
{\sum_w {d}_w\tau_w \left(\frac{\theta_{u}(t/2)(w)}{\sqrt{d_w\tau_w}}-z\right)^2}
\end{eqnarray*}
Therefore,
\begin{eqnarray*}
\frac{\frac{\partial\theta_{u}(t)[u]}{\partial t}}{\theta_{u}(t)[u]-\pi_u[u]}
&=&-\frac{\sum_{i\leq j}w_{ij}\left(\frac{\theta_{u}(t/2)[i]}{\sqrt{d_i\tau_i}} -\frac{\theta_{u}(t/2)[j]}{\sqrt{d_j\tau_j}} \right)^2}
      {\sum_w {d}_w\tau_w \left(\frac{\theta_{u}(t/2)(w)}{\sqrt{d_w\tau_w}}-z\right)^2} \nonumber \\
&=&-\frac{\sum_{i\leq j}w_{ij}\left(\left(\frac{\theta_{u}(t/2)[i]}{\sqrt{d_i\tau_i}}-z\right) - \left(\frac{\theta_{u}(t/2)[j]}{\sqrt{d_j\tau_j}}-z\right) \right)^2}
      {\sum_w {d}_w\tau_w \left(\frac{\theta_{u}(t/2)(w)}{\sqrt{d_w\tau_w}}-z\right)^2} \nonumber \\
& = & -\frac{\sum_{i\leq j}w_{ij}\left(g[i]-g[j]\right)^2} {\sum_w {d}_w\tau_w (g[w])^2} \nonumber \\
\end{eqnarray*}
if we make the substitution $g[v] = \frac{\theta_{u}(t/2)[v]-z\sqrt{d_v\tau_v}}{\sqrt{d_v\tau_v}}$.

Recall as we analyzed in Section~\ref{sec:centrality}, for all $t$, the projection of the conservative process $\ttheta_{u}(t)$ in the direction
of $\bvec{\pi}$ is always $z\bvec{\pi}$, we have $(\theta_{u}(t/2)-z\ppi) \perp \ppi$. In other words,
$\forall t, \sum_v(\theta_{u}(t/2)[v]-z\sqrt{d_v\tau_v}) \cdot \sqrt{d_v\tau_v}=0$, and
$$\sum_v g[v]d_v\tau_v = \sum_v (\theta_{u}(t/2)[v]-z\sqrt{d_v\tau_v})\sqrt{d_v\tau_v} = 0.$$
Therefore, the vector $g$ satisfies all the condition as it does in the proof of Theorem \ref{th:1}.
With the same argument, by sweeping the vertices according to the order $\frac{\theta_{\bvec{\mu}}(t)[w_1]}{\sqrt{d_1\tau_1}} \ge \frac{\theta_{\bvec{\mu}}(t)[w_2]}{\sqrt{d_2\tau_2}}  \ge \cdots \ge \frac{\theta_{\bvec{\mu}}(t)[w_n]}{\sqrt{d_n\tau_n}}$, we have the \emph{Generalized Cheeger inequality using vertex state variables}:
\begin{eqnarray*}
\frac{\frac{\partial \theta_{u}(t)[u]}{\partial t}}{\theta_{u}(t)[u]-\pi_u[u]} \leq -\frac{{h^*_{u}}^2(t)}{2}\;.
\end{eqnarray*}

This means
$$
\frac{\partial}{\partial t} \log(\theta_{u}(t)[u]-\pi_u[u]) \leq -\frac{{h^*_{u}}^2(t)}{2}\;,
$$
which leads to
$$\theta_{u}(t)[u] - \pi_u[u] \leq C_1+C_2  e^{-t\cdot\frac{{{h^*_u}^2(t)}}{2}}.$$
By considering the boundary conditions $\theta_{u}(0)[u]=1/\sqrt{d_u\tau_u}$ and $\lim_{t\rightarrow \infty} (\theta_{u}(t)[u]-\pi_u[u]) =0$, we set $C_1=0$ and $C_2=1/\sqrt{d_u\tau_u}$. Therefore,
$$\left|\theta_{u}(t)[v]-\pi_u[v]\right| \le \left|{\theta_{u}(t)[u]}-\pi[u]\right|^{1/2}\cdot \left|{\theta_{v}(t)[v]}-\pi[v]\right|^{1/2}
\le \frac{1}{\sqrt{d_u\tau_u d_v\tau_v}} e^{- t\frac{\left({h^*_{u}}^2(t)+{h^*_{v}}^2(t)\right)}{4}}$$
\end{proof}

Define ${h_S^*}^2(t) = \min_{u\in S} {h_u^*}^2(t)$, and let $\Pi_{u}(S)=\sum_{i\in S}\pi_{u}[i] \sqrt{d_i\tau_i}$ and $\boldsymbol{\Theta}_{u}(t,S)=\sum_{i\in  S}\theta_{u}(t)[i] \sqrt{d_i\tau_i}$ under the symmetric basis. We can then bound the rate of the convergence  of $\boldsymbol{\Theta}_{u}(t,S) - \Pi_{u}(S)$ using the following corollary.
\begin{corollary}
\label{co:2}
For any subset $S\subset V $ in a network $G = (V,E,\WW)$
  with $\vol_{\LL}(S) \leq
\vol_{\LL}(V)/2$, we have
\begin{eqnarray*}
\boldsymbol{\Theta}_{u}(t,S) - \Pi_{u}(S)
 \leq \frac{|S|}{\sqrt{d_u\tau_u}} e^{- t\frac{{h_S^*}^2(t)}{2}}\;,
\end{eqnarray*}
where $|S|$ represents the number of vertices in the set $S$.
\end{corollary}
\begin{proof}
\begin{eqnarray*}
\boldsymbol{\Theta}_{u}(t,S) - \Pi_{u}(S)
& = & \sum_{v\in S} \theta_{u}(t)[v] \sqrt{d_v\tau_v} - \sum_{v\in S}\pi_{u}[v] \sqrt{d_v\tau_v} \\
&\leq&\sum_{v\in S} \left|\theta_{u}(t)[v]-\pi_u[v]\right|\sqrt{d_v\tau_v}\\
&&\mbox{by Theorem 5}\\
&\le& \frac{1}{\sqrt{d_u\tau_u}}e^{- t \frac{{h_u^*}^2(t) +{h_v^*}^2(t) }{4}} \left(\sum_{v\in S} 1 \right)\\
&\le& \frac{|S|}{\sqrt{d_u\tau_u}}e^{- t \frac{{h_S^*}^2(t)}{2}}
\end{eqnarray*}
\end{proof}

\end{document}